\documentclass{elsarticle}
\pdfoutput=1 

\makeatletter
\def\ps@pprintTitle{%
 \let\@oddhead\@empty
 \let\@evenhead\@empty
 \def\@oddfoot{}%
 \let\@evenfoot\@oddfoot}
\makeatother

\biboptions{sort&compress}

\usepackage{etex}
\reserveinserts{28}
\usepackage{makeidx}  

\usepackage[T1]{fontenc}
\usepackage[utf8]{inputenc}

\usepackage{mathrsfs}
\usepackage{latexsym}
\usepackage{paralist}

\usepackage{amsfonts,amssymb,amstext,amsthm} 
\usepackage{thmtools}
\usepackage{mathtools}

\usepackage[usenames,dvipsnames,table]{xcolor}
\usepackage{graphicx}

\usepackage{wrapfig}
\usepackage{booktabs} 
\usepackage{array}
\usepackage{longtable}
\usepackage{float}

\usepackage{enumerate}
\usepackage{xspace,xpunctuate}
\usepackage{ellipsis}
\usepackage{url}

\usepackage[protrusion=true,expansion=true]{microtype} 

\usepackage[ruled,longend,vlined,linesnumbered]{algorithm2e}

\usepackage{todonotes}

\usepackage{stmaryrd}
\def\Iver#1{ \llbracket #1 \rrbracket }

\def\FullStop{{\kern-0.5ex\raise0.3ex\hbox{\normalsize.}}}

\theoremstyle{plain}
\newtheorem*{claim}{Claim}

\def\isom{\simeq}
\def\lexprod{\mathbin{ \,\bullet\, }}
\def\cjoin{ \mathbin{ \,\ast\, } }
\newcommand\overbar[2][3]{{}\mkern#1mu\overline{\mkern-#1mu#2}}
\def\comp#1{\overbar #1} 

\newcommand{\N}{\ensuremath{\mathbf{N}}\xspace} 
\newcommand{\R}{\ensuremath{\mathbf{R}}\xspace} 
\renewcommand{\cal}{\mathcal}

\newcommand{\mc}{\mathcal}
\newcommand{\G}{\mathcal{G}}

\def\grad_#1{\nabla\!_#1}
\def\topgrad_#1{\widetilde \nabla\!_{#1}}

\newcommand{\bound}{\ensuremath{\operatorname{bd}}}    
\renewcommand{\leq}{\leqslant}
\renewcommand{\geq}{\geqslant}

\newcommand{\degavg}{\ensuremath{d_{avg}}}
\newcommand{\degseq}{D}

\newcommand{\topnab}{\mathop{\widetilde \triangledown}}
\newcommand{\dist}{\ensuremath{\text{dist}}}

\DeclareMathOperator*{\E}{\mathbb{E}}
\newcommand{\ErdosRenyi}{{Erd\H{o}s--R\'{e}nyi}\xspace}

\renewcommand{\th}{\ensuremath{^{\text{th}}}}

\renewcommand{\epsilon}{\varepsilon}
\renewcommand{\emptyset}{\varnothing}
\newcommand*\xbar[1]{%
  \hbox{%
    \vbox{%
      \hrule height 0.5pt 
      \kern0.5ex
      \hbox{%
        \kern-0.1em
        \ensuremath{#1}%
        \kern-0.1em
      }%
    }%
  }%
}

\renewcommand{\P}{\operatorname{P}}
\renewcommand{\E}{\operatorname{E}}
\newcommand{\M}{\operatorname{M}}
\newcommand{\Var}{\operatorname{Var}}

\def\any{\mathord{\color{black!33}\bullet}}%

\def\rndmodelnotation#1{G^{\text{#1}}}

\def\confmodel{ \rndmodelnotation{CF} }
\def\chunglu{ \rndmodelnotation{CL} }

\def\rndmodel{ \rndmodelnotation{R} }
\def\kleinberg{ \rndmodelnotation{KL} }

\newlength{\convarrowwidth}
\def\converges{\xrightarrow{  \mathmakebox[\convarrowwidth]{}  }}
\def\convd{\xrightarrow{  \mathmakebox[\convarrowwidth]{d}  }}

\def\whp{w.h.p\xperiod}
\def\aas{a.a.s\xperiod}

\newcommand*{\eg}{e.g\xperiod}
\newcommand*{\ie}{i.e\xperiod}
\newcommand*{\cf}{cf\xperiod}

\newcommand*{\etal}{et~al\xperiod}
\makeatletter
\newcommand*{\etc}{%
    \@ifnextchar{.}%
        {etc}%
        {etc.\@\xspace}%
}
\makeatother

\usepackage{xfrac} 
\def\half{{\sfrac{1}{2}}} 
\def\twothird{{\sfrac{2}{3}}} 

\def\fthresh{f_\text{thresh}}
\def\fdeg{f_{\deg}}
\def\fnabla{f_{\topgrad_{}}}
\def\fH{f_{H}}

\def\Deltak{\Delta^{\mathrlap{\!k}}}
\def\Deltar{\Delta^{\mathrlap{\!r}}}

\def\Nesetril{Ne\v{s}et\v{r}il\xspace}
\def\Dvorak{Dvo\v{r}\'{a}k\xspace}

\def\Erdos{Erd\H{o}s\xspace}
\def\Renyi{R\'{e}nyi\xspace}
\def\Barabasi{Barab\'{a}si\xspace}
\def\Bollobas{Bollob\'{a}s\xspace}

\usepackage{graphicx}
\graphicspath{{figures/}}

\def\paragraph#1{\par\textbf{#1} \ignorespaces}

\usepackage{url}

\newtheorem{definition}{Definition}
\newtheorem{theorem}{Theorem}
\newtheorem{lemma}{Lemma}
\newtheorem{proposition}{Proposition}
\newtheorem{corollary}{Corollary}
\newtheorem{observation}{Observation}

\def\plog#1{\log^{\kern-.1pt{\scriptscriptstyle\Theta(1)}}\kern-1pt(#1)}

\usepackage{verbatim}

\usepackage[perpage]{footmisc}

\begin{document}

\begin{frontmatter}
  \title{Structural Sparsity of Complex Networks: Bounded Expansion in Random Models and Real-World Graphs}

  \cortext[cor]{Corresponding author} 
  \fntext[moved]{Present address: DCSIS, Birkbeck University of London, London, UK.}

  \author[MIT]{Erik D. Demaine}
  \ead{edemaine@mit.edu}

  \author[RWTH]{Felix Reidl\corref{cor}\fnref{moved}}
  \ead{f.reidl@dcs.bbk.ac.uk}

  \author[RWTH]{Peter Rossmanith}
  \ead{rossmani@cs.rwth-aachen.de}

  \author[RWTH]{Fernando S\'{a}nchez Villaamil}
  \ead{fernando.sanchez@cs.rwth-aachen.de}

  \author[RWTH]{Somnath Sikdar}
  \ead{sikdar@cs.rwth-aachen.de}

  \author[NCSU]{Blair D. Sullivan}
  \ead{blair\_sullivan@ncsu.edu}

  \address[MIT]{Computer Science and Artificial Intelligence Laboratory,
    Massachusetts Institute of Technology, Cambridge, MA.}

  \address[RWTH]{Theoretical Computer Science, Dept. of Computer Science,
    RWTH Aachen University, Aachen, Germany.}

  \address[NCSU]{Department of Computer Science, North Carolina State University, Raleigh, NC.}

  \begin{abstract}
    This research establishes that many real-world networks exhibit
    \emph{bounded expansion}%
    \footnote[2]{Not to be confused with the notion of
    expansion related to expander graphs.}, a strong notion of structural
    sparsity, and demonstrates that it can be leveraged to design efficient
    algorithms for network analysis.

    We analyze several common network models regarding their
    structural sparsity. We show that, with high probability,
    (1)~graphs sampled with a prescribed sparse degree sequence;
    (2)~perturbed bounded-degree graphs;
    (3)~stochastic block models with small probabilities;
    result in graphs of bounded expansion. In contrast,
    we show that the Kleinberg and the Barab\'asi--Albert model
    have unbounded expansion.
    We support our findings with empirical measurements
    on a corpus of real-world networks.
  \end{abstract}

\end{frontmatter}



%
%
%
\section{Introduction}\label{sec:Introduction}
\paragraph{Complex networks vs.\ structural graph algorithms.}
Social networks (such as Facebook or physical disease propagation networks),
biological networks (such as gene interactions or brain networks), and
informatics networks (such as autonomous systems) are all examples of
\emph{complex networks}, which have been the attention of much study in recent
years, given the surge of available network data. Viewed as graphs, complex
networks seem to share several structural properties.  Perhaps most famous is
the \emph{small-world} property (``six degrees of separation''): typical
distances between vertex pairs are small compared to the size of the network.
Another important property is that their degree distribution tends to be
heavy-tailed, \ie~not exponentially bounded.
In many cases, this degree distribution is close to a \emph{power-law}: the
fraction of vertices of degree~$k$ is proportional to $k^{-\gamma}$, for some
constant~$\gamma$ typically between $2$ and~$3$
(recent work~\cite{clauset2009power,clauset2018power} has shown that this relationship is only approximate;
instead these distributions typically are power-law with an exponential cut-off, making
our results -- which require polynomial tail-bounds -- applicable. See Section~\ref{sec:models} for discussion). Networks furthermore often exhibit
\emph{high clustering} and admit a natural division into a \emph{community
structure}. Despite the clustering property, complex networks are
\emph{sparse}: the ratio of edges to vertices (edge density) is usually small.

On the other hand, the field of structural graph algorithms has led to
impressively efficient and precise algorithms (efficient PTASs, subexponential
fixed-parameter algorithms, linear kernelizations, etc.)\ for increasingly
general families of graphs; see, e.g.,
\cite{ContractionMinorFree_STOC2011,OddMinorFree_SODA2010,BidimensionalSurvey_CJ,
Fomin:2010:BK:1873601.1873644,HMinorFree_JACM, KLP13,GHO13}.
Many such results proved initially for planar graphs have since been extended
to bounded-genus graphs, graphs of bounded local treewidth, and
graphs excluding a fixed minor; yet such results are known to be impossible
for general graphs.
Can we apply these powerful algorithms to analyze complex networks?

We propose the following litmus test for whether a type of network
sparsity is ``useful'': Does it enable efficient algorithms for a
broad set of NP-hard problems? Unfortunately, the above-mentioned
structural properties of complex networks seem too weak to enable
better algorithms, while the discussed graph classes seem too
restrictive to apply to complex networks.  The goal of this paper is
to bridge this gap, by identifying a more general graph class that
simultaneously enables better algorithms and includes many complex
networks.

\paragraph{Bounded expansion.}
In general, complex networks seem to exhibit an intermediate-scale structure composed
of small dense parts---representing clusters or com\-munities---that are sparsely
interconnected. This hierarchical behavior has
been established for many networks~\cite{LLDM08_communities_CONF} and is consistent with
the tree-like intermediate structure observed in~\cite{AMS13}.

How can this notion be captured formally?
If we contract disjoint small-diameter subgraphs (representing
potentially dense local substructures in the network), then the
resulting graph (representing the global connectivity of these substructures)
should be sparse.
This gives rise to the notion of an \emph{$r$-shallow minor}, where $r$ is the maximum
diameter of the subgraphs that were contracted in the construction process.
(For formal definitions, see Section~\ref{sec:Preliminaries}.)
We cannot expect the edge density of all $r$-shallow minors to be constant
(then $r$ would play no role), but we require it to grow as any function of~$r$,
and thus be independent of the size of the graph.
A graph class for which this property holds has \emph{bounded expansion},
a concept introduced by \Nesetril and Ossona de Mendez~\cite{NOdM08}.

\paragraph{Theoretical results.}
Since the definition of bounded expansion applies to graph classes instead of
individual graphs, it is impossible to settle this question empirically.  To
ground our hypothesis in theory, we analyze several random graph models which
were designed to mimic the properties of specific types of real world
networks. Although not perfect, random graph models play a central role in
network analysis, both to guide our understanding of complex networks and as
a convenient source of synthetic data for algorithm testing and validation.
In our case, random graph models allow us to determine whether (synthetic)
complex networks have bounded expansion with high probability. We analyze
several popular random graph models:
\begin{enumerate}[(i)]
  \item the configuration model~\cite{MR95a} and the Chung--Lu
        model~\cite{chung2002average,chung2002connected}
        with specified asymptotic degree sequences, which includes
        graphs with heavy-tailed degree distributions;
  \item a variant of the configuration model which achieves high clustering~\cite{BST09};
  \item a significant generalization of \ErdosRenyi graphs we
        call \emph{perturbed bounded-degree graphs}
       (allowing the
        network to be built on  top of a fixed or random
        base graph of bounded degree),
        which includes the stochastic block model with small probabilities;
  \item the Kleinberg model~\cite{kleinberg2000navigation,Kle00};
  \item and the \Barabasi--Albert model~\cite{BA99,barabasi-albert}.
\end{enumerate}
We will show that the configuration model, the Chung--Lu model and perturbed
bounded-degree graphs have bounded expansion, while the
Kleinberg model and the \Barabasi--Albert model do not (actually they are not
even \emph{nowhere dense}, a strict generalization of bounded expansion).

\paragraph{Experimental results.}
We present an experimental study suggesting that important real-world
networks have small \emph{grad}, which is the density measure for
single graphs that defines the expansion for graph classes.
Interestingly, the algorithmic tools that become efficient when the
grad of a graph  is small can be directly applied without knowing
its actual value.

We will make extensive use of
\emph{$p$-treedepth colorings}: for any integer~$p$,
a graph can be colored with~$f(p)$ colors such that any set of at most~$p-1$
colors induces a graph of treedepth at most~$p-1$, where~$f$ only depends on the grad of
the graph. Generally the running time of algorithms based on $p$-treedepth
colorings depends heavily on the number of colors $f(p)$. In
Section~\ref{sec:experiments} we present experimental results obtained by
computing and evaluating $p$-treedepth colorings with a simple algorithm in a
variety of real-world networks.

Our results show that, in general, real networks exhibit even better
structure (require fewer colors for a $p$-centered coloring) than
randomly generated graphs with the same degree distribution via the
configuration model.  These results support our hypothesis that
``community structure'' (not captured by the degree distribution)
further increases the algorithmic tractability.

\paragraph{Algorithmic results.}
With both theoretical and empirical results in hand, we exploit the
aforementioned tools for graphs with small grad in
Section~\ref{sec:AlgoImplications} to solve typical problems for complex
networks: First we develop a faster algorithm than the one presented in
\cite{NOdM12} to count subgraph homomorphisms based on $p$-treedepth
colorings. Counting subgraphs is fundamental to \emph{motif counting}, a
widely used concept to analyze networks. Then we develop an algorithm which
computes localized versions of several centrality measures in
linear time on graphs of bounded expansion. Specifically, we
present:

\begin{enumerate}[(i)]
\item A linear-time algorithm to count the number of times a fixed
  subgraph appears as an (induced) subgraph/homomorphism in
  graphs of bounded grad. We do this by improving the previously best
  known algorithm to count the appearances of a structure of size $h$
  on graphs of treedepth $t$ from $O(2^{ht} ht\cdot n)$ to $O(t^{h}
  6^{h} h^2 \cdot n)$, thus removing the exponential dependency on
  $t$, while keeping the algorithm simple and avoiding big hidden
  constants.
\item A linear-time algorithm to compute localized variants (i.e.,
  computed in a constant-radius neighborhood around each vertex) of
  the \emph{closeness centrality} and two other related
  measures.  The constant in the running time depends on the radius
  and the grad of the graph.
\end{enumerate}

\noindent
For the second algorithm we provide experimental results which
indicate that the local variants of these centrality measures can be
used to estimate the top 10 percent of the most central nodes.

\paragraph{Previous results.}
There is substantial empirical work studying structural properties of
complex networks, so we focus here on work closest to structural graph theory.
In general, large real-world complex networks are not easily classified as either
low- or high-dimensional.
In particular, data-mining tools which implicitly assume low dimensionality (such
as singular value decomposition) produce models and results
incompatible with observed structure and dynamics, yet traditional
high-dimensional tools (like sampling) often fail to achieve measure
concentration due to the extreme sparsity of the networks.  Adcock
\etal~\cite{AMS13} recently empirically established that, when
compared with a suite of idealized graphs\footnote{Representing
  low-dimensional structures, common hierarchical models,
  constant-degree expanders, etc.}, realistic large social and
informatics networks exhibit meaningful ``tree-like'' structure at an
intermediate scale. Their work related the $k$-core structure (whose
extremal statistic is the degeneracy) to the networks' Gromov
hyperbolicity and tree decompositions.  Unfortunately, they showed
that straightforward applications of these measures are often
insufficient to reveal meaningful structure because of noisy/random
edges in the network which contradict the strict structural
requirements. (For example, several families of popular random graph
models have been shown to have very large treewidth~\cite{Gao12}.)

Some simple preliminary observations about networks and bounded expansion were
made in \cite{gago}, such as the \emph{linear growth model} not having bounded
expansion since it was known that it contains growing bi-cliques, and they
conjecture that \Barabasi-Albert is  somewhere dense \aas. Here we prove that
it is at least somewhere dense with non-vanishing probability. For the
\emph{random intersection graph} model, which is used to model real world
networks where connections depend on shared attributes, it was shown
in~\cite{farrell2014hyperbolicity} that it has bounded expansion exactly when
it is degenerate.

%
%
%
\section{Preliminaries}\label{sec:Preliminaries}
\noindent
For a natural number~$n$ we use the notation~$[n] := \{1,\ldots,n\}$.
For a graph~$G$, we denote by~$\Delta(G)$ its maximal degree and by
$\omega(G)$ its clique number, \eg the largest integer~$t$ such that
the complete subgraph~$K_t$ is contained in~$G$.
We will make use of the following graph operations. For graphs~$G_1, G_2$,
the \emph{complete join} $G_1 \cjoin G_2$ is the graph obtained by
first taking the disjoint union of $G_1,G_2$ and then connecting every
vertex of~$V(G_1)$ to every vertex~$V(G_2)$. For example, $G \cjoin \comp K_2$
is the graph obtained from~$G$ by adding two universal vertices.
The \emph{lexicographic product} $G_1 \lexprod G_2$ is the graph with
vertices $V(G_1) \times V(G_2)$ and edges
\[
    (u,x)(v,y) \in E(G_1 \lexprod G_2)
      \iff
    uv \in E(G_1) ~\text{or}~ (u = v ~\text{and}~ xy \in E(G_2)).
\]
We use the notation $H \isom G$ to denote that the graphs~$H,G$ are
isomorphic. In the following we sometimes employ the notation~$\bar k^2 := {k \choose 2}$,
in particular when the expression appears as an exponent.

%
%
\subsection{Bounded expansion classes}

\noindent
We will use \emph{$(\leq r)$-subdivisions} to formalize the notion of shallow
topological minors. A $(\leq r)$-subdivision of a graph $H$ is any graph which
can be obtained from $H$ by replacing edges with disjoint paths of length at most~$r+1$.

\begin{definition}[Shallow topological minor~\cite{NOdM08}]\label{def:shallowtopminor+}
  A graph $H$ is an \emph{$r$-shal\-low topological minor} of~$G$ if a
  $(\leq 2r)$-subdivision of $H$ is isomorphic to a subgraph $G'$ of~$G$.
  We call $G'$ a \emph{model of $H$ in $G$}. For simplicity, we assume
  by default that $V(M) \subseteq V(G')$ such that the isomorphism
  between $H$ and $G'$ is the identity when restricted to $V(M)$. The vertices $V(M)$
  are called \emph{nails} and the vertices $V(G') \setminus V(M)$ \emph{subdivision vertices}.
  The set of all $r$-shallow topological minors of a graph $G$ is denoted by $G \topnab r$.
\end{definition}

\noindent
The theory of bounded expansion classes crucially relies on a `parameterization'
of shallow minors. The \emph{grad} of a graph is one such parameterization:

\begin{definition}[Topological grad]
  For a graph $G$ and an integer $r \geq 0$, the
  \emph{topological greatest reduced average density (topological grad) at depth $r$} is defined as
  $$
    \topgrad_r(G) = \max_{H \in G \topnab r} \frac{|E(H)|}{|V(H)|}.
  $$
  For a graph class $\mc G$, define $\topgrad_r(\mc G) = \sup_{G \in \mc G} \topgrad_r(G)$.
\end{definition}

\noindent
Given the notion of shallow topological minors and topological grad, we can
now define what it means for a class to have bounded expansion.
\begin{definition}[Bounded expansion]
  A graph class $\mc G$ has \emph{bounded expansion} if and
  only if there exists a function $f$ such that for all integers~$r \geq 0$, it holds that
  $\topgrad_r(\mc G) < f(r)$.
\end{definition}

\noindent
When introduced in~\cite{NOdM08}, bounded expansion was originally defined
using a characterization based on the notion of \emph{shallow minors}: $H$ is
an $r$-shallow minor of $G$ if $H$ can be obtained from $G$ by contracting
disjoint $r$-balls and then taking a subgraph. Taking the maximum over the
density of all $r$-shallow minors then defines the \emph{grad} of graph~$\grad_r(G)$.
An $r$-ball in a graph $G$ is a
subgraph $G' \subseteq G$ with the property that there exists~$v \in V(G')$
such that for all $u \in V(G')$, $d_{G'}(u, v) \leq r$. Both characterizations
are equivalent: for any graph~$G$ and integer~$r$, we have the relation (\cf~\cite{NOdM12})
\[
  \topgrad_r(G) \leq \grad_r(G) \leq 4 (4 \topgrad_r(G))^{(r+1)^2}
\]
thus it follows that a graph class has bounded topological grad if and only if
it has bounded grad, \eg if either is bounded the class has bounded expansion.
We will in the following exclusively use the
topological grad and will therefore often drop the term `topological'.

We note that graphs excluding a topological minor---in
particular planar graphs and bounded-degree graphs---have bounded
expansion. This generalizes to graphs excluding a minor
(and thus to those of bounded treewidth). Finally, we point out that
bounded expansion implies bounded degeneracy, where a graph
$G$ is $d$-degenerate if any subgraph of $G$ contains a node of degree
smaller than $d$. The converse does not hold.

The following alternative characterization of bounded expansion uses a special
coloring number with nice algorithmic properties.
\begin{definition}[$p$-centered coloring~\cite{NOdM08}]\label{def:pcenteredcoloring}
  Given a graph~$G$, let $c \colon V(G) \to [r]$ be a vertex coloring of~$G$
  with $r$ colors. We say that the coloring~$c$ is \emph{$p$-centered},
  for $p \geq 2$, if any connected subgraph of $G$ either receives at least~$p$ colors
  or contains some color exactly once. Define $\chi_p(G)$ to be the minimum
  number of colors needed for a $(p+1)$-centered coloring.
\end{definition}

\noindent
While this definition looks rather cryptic, it is easy to see that every graph has
a $p$-centered coloring for any $p$: simply assign a distinct
color to each vertex of the graph. Note that $p$-centered colorings are proper colorings
for $p \geq 2$ and in particular, $\chi_1$ is precisely the chromatic
number. Typically, the number of colors $q$ is much larger than~$p$ and one is
interested in minimizing~$q$. \Nesetril and Ossona de Mendez~\cite{NOdM08} showed that
for every graph~$G$,
\[
  \topgrad_r(G) \leq {\chi_{2r+2}(G) \choose 2r + 2}
  \quad\text{and}\quad
  \chi_r(G)  \leq P(\topgrad_{2^{p-2}+1}(G))
\]
where~$P$ is some polynomial of degree around~$2^{2^r}$. Accordingly,
for a graph class~$\mathcal G$, the quantity~$\chi_r(\mathcal G)$ is bounded
by a function of~$r$ if and only if the quantity~$\topgrad_r(\mathcal G)$ is.
and bounded expansion can equivalently be defined in terms of
$p$-centered colorings.

The following structural property, which follows directly from the
equivalence between centered colorings and \emph{treedepth}, make them
an attractive tool for algorithm design.

\begin{proposition}[$p$-treedepth colorings~\cite{NOdM08a}]\label{theorem:low-td-coloring}
  Let $\mc G$ be a graph class of bounded expansion. There exists a
  function $f$ such that for every $G \in \mc G$, $p \in \N$, the
  graph $G$ can be colored with $f(p)$ colors so that any $i < p$
  color classes induce a graph of treedepth $\leq i$ in $G$. This
  coloring can be computed in linear time.
\end{proposition}

\noindent
For those unfamiliar with this notion, the \emph{treedepth} of a
graph~$G$ is the lowest depth of a rooted forest whose closure contains~$G$ as a
subgraph. A \emph{treedepth decomposition} of $G$ is simply a forest
with vertex set $V(G)$ witnessing this fact. To put this width measure
into perspective: a graph of treedepth at most~$t$ cannot contain a
path of length $2^t$ \emph{and} has pathwidth at most $t-1$. Importantly,
a graph has treedepth~$t$ if and only if it admits a centered coloring with
$t$ colors. In one direction, imagine coloring the vertices of a graph according to
their depth in the treedepth decomposition; then every connected subgraph
receives a center (since it is connected, there is a single highest vertex
in the decomposition, which thus receives a unique color).
In the other direction, given a centered coloring with $t$ colors, we can
construct a treedepth decomposition of depth $t$ by using the centers of its
connected components as roots for the decomposition, removing these vertices from the
graph and recursing on the resulting connected components. Thus the terms
`$p$-centered colorings' and `$p$-treedepth colorings' are synonymous.

An example of a $5$-centered coloring and treedepth-decompositions of a
selected subset of colors can be found in Figure~\ref{fig:td-coloring}
on page \pageref{fig:td-coloring}. For more information about treedepth
see e.g.~\cite{treedepthfpt}.

\begin{figure}[hbtp]
  \centering
    \hspace*{-.2em}\includegraphics[width=1.01\textwidth]{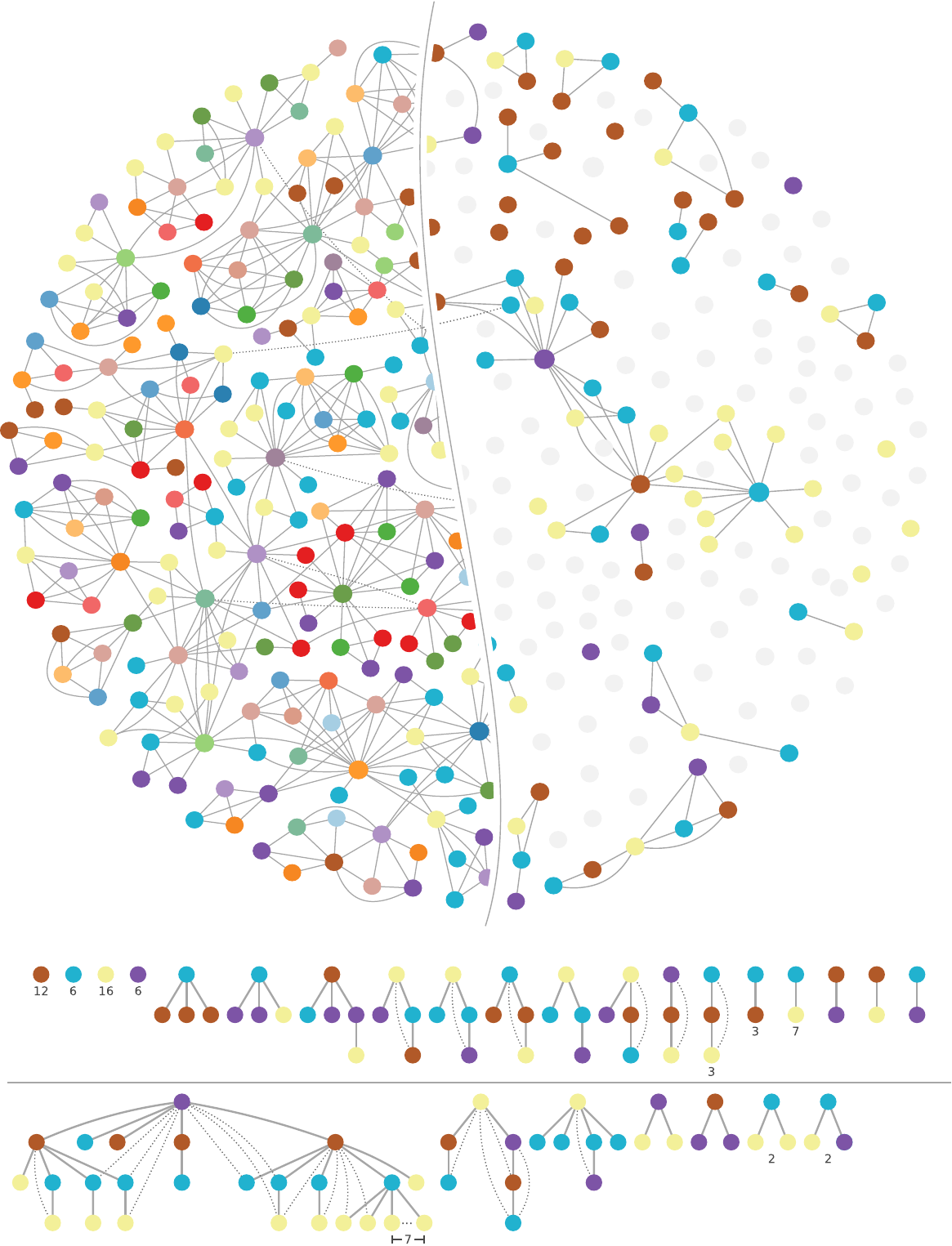}
    \caption{\footnotesize
    A $5$-centered coloring (using $21$ colors) of a
      real-world social network, Newman's Network
      Science~\cite{newman2006finding} (giant component shown), which
      represents co-authorships between researchers in the field of
      network science as of 2006. The right half is restricted to
      a a subgraph formed by~$4$ color classes. Below, the corresponding
      representation by trees of
      depth~$\leq 4$ (with multiplicities
      noted).\looseness-1} \label{fig:td-coloring}
\end{figure}

\Nesetril and Ossona de Mendez show that graph classes of bounded expansion
are precisely those for which there exists a function $f$ such that every
member~$G$ of the graph class satisfies $\chi_p(G) \leq f(p)$ (see Theorem~7.1
in~\cite{NOdM08}). In~\cite{NOdM08a}, the authors also showed how to obtain a
$p$-centered coloring with at most $P(f(p))$ colors for each fixed~$p$ in
linear time, where~$P$ is some polynomial of degree roughly~$2^{2^p}$. We will
make use of this algorithm in Sections~\ref{sec:experiments}
and~\ref{sec:AlgoImplications} and see that the actual number of colors is
manageable.

When working with random graphs, we will make heavy use of the following
alternative characterization of graphs of bounded expansion:
\begin{proposition}[\Nesetril, Ossona de Mendez, Wood{\rm ~\cite{NOdMW12,NOdM12}}]\label{prop:BoundedExpChar}
  A class~$\cal G$ of graphs has bounded expansion if and only if there exist
  real-valued functions $\fthresh, \fdeg, \fnabla, \fH \colon \R^{+} \rightarrow \R$ such
  that the following two conditions hold:
  \begin{enumerate}[(i)]
    \item For all $\epsilon > 0$ and for all $G \in \cal G$ with $|G| > \fthresh(\epsilon)$,
      it holds that
    $$
      \frac{1}{|V(G)|} \cdot |\{v \in V(G) \colon \deg(v) \geq \fdeg(\epsilon)\}| \leq \epsilon.
    $$
    \item For all $r \in \N$ and for all $H \subseteq G \in \cal G$ with $\topgrad_r(H) > \fnabla(r)$,
      it follows that
    $$
     |V(H)| \geq \fH(r) \cdot |V(G)|.
    $$
  \end{enumerate}
\end{proposition}

\noindent
Intuitively, Proposition~\ref{prop:BoundedExpChar} characterizes classes of graphs with bounded expansion as those where:
\begin{enumerate}[(i)]
  \item all sufficiently large members of the class have a small fraction
    of vertices of large degree;
  \item all subgraphs of $G \in \cal G$ whose shallow topological minors are
    sufficiently dense must necessarily span a large fraction of the vertices of~$G$.
\end{enumerate}

\noindent
Finding the first pair of functions~$\fthresh$,$\fdeg$ is usually
straightforward, the real challenge lies in proving that functions
$\fnabla$,$\fH$ exist. Consider the contra-positive of the second
condition: in terms of random graph models, we want to show that
subgraphs that span at most a $\fH(r)$-fraction of the vertices
of~$G$ have their (topological) grad bounded by~$\fnabla$ with high
probability.

One of the useful properties of bounded expansion classes is that
their expansion does not increase arbitrarily through the addition of a
universal vertex.

\begin{lemma}\label{lemma:grad-apex}
  For every graph~$G$ it holds that
  $$
    \topgrad_r(G) \leq \topgrad_r(G \cjoin K_1) < \topgrad_r(G) + 1
  $$
\end{lemma}
\begin{proof}
  It is easy to see that if $H$~is an $r$-shallow topological minor of~$G$, then
  $H \cjoin K_1$ is a $r$-shallow topological minor of~$G \cjoin K_1$.
  The density of $H \cjoin K_1$ is then given by
  \begin{align*}
    \frac{|E(H \cjoin K_1)|}{|V(H \cjoin K_1)|} = \frac{|E(H)| + |V(H)|}{|V(H)|+1}
  \end{align*}
  And therefore we obtain upper and lower bounds via
  \begin{align*}
      \frac{1}{1+(1/|V(H)|)} \left( \frac{|E(H)|}{|V(H)|} + 1 \right)
        \leq \frac{|E(H \cjoin K_1)|}{|V(H \cjoin K_1)|}
        < \frac{|E(H)|}{|V(H)|} + 1
  \end{align*}
  Observing that $|V(H)| \geq \topgrad_r(G)$ proves the claim.
\end{proof}

\noindent
Therefore even a constant number of universal vertices will not
influence the grad too much. In terms of graph classes this means
that if~$\mc G$ has bounded expansion, then the class
$\{G \cjoin K_c\}_{G \in \mc G}$ for any constant~$c$ also has bounded expansion,
albeit by a different function.

Similarly, the lexicographic product with a constant-sized graph does not
change the expansion of a class arbitrarily.
\begin{proposition}[\Nesetril and Ossona de Mendez~{\cite{NOdM12}}]\label{prop:grad-lexprod}\mbox{}\\
  For every graph~$G$, integer~$p \geq 2$ and half-integer~$r$ it
  is true that
  \[
    \topgrad_r(G \lexprod K_p) \leq \max\{2r(p{-}1){+}1,p^2\} \cdot \topgrad_r(G) + p - 1
  \]
\end{proposition}

\noindent
In conclusion, bounded expansion is a robust description of structural sparseness: it
is closed under taking shallow minors and the above two `algebraic' operations. However,
bounded expansion classes are not the ultimate definition of sparseness. In order to obtain
a dichotonomical view of structural sparseness, \Nesetril and Ossona de Mendez defined
\emph{nowhere dense} classes. Their definition mirrors that of bounded expansion classes,
only that we now measure the clique number of shallow minors instead of their density.

\begin{definition}[Nowhere dense]
  A graph class $\mc G$ is \emph{nowhere dense} if and
  only if there exists a function $f$ such that for all $r \geq 0$, we have
  $\omega(\mc G \topnab r) < f(r)$.
\end{definition}

\noindent
See the book by \Nesetril and Ossona de Mendez for further definitions~\cite{NOdM12} and
the proof that nowhere dense classes indeed provide a dichotomy of structural sparseness.
We will need the following crucial result proved independently by \Dvorak and Jiang.

\begin{theorem}[\Dvorak~\cite{DvorakThesis}, Jiang~\cite{JiangCliqueMinors}]%
\label{thm:clique-subdiv-dense}
  Let~$\ell \in \N$ and $\epsilon > 0$. There exists integers~$n_{\ell,\epsilon}$
  and~$c_\epsilon$ such that every graph~$G$ with $n > n_{\ell,\epsilon}$
  vertices and at least~$n^{1+\epsilon}$ edges contains a
  $c_\epsilon$-subdivision of~$K_\ell$.
\end{theorem}

%
%
\subsection{Random graph models}

\noindent
We usually denote random variables by upper-case letters. Probabilities are
denoted by $\P[\any]$, expectation, variance, and median by $\E[\any]$, $\Var[\any]$,
$\M[\any]$, respectively. If we need to clarify which probability measures we employ,
we use subscripts like~$\P[\any]_M$. We will use the Iverson bracket notation
$\Iver{\phi}$ for boolean expressions~$\phi$ where~$\Iver{\phi}$ is one if
$\phi$ is true and zero otherwise.
For a sequence of random variables~$(X_n)_{n \in \N}$ and a
random variable~$Y$, recall that~$(X_n)$ \emph{converges in distribution}~to $X$
if it holds that
\[
  \forall k \lim_{n \to \infty} \P[X_n \leq k] = \P[X \leq k].
\]
We denote the convergence in distribution with~$(X_n) \convd X$.

A \emph{random graph model} is a sequence of random variables~$(G_n)_{n \in \N}$
over $n$-vertex graphs. For simplicity, we fix~$V(G_n) = [n]$.
The \emph{parametrisation} of the model is a function~$\rho\colon \N \to
\R^t$ that creates a tuple of~$t$ parameters depending on~$n$ which in turn
determine the probability distribution of each variable~$G_n$.
By~$G(n,\rho(n))$ we denote the random variable~$G_n$ with the probability
distribution prescribed by the model with parameters~$\rho(n)$. In order to
distinguish random models we will introduce superscripts like~$\chunglu$
or~$\kleinberg$.

For a random graph model~$G(n,\rho(n))$ and an integer~$r$ the notation
$G(n,\rho(n)) \topnab r$ denotes a random variable over sets of
graphs with at most~$n$ vertices whose probability distribution is given by
\[
  \P\big[G(n,\rho(n)) \topnab r = A\big] = \!\! \sum_{G: A = G \topnab r} \!\! \P[G(n,\rho(n)) = G],
\]
where~$A$ is a set of graphs. With this definition the
quantity~$\topgrad_r$ is well-defined for every integer~$r$ as a
rational-valued random variable.
As noted above, we study the properties of random graphs in the limit
and hence define the property of having bounded expansion as follows.

\begin{definition}\label{def:probmodel-bnd-exp}
  A graph model~$G(n,\rho(n))$ has \emph{bounded expansion asymptotically
  almost surely (\aas)} if
  there exists a function~$f$ such that for all~$r \geq 0$
  \[
    \lim_{n \to \infty} \P\!\big[\, \topgrad_r( G(n,\rho(n)) ) < f(r) \,\big] = 1.
  \]
  It has \emph{bounded expansion with high probability (\whp)} if
  for every~$c \geq 1$ there exists a function~$f$ such that,
  again for all~$r \geq 0$,
  \[
    \P\!\big[\, \topgrad_r( G(n,\rho(n)) ) < f(r) \,\big] \geq 1 - O(n^{-c}).
  \]
\end{definition}

\noindent The very same definition is possible to define when a random graph
model is nowhere dense.

\begin{definition}\label{def:probmodel-nd}
  A graph model~$G(n,\rho(n))$ is \emph{\aas nowhere dense} if
  there exists a function~$f$ such that for all~$r \geq 0$
  \[
    \lim_{n \to \infty} \P\!\big[\, \omega( G(n,\rho(n) \topnab r) ) < f(r) \,\big] = 1.
  \]
  It is \emph{nowhere dense \whp} if
  for every~$c \geq 1$ there exists a function~$f$ such that,
  again for all~$r \geq 0$,
  \[
    \P[\, \omega( G(n,\rho(n) \topnab r) < f(r) \,] \geq 1 - O(n^{-c}).
  \]
\end{definition}

\noindent
The following notions are needed to prove negative results
about random models.

\begin{definition}
  A graph model~$G(n,\rho(n))$ is \emph{\aas somewhere dense}
  if there exists~$r \in \N$ such that for all functions~$f$ it holds that
  \[
    \lim_{n \to \infty} \P\!\big[\, \omega( G(n,\rho(n)) \topnab r ) > f(r) \,\big] = 1.
  \]
  It is \emph{not \aas nowhere dense} if there exists~$r \in \N$ such that
  for all functions~$f$ it holds that
  \[
    \lim_{n \to \infty} \P\!\big[\, \omega( G(n,\rho(n)) \topnab r ) > f(r) \,\big] > 0.
  \]
\end{definition}

\noindent
Note that the above definition for random graphs is, in contrast to the
definition of structural sparseness for graph classes, not a dichotomy: we
can easily build an artificial random graph model where the probability that a
dense shallow minor exists converges to some value bounded away from zero and one.
This might now just seem like a technicality, but we will see two examples
of random graph models developed to replicated certain aspects of complex
networks that fall exactly in this category.

The following basic lemma will help to simplify the following proofs
by enabling us to work with randomly chosen subgraphs.

\begin{lemma}\label{lemma:why-it-works}
  Let~$X_1, \ldots, X_n$ be binary random variables and let~$S = \sum_{i = 1}^n X_i$.
  Let further~$I \in [n]$ be uniformly distributed. Then
  \[
    \P[S \geq 1] \leq n \cdot \P[X_I = 1].
  \]
\end{lemma}
\begin{proof}
  By Markov's inequality we have that
  \[
    \P[S \geq 1] \leq \E[S] = \sum_{i=1}^n \P[X_i = 1].
  \]
  Observe that
  \[
    \P[X_I = 1] = \sum_{i = 1}^n \P[I = i] \cdot \P[X_i = 1] = \frac{1}{n} \sum_{i=1}^n \P[X_i = 1],
  \]
  and hence
  \[
    \P[S \geq 1] \leq n \cdot \P[X_I = 1].  \qedhere
  \]
\end{proof}

\noindent
We apply this statement to subgraphs in a random graph and obtain the
following corollary that lies closer to our application.

\begin{corollary}\label{cor:why-it-works}
  Let~$G(n, \rho(n))$ be a random graph model with parametrisation~$\rho$
  and~$\Pi$ be a graph property. Then
  \[
    \P[\exists X \subseteq V(G)\colon G[X] \in \Pi ]
    \leq \sum_{k = 1}^n {n \choose k} \P\!\big[ G[Y_k] \in \Pi \big],
  \]
  where~$Y_k$ is a $k$-vertex subset of~$V(G)$ chosen uniformly at random.
\end{corollary}

%
%
\subsection{Asymptotic degree distributions}

\noindent
Given a graph~$G$, the \emph{degree sequence}~$\degseq(G)$ of~$G$
is the sequence~$(\deg(v))_{v \in G}$. We say that two graphs~$G_1$ and~$G_2$ have
the same degree sequence, if~$\degseq(G_1) = \pi(\degseq(G_2))$ for
some permutation~$\pi$. A sequence of~$n$ integers $(d_i)_{1 \leq i \leq n}$
is a degree sequence if it can be realized by a graph, \ie
there exists a graph~$G$
with~$\degseq(G) = (d_i)_{1 \leq i \leq n}$. Such sequences are also
called \emph{graphical}.

To define degree distributions,
consider a random variable~$D_n$ describing the degree of a vertex
chosen uniformly at random from an $n$-vertex graph~$G$.
The pmf~$f_n$ for~$D_n$ is then given by
\[
  f_n(d) = \frac{b_n(d)}{n} = \frac{1}{n} \sum_{v \in V(G)} \Iver{\deg(v) = d},
\]
where~$b_n(d)$ denotes the number of degree-$d$ vertices in~$G$
and~$\Iver{\any}$ is the Iverson bracket.

\begin{definition}[Degree distribution]
  A \emph{$n$-vertex degree distribution} is a random variable~$D$ with
  probability mass function~$f$ such that
  \begin{enumerate}
    \item $f(d) = 0$ for $d \leq 0$ and $d \geq n-1$, and
    \item $nf(d) \in \N_0$ for all~$d \in \N$.
  \end{enumerate}
\end{definition}

\noindent
Note that in the above definition we exclude vertices of degree zero, this
greatly simplifies some of the following proofs. Note that the addition of
any finite number of degree-zero vertices does not change the structural
sparsity characterization of a graph class: under this operation,
a bounded expansion class still has bounded expansion,
a nowhere dense class remains nowhere dense, and a
somewhere dense class is still somewhere dense. Thus, in our setting, this
omission does not affect the generality of our results.

We say that a degree sequence~$(d_i)_{1 \leq i \leq n}$ \emph{matches}
an $n$-vertex degree distribution~$D_n$ with pmf $f_n$
if for every~$1 \leq k \leq n-1$ it holds that
\[
  \sum_{i=0}^n \, \Iver{d_i = k} = n f_n(k).
\]
Consequently, a graph~$G$ \emph{matches} a degree distribution~$D_n$
if its degree sequence does.
Since we will consider sequences of random graphs we need to introduce
a related notation of sequences of degree distributions.

\begin{definition}[Degree distribution sequence, limit, sparse]
  A \emph{degree distribution sequence} is an infinite sequence~$(D_n)$
  of $n$-vertex degree distributions. A random variable
  $D$ is the \emph{limit} of $\cal D$ if $(D_n) \convd D$.
  We say that $\cal D$ is \emph{sparse} if $\E[D] < \infty$
  and $(\E[D_n])_{n \in \N} \converges \E[D]$.
\end{definition}

\noindent
To motivate the definition of \emph{sparse} sequences, note that for
a degree sequence~$D_n$ we have that
\[
  \E[D_n] = \sum_{d} d f_n(d) = \frac{1}{n} \sum_{d=1}^{n-1} d b_n(d),
\]
thus for a graph~$G$ with degree distribution~$D_n$ it holds that
$\E[D_n]$ is exactly its average degree~$\degavg(G)$. The condition of a degree
sequence being sparse will not quite suffice to prove structural sparseness: the
situation parallels how graph class with bounded average degree can still
harbour dense structures. Consider, for example, the degree distribution
sequence of the class consisting of all one-subdivided cliques. While the
sequence has constant mean, the graphs matching it are certainly not
structurally sparse. This observation motivates the following stronger
condition.

\begin{definition}[Tail-bound]
  \label{def:tail-bound}
  A degree distribution sequence~$(D_n)$ with limit~$D$ has the function~$h$
  as its \emph{tail-bound} if there exists a constant~$\tau \geq 0$ such that for
  all~$d \geq \tau$ and large enough~$n$ it holds that
  \[
    \P[\, D_n \geq d \,] = O\Big( \,\frac{1}{ h(d) }\, \Big).
  \]
\end{definition}

\noindent
We observe that in Table~\ref{table:degree-dist}, all listed functions
have a tail-bound that is at least quadratic. This is self-evident for
power-laws with~$\gamma \geq 2$; for the other functions
we simply note that their second moment exists
and hence by Chebyshev's inequality satisfy
\[
  \P[\, D \geq d \,] \leq \frac{\Var[D] }{(d - \E[D])^2} = O\Big( \,\frac{1}{ d^{2} }\, \Big).
\]
for~$d \geq \E[D]$.

We will need the following simple observation about the median $\M[D_n]$
of a degree sequence~$(D_n)$:

\begin{observation}\label{obs:median-conv}
  Let~$(D_n)$ be a sparse degree distribution sequence with limit~$D$.
  Then~$\M[D_n] \converges \M[D]$.
\end{observation}

%
%
%
\section{Graph Models with Bounded Expansion}\label{sec:models}

\noindent
Analytic methods are often selected based on their behavior when applied to
random graphs that are believed to mimic characteristics of the particular
networks under consideration. We will not digress into the arguments for and
against this methodology, but simply note that it is a \emph{de facto} part of
standard practice at this point in time. Accordingly, we must be able to
establish whether or not graphs generated by such models have bounded grad.
For more information on random graph models, we refer the readers to the
surveys in~\cite{NSW01, New03, MR95a,MR98}. In this section, we determine
whether several such models have bounded expansion (as by
Definition~\ref{def:probmodel-bnd-exp}.)

We show that the popular configuration model~\cite{MR95a,MR98}, including the
version with households exhibiting high clustering~\cite{BST09}, and the
Chung--Lu model~\cite{chung2002average,chung2002connected}, has
bounded expansion \whp in the typical parametric range found in its application
to complex networks.

Prior work~\cite{NOdMW12}
has shown that the
\ErdosRenyi\ model has bounded expansion \aas Unfortunately, empirical
analysis of real-world networks (including friendships/social networks,
telephone/communication networks, and biological/neural networks) has shown
that typical degree distributions are measurably different from the Poisson
distribution exhibited by \ErdosRenyi\ graphs (see~\cite{NSW01} and the references
therein.)

We built upon the results of \Nesetril, Ossona de Mendez, and Wood in order
to obtain two major results applicable to complex networks. First, the
configuration and Chung--Lu model can be seen as an extension of \ErdosRenyi\ graphs
that inherit some of the nice mathematical properties (edge probabilities are somewhat independent
in the former and strict independent in the latter) while properly replicating the degree
distribution of complex networks by design: Both models allow us to \emph{prescribe}
the desired distribution. Typical degree distributions found in real-world graphs are listed
in Table~\ref{table:degree-dist}. Since the prescribed distribution is the one parameter
for both models, it is not surprising that the properties of these distributions ultimately
determine their structural sparseness. In order to generate sparse graphs, it is necessary
for the distributions to have a mean that is independent of~$n$. To generate \emph{structurally}
sparse graphs, we show that the distribution in question must have a tail that shrinks at
least as fast as a cubic polynomial. This particularly excludes power-law distributions
with exponent less than three. However, recent findings have shown that \emph{pure}
power-law distributions seem to be rare and a precise statistical analysis often favors
a power-law distribution with an exponential cut-off~\cite{clauset2009power,clauset2018power}. Since the
latter have in particular tails that are dominated by any polynomial, we conclude that
our findings are applicable to the majority of complex networks.

\begin{table}[th]
  \begin{center}
  \begin{tabular}{l>{\quad}ll}
    \toprule
    Name        & Definition~$f(d)$ & Parameters \\
    \midrule
    Power-law             & $d^{-\gamma}$ & $\gamma > 2$ \\
    Power-law w/ cutoff   & $d^{-\gamma} e^{-\lambda d}$ & $\gamma > 2$, $\lambda > 0$ \\
    Exponential           & $e^{-\lambda d}$ & $\lambda > 0$ \\
    Stretched exponential & $d^{\beta-1} e^{-\lambda d^\beta}$ & $\lambda,\beta > 0$ \\
    Gaussian              & $\exp(-\frac{(d-\mu)^2}{2\sigma^2})$ & $\mu,\sigma$ \\
    Log-normal              & $d^{-1} \exp(-\frac{(\log d-\mu)^2}{2\sigma^2})$ & $\mu,\sigma$ \\
    \bottomrule
  \end{tabular}
  \end{center}
  \caption{\label{table:degree-dist}\footnotesize
    A selection of established functions used to model degree distributions of
    complex networks, listed without the necessary
    normalization factors. Here $f(d)$ is the fraction of nodes which have
    degree $d$. These functions were taken from an
    empirical analysis of degree distributions in real-world networks~\cite{clauset2009power}.}
\end{table}

\noindent
The other major result pertains to what we call the \emph{perturbed
bounded-degree model}, which allows the inclusion of an arbitrary (or random) bounded
degree graph in addition to probabilistically generated edges (with non-identical
probabilities, subject to a uniform bound). We show this model to
have bounded expansion with high probability. This in particular strengthens
the aforementioned result on the \ErdosRenyi model in terms of speed of
convergence. Including a base graph and allowing non-uniform edge
probabilities drastically increases the structural variability in the model's
output.  We will further argue that using graphs with many vertices of
unbounded degree as the base graph will necessarily generate structurally
dense graphs. Since this model includes as a special case of certain types of
stochastic block models, it is relevant to the field of complex network.

%
\subsection{The Chung--Lu and configuration model}\label{sec:chung-lu-conf}

\noindent
Given a degree distribution sequence~$\cal D = (D_n)_{n \in \N_0}$ with limit~$D$
and an integer~$n$, we are faced with the task to sample graphs uniformly
at random from the set
\[
  \{ G \mid G~\text{has degree distribution}~D_n \}.
\]
Two methods to accomplish this task---with certain caveats---have
been put forward: the \emph{configuration model} as described by
Bender and Canfield~\cite{BC78} and the
model proposed by Chung and Lu~\cite{chung2002connected,chung2002average}.
The latter is a special case of what has been discussed in the mathematical literature
as inhomogeneous random graphs (see, \eg, the work by \Bollobas, Janson, and Riordan
on the phase transition of such graphs~\cite{bollobas2007phase}).

To sample a graph according to the configuration model, we proceed as follows:
\begin{enumerate}
  \item Build a degree sequence~$(d_i)_{1 \leq i \leq n}$ that
        matches~$D_n$.
  \item Construct a vertex set~$V^C = \{ v^1_i, \ldots, v^{d_i}_i \}_{1 \leq i \leq n}$,
        \ie create~$d_i$ copies (called \emph{stubs}) for what will be the vertex~$v^i$
        in the final graph.
  \item Generate an auxiliary graph~$H$ with vertex set~$V^C$ and a
        random matching as its edge set.
  \item Assemble the multi-graph~$G'$ with vertex set~$\{ v_i \}_{1 \leq i \leq n}$
        and
        \[
          |E(v_i,v_j)| = |E_H( \{ v^1_i, \ldots, v^{d_i}_i\}, \{ v^1_j, \ldots, v^{d_j}_j \} )|,
        \]
        that is, we connect the vertices~$v_i$ and~$v_j$ with as many
        edges as we find between their respective copy-classes in~$H$.
  \item Return the graph~$G$ derived from~$G'$ by removing all parallel edges
        and loops.
\end{enumerate}

\noindent
Graphs generated this way will henceforth be denoted by~$\confmodel(D_n)$. The
name \emph{stubs} derives from the following picture of the process: we affix
to every vertex a number of half-edges, the stubs, that matches its degree
according to the generated degree sequence, and then obtain the multi-graph by
randomly wiring the stubs to each other. This inspires two other methods of
generating the configuration model: Instead of generating a matching between
the stubs, we can instead choose a random permutation of them and match them
up pair-by-pair according to that permutation. It is easy to see that this
method is equivalent to the process described above since every matching
of~$V^C$ has the same number of permutations that generate it. The second
method works as follows. We prescribe an (arbitary) order~$u_1,\ldots,u_m$ on
the set~$V^C$ and generate edges one by one as follows:  pick the first yet
unmatched vertex~$u_i$ as the first endpoint of the next edge and then choose
the other endpoint~$u_j$, $j > i$, uniformly at random among all yet unmatched
vertices. To see that this is equivalent to drawing a matching uniformly at
random, simply note that the probability for each fixed matching given the
order~$u_1,\ldots,u_m$ is precisely~$1/(m (m-2) (m-4) \ldots) = 1/m!!$, or
one over the number of matchings on~$m$ stubs.
Both these alternative views of the process will be beneficial later on.

It is a priori not clear that either procedure generates graphs with the
correct degree sequence. The intermediate multi-graph~$G'$ trivially
exhibits the degree sequence~$(d_i)_{1 \leq i \leq n}$ and thus has
the degree distribution~$D_n$. The last step, however, might skew the
result by removing parallel edges and loops---we therefore need
the probability that~$G'$ contains such offending edges
to be reasonably low. The conditions under which this is the case
have been proved by Molloy and Reed~\cite{MR95a} and
subsequently improved by Janson, whose result we present here using
our own notation.

\begin{theorem}[Janson~\cite{ConfModelSimple}]\label{thm:confmodel-simple}
  Let~$(D_n)_{n \in \N_0}$ be a degree distribution sequence. Then
  we have that
  \[
    \liminf_{n \to \infty} \, \P[\, \confmodel(D_n)~\text{is simple} \,] > 0
    \iff
    \E[D^2_n] = O( \E[D_n] ).
  \]
\end{theorem}

\noindent
The original formulation of the theorem's condition is that
\begin{align*}
  \sum_{d \geq 0} d^2 b_n(d) = O\big( \sum_{d \geq 0} d b_n(d) \big).
\end{align*}
Note that in the case of sparse degree
distribution sequences this condition is equivalent to saying that
$\Var[D]$ is finite. Luckily,
this is the case for all degree distributions listed in Table~\ref{table:degree-dist}
with the exception of the power-law distribution where
the variance is finite only for~$\gamma > 3$.

The second method for sampling graphs with a prescribed degree distribution,
the Chung--Lu model, forgoes the above problems by generating graphs
whose \emph{expected} degree distribution matches the given one.
Given~$D_n$, it constructs a random graph as follows\footnote{We excluded loops
for simplicity here, including loops does not change our result.}:
\begin{enumerate}
  \item Build a degree sequence~$(d_i)_{1 \leq i \leq n}$ that
        matches~$D_n$. We will call~$d_i$ the \emph{weight}
        of the vertex~$i$.
  \item Create a graph on~$n$ vertices~$v_1, \ldots v_n$ and connect
        each pair of vertices~$v_i, v_j$, $i < j$, with probability~$d_i d_j / m$
        where~$m = \sum_{k=0}^n d_k$.
\end{enumerate}

\noindent
As network models, both the configuration and the Chung--Lu model suffer from
some shortcomings. While they, by design, generate graphs with the correct
degree-distribution and small diameter, other statistics found in complex
networks are not replicated. In
particular, both models have a vanishing clustering-coefficient
(see, \eg, Newman's survey~\cite{NewmanSurvey}).

Since this statistic is critical in many real-world applications,
methods to `fix' these models have been put forward. A notable
example is the \emph{configuration model with household
structure} as defined by Ball, Sirl, and Trapman~\cite{BST09}.
For this variant, one samples a graph with a prescribed degree sequence and
then replaces every vertex by a constant-sized `household'-graph (for example
a clique), distributing
the edges incident to a household uniformly to the vertices that comprise it.
The resulting graph has a provably constant clustering coefficient.

The main goal of this section will be the proof of the following
Theorem~\ref{thm:conf-chunglu-char}. To this end, we will first show the
behaviour of the Chung--Lu random graphs with respect to~$\topgrad_0$
and~$\omega$. By a simple observation about the probability of short paths
existings, we then relate the statistics~$\topgrad_r$ and~$\omega( \any \topnab
r)$ of different graph models to these base cases. As it turns out, the
structural sparseness of the Chung--Lu and the configuration model is entirely
determined by the \emph{tail} of the prescribed degree distribution:

\begin{theorem}\label{thm:conf-chunglu-char}
  Let~$(D_n)$ be a sparse degree distribution sequence with tail~$O(1/d^\gamma)$.
  Both the configuration model~$\confmodel(D_n)$ and the Chung--Lu
  model~$\chunglu(D_n)$, with high probability,
  \begin{itemize}
    \item have bounded expansion for~$\gamma > 3$,
    \item are nowhere dense (with unbounded expansion)
          for~$\gamma = 3$
    \item and are somewhere dense for~$\gamma < 3$.
  \end{itemize}
\end{theorem}

\noindent
Since we can emulate household structures
by simply taking the lexicographic product with some constant-size
clique and then taking a subgraph, Theorem~\ref{thm:conf-chunglu-char}
immediately implies the same for those variants.

\begin{corollary}
  Let~$(D_n)$ be a sparse degree distribution sequence.
  Then the configuration model~$\confmodel(D_n)$ as well as the Chung--Lu
  model~$\chunglu(D_n)$ with households have bounded expansion \whp
  if the sequence has a superquadratic tail-bound and are nowhere dense \whp
  if it has a quadratic tail-bound.
\end{corollary}

%
%
\subsection{Tools for sparse degree distribution sequences}

\noindent
In both the Chung--Lu and the configuration model, the first phase consists
of assigning weights to vertices according to a degree distribution~$D_n$.
This process chooses degrees without replacement, therefore we
need to ensure that the important properties of~$D_n$ carry over even if a
fraction of the vertices have been uncovered already, that is, we know
their weight and hence cannot assume they are randomly distributed.

Since in the following proofs low weights are always preferable, we consider
a `worst-case' variable describing the degree of a vertex after at most~$n/c$
other vertices have been assigned the lowest available degrees. Luckily, sparse
degree distributions are robust under truncation up to the median.
Using this idea, the following lemma shows that we can, up to a point,
assume that the vertex degrees are drawn independently according to
a modified distribution.

\begin{lemma}\label{lemma:no-replacement}
  Let~$(D_n)$ be a sparse degree distribution sequence with
  limit~$D$ and a tail-bound~$O(h(d)^{-1})$.
  Let~$\mu_\half = \M[D]$ be the median of~$D$.
  Then the sequence~$(\hat D_n)$ defined via
  \[
    \P[\hat D_n = d\,] = \P\!\big[D_n = d \bigm\vert D_n \geq \mu_\half\big]
  \]
  is sparse and has tail-bound~$O(h(d)^{-1})$.
\end{lemma}
\begin{proof}
  Let~$\hat D_n$ be defined as the conditioned random variable~$(D_n \mid D_n > \mu_\half)$,
  where~$\mu_\half$ is the median of~$D_n$. Then we have that
  \begin{align*}
    \E[\hat D_n] &= \sum_{d > 0} d \P[D_n = d \mid D_n \geq \mu_\half]
          = \sum_{d \geq \mu_\half} d \cdot \frac{ \P[D_n = d\,]}{ \P[D_n \geq \mu_\half] }  \\
          &=  \frac{\E[D_n]}{\P[D_n \geq \mu_\half]} \leq \frac{ \E[D_n] }{2}.
  \end{align*}
  Hence~$\E[\hat D_n]$ is finite and~$(\hat D_n)$ is sparse.

  Let~$\tau$ be the threshold for which the tail-bound~$O(h(d)^{-1})$ on~$(D_n)$
  holds. To see that the same tail-bound works for~$(\hat D_n)$, note that
  \begin{align*}
    \P[\hat D \geq d\,] &= \P[D \geq d \mid D \geq \mu_\half]
                  = \frac{\P[D \geq \max\{d,\mu_\half\} ]}{\P[D \geq \mu_\half]}.
  \end{align*}
  Hence for~$\hat \tau \geq \max\{\tau,\mu_\half\}$ for all~$d \geq \hat \tau$ the bound
  $
    \P[\hat D \geq d\,] = O( h(d)^{-1} )
  $
  holds, as claimed.
\end{proof}

\noindent In the remainder of this section, we work towards a proof of
Theorem~\ref{thm:conf-chunglu-char}. Ultimately, our approach relates the
density of shallow minors for graphs with degree distribution~$D$ to the density
of subgraphs of a random graph with distribution~$\eta D$, where~$\eta$ is an
appropriate scaling factor. This factor traces back to the value of the harmonic
sum $ \sum_{d = 1}^{\Delta} \frac{1}{d^{\gamma}} $ which is a constant
for~$\gamma > 1$, roughly $\log \Delta$ for~$\gamma = 1$ and
approximately~$\Delta^{1-\gamma}$ for~$\gamma < 1$. We will make use of the
following simple bounds for harmonic sums.

\begin{lemma}\label{lemma:harmonic-bound}
  For all integers~$0 < \delta \leq \Delta$ the bound
  $
    R_\gamma \leq \sum_{k = \delta}^{\Delta} \frac{1}{k^\gamma} \leq\frac{1}{\delta^\gamma} + R_\gamma
  $
  holds where
  \[
    R_\gamma = \begin{cases}
      \frac{1}{\gamma-1} ( \delta^{1-\gamma} - \Delta^{1-\gamma}) & \text{for}~\gamma > 1, \\
      \ln \Delta - \ln \delta                                     & \text{for}~\gamma = 1,~\text{and} \\
      \frac{1}{1-\gamma} (\Delta^{1-\gamma} - \delta^{1-\gamma})  & \text{for}~0 < \gamma < 1.
    \end{cases}
  \]
\end{lemma}
\begin{proof}
  The function~$k^{-\gamma}$ is monotonically decreasing on~$(0,\infty)$
  and therefore it is true that
  \[
    \int_{\delta}^{\Delta} \! \frac{1}{k^\gamma} \, dk
    \leq \sum_{k = \delta}^{\Delta}  \frac{1}{k^\gamma}
    \leq \frac{1}{\delta^\gamma} + \int_{\delta}^{\Delta} \! \frac{1}{k^\gamma} \, dk.
  \]
  For~$\gamma > 1$, the indefinite integral~$\int k^{-\gamma} \, dk$ equals
  $-\frac{1}{(\gamma-1) k^{\gamma-1}}$, for~$\gamma = 1$ it is~$\ln(k)$, and
  for~$\gamma < 1$ it equals~$\frac{1}{(1-\gamma)} k^{1-\gamma}$. Evaluating the
  integral on the interval~$[\delta,\Delta]$ yields~$R_\gamma$ and the claim
  follows.
\end{proof}


\noindent
We furthermore need the following bounds for harmonic sums that contain a logarithmic factor:
\begin{lemma}\label{lemma:log-harmonic-bound}
  For~$\gamma > 0$, $r \geq 1$ and integers~$r^{2r} < \delta \leq \Delta$ the bound
  \[
    R'_\gamma \leq \sum_{k = \delta}^{\Delta} \frac{\ln^r k}{k^\gamma}
               \leq \frac{\ln^r \delta}{\delta^{\gamma}} + \zeta R'_\gamma
  \]
  holds where~$\zeta$ is a constant and
  \[
    R'_\gamma = \begin{cases}
        \frac{1}{\gamma-1} ( \delta^{1-\gamma} \ln^r \delta - \Delta^{1-\gamma} \ln^r \Delta)    & \text{for}~\gamma > 1, \\[1ex]
        \frac{1}{\zeta(r+1)} ( \ln^{r+1} \Delta - \ln^{r+1}\delta)                               & \text{for}~\gamma = 1,~\text{and} \\[1ex]
        \frac{1}{(1-\gamma)^{r+1}} ( \Delta^{1-\gamma} \ln^r \Delta - \delta^{1-\gamma} \ln^r \delta )    & \text{for}~0 < \gamma < 1.
    \end{cases}
  \]
\end{lemma}
\begin{proof}
  Since~$\ln^r( r^{2r} ) \leq r^{2r\gamma}$ holds for all positive~$r$,
  the function~$\ln^r(k) k^{-\gamma}$ is monotonically decreasing on~$(r^{2r},\infty)$
  and therefore it is true that
  \[
    \int_{\delta}^{\Delta} \! \frac{\ln^r k}{k^\gamma} \, dk
    \leq \sum_{k = \delta}^{\Delta} \frac{\ln^r k}{k^\gamma}
    \leq \frac{\ln^r \delta}{\delta^\gamma} + \int_{\delta}^{\Delta} \! \frac{\ln^r k}{k^\gamma} \, dk.
  \]
  For~$\gamma > 1$, the integral can be bounded by
  \[
    \frac{\ln^r k}{(\gamma-1) k^{\gamma-1}} + \Theta(1) \leq \int \frac{\ln^r k}{k^\gamma} \, dk
    \leq \frac{\zeta \ln^r k}{(\gamma-1) k^{\gamma-1}} + \Theta(1),
  \]
  for some constant~$\zeta$. For~$0 < \gamma < 1$ we obtain the bounds
  \[
    \Big( \frac{1}{1-\gamma} \Big)^{r+1} k^{1-\gamma} \ln^r k + \Theta(1)
    \leq \int \frac{\ln^r k}{k^\gamma} \, dk
    \leq
    \zeta \Big( \frac{1}{1-\gamma} \Big)^{r+1} k^{1-\gamma} \ln^r k + \Theta(1).
  \]
  Finally, for~$\gamma = 1$, the integral simply evaluates to
  \[
    \int \frac{\ln^r k}{k} \, dk
    = \frac{\ln^{r+1} k}{r+1} + \Theta(1).
  \]
  Applying the respective indefinite integral to the interval~$[\delta, \Delta]$
  yields the claimed bounds.
\end{proof}

\noindent
We will use the following lemma to bound the probability that dense subgraphs
appear in a Chung--Lu graph, assuming that the weights of the subgraph obey
some bound.

\begin{lemma}\label{lemma:chunglu-chernoff}
  Consider~$k$ vertices~$\{v_i\}_{i \in [k]}$ with associated weights~$\{ d_i
  \}_{i \in [k]}$. Let~$G$ be a random graph on these vertices where each
  edge~$v_iv_j$ is independently present with probability~$\leq \beta d_i d_j /
  n$. Then, for any positive constant~$\xi$,
  \[
    \P\!\big[\,|E(G)| \geq \xi k\,\big]
    \leq \Big(  \frac{  e\beta d^2  }{  2 n \xi k e^{d^2 / 2n} }   \Big)^{\xi k}
  \]
  if~$d := \sum_{i} d_i$ satisfies~$\beta d^2 \leq 2n\xi k$.
\end{lemma}
\begin{proof}
  We associate a random variable~$X_{ij}$ with every edge~$v_i v_j$. The expected
  number of edges is then
  \[
    \E\!\big[\,\sum_{i < j} X_{ij}\,\big]
     = \sum_{i < j} \E[\,X_{ij}\,] \leq \sum_{i < j} \frac{\beta d_i d_j}{n}
     \leq \frac{\beta}{2n} \Big( \sum_{i} d_i \Big)^2 = \frac{\beta d^2}{2n}.
  \]
  We apply the Chernoff-bound
  \[
    \P\Big[\, \sum_{ij} X_{ij} \geq (1+\delta)\frac{\beta d^2}{2n} \,\Big]
    \leq \Big( \frac{ e^{\delta} }{ (1+\delta)^{1+\delta}  } \Big)^{\beta d^2/2n}
  \]
  choosing~$\delta = \frac{2n \xi k}{\beta d^2} - 1$ and obtain
  \[
    \P\Big[\, \sum_{ij} X_{ij} \geq \xi k \,\Big]
    \leq \frac{ e^{\xi k - \beta d^2 / 2n\xi k}  }{ (\frac{2 n \xi k}{\beta d^2})^{ \xi k }  }
    = \Big(  \frac{ e\beta d^2 }{  2 n \xi k e^{d^2 / 2n} } \Big)^{\xi k},
  \]
  as claimed.
\end{proof}

\noindent
Let us convince ourselves that the above lemma will be applicable for
\emph{any} choice of~$k$ vertices for degree distributions with supercubic
tails. To that end, we derive the following bound on the degree-sum of $k$
vertices, which implies that Lemma~\ref{lemma:log-harmonic-bound} is applicable
to such distributions with~$\xi \geq 8$.

\begin{lemma}
  Let~$(D_n)$ be a sparse degree distribution sequence with
  tail-bound~$\lambda/d^{\alpha + 1}$, $\alpha \geq 2$. Let~$d_1,\ldots,d_k$ be
  the highest~$k$ degrees of~$D_n$. Then, for large enough~$n$,
  \[
    \Big( \sum_{i=1}^{k} d_i \Big)^2 \leq 16 n k.
  \]
\end{lemma}
\begin{proof}
  Let~$\Delta = (\lambda n)^\frac{1}{\alpha+1} $ be the maximum realizable
  degree of~$D_n$. To bound the sum-of-degrees for the top~$k$ vertices, let us
  determine a degree~$\delta$ such that
  \[
    \sum_{d = \delta}^{\Delta} \frac{n}{d^{\alpha+1}} \geq k.
  \]
  Applying the bound from Lemma~\ref{lemma:harmonic-bound}, we can solve
  for~$\delta$:
  \[
    \sum_{d = \delta}^{\Delta} \frac{n}{d^{\alpha+1}}
    \geq \frac{n}{\alpha} \Big( \frac{1}{\delta^\alpha} - \frac{1}{\Delta^\alpha}\Big)
    \stackrel{!}{\geq} k
    \iff
    \delta \leq \big( \alpha k/n + 1/\Delta^\alpha \big)_\FullStop^{-\frac{1}{\alpha}}
  \]
  Given the weight~$\delta = (\alpha k/n + 1/\Delta^\alpha)^{-\frac{1}{\alpha}}$
  we can now bound the degree-sum of the topmost~$k$ vertices by applying the
  other side of the bound given in Lemma~\ref{lemma:harmonic-bound}:
  \begin{align*}
    \sum_{d = \delta}^{\Delta} \frac{n \cdot d}{d^{\alpha+1}}
     &\leq \frac{n}{\alpha-1}\Big( \frac{\delta + \alpha-1}{\delta^{\alpha}} - \frac{1}{\Delta^{\alpha-1}} \Big)
      \leq \frac{2n}{\alpha-1} \cdot \frac{1}{\delta^{\alpha-1}},
  \end{align*}
  were we used that~$\delta > \alpha-1$. Plugging in the above values
  for~$\delta$ and~$\Delta$ we obtain
  \begin{align*}
      &\phantom{{}\leq{}}
      \frac{2n}{\alpha-1} \big( \frac{\alpha k}{n} + (\lambda n)^{-\frac{\alpha}{\alpha+1}}  \big)^{\frac{\alpha-1}{\alpha}}
  \end{align*}
  as an upper bound for the degree-sum. Let~$k := \tau n$ for some~$0 < \tau
  \leq 1$. We claim the above expression can be upper-bounded by
  $\big(\frac{2}{1-\alpha} (2\alpha)^{\frac{\alpha-1}{\alpha}} \big)^2 nk =
  \big( \frac{2}{1-\alpha} (2\alpha)^{\frac{\alpha-1}{\alpha}} \big)^2 \tau
  n^2$:
  \begin{alignat*}{3}
    &&
      \Big(
        \frac{2n}{\alpha-1}
          \big( \frac{\alpha \tau n}{n} + (\lambda n)^{-\frac{\alpha}{\alpha+1}}  \big)^{\frac{\alpha-1}{\alpha}}
      \Big)^2
      &\stackrel{!}{\leq} \big(\frac{2}{\alpha-1} (2\alpha)^{\frac{\alpha-1}{\alpha}} \big)^2 \tau n^2 & \\
    &\iff&
      (\alpha \tau)^{\frac{\alpha-1}{\alpha}}
        \big(1 + (\alpha \tau)^{-1}(\lambda n)^{-\frac{\alpha}{\alpha+1}}  \big)^{\frac{\alpha-1}{\alpha}}
      &\stackrel{!}{\leq} (2\alpha)^{\frac{\alpha-1}{\alpha}} \sqrt \tau & \\
    &~\Longleftarrow&
        (\alpha \tau)^{\frac{\alpha-1}{\alpha}} 2^{\frac{\alpha-1}{\alpha}}
        &\stackrel{!}{\leq} (2\alpha)^{\frac{\alpha-1}{\alpha}} \sqrt \tau & \\
    &\iff&
      \tau^{\frac{\alpha-2}{2\alpha}} &\stackrel{!}{\leq} 1 &
  \end{alignat*}
  where the backwards implication holds when $n$ is large enough such that
  $(\alpha \tau)^{-1}(\lambda n)^{-\frac{\alpha}{\alpha+1}} \leq 1$. Note that
  since~$\alpha \geq 2$, the exponent of~$\tau$ is positive and hence the last
  inequality holds. Finally, the function $\big(\frac{2}{1-\alpha}
  (2\alpha)^{\frac{\alpha-1}{\alpha}} \big)^2$ restricted to~$[2,\infty)$
  attains its maximum at~$\alpha = 2$ with the value~$16$. Thus,
  $\big(\frac{2}{1-\alpha} (2\alpha)^{\frac{\alpha-1}{\alpha}} \big)^2 \tau nk
  \leq 16nk$, as claimed.
\end{proof}

\noindent
We further need a bound on the distribution of the degree-sum for~$k$ vertices chosen
at random. Note that we can apply the following lemma for drawing~$k$ vertices from
the \emph{same} distribution by applying Lemma~\ref{lemma:no-replacement} as long
as~$k \leq n/2$.

\begin{lemma}\label{lemma:weightsum-chernoff}
  Let~$(D_n)$ be a sparse degree distribution with tail-bound~$\lambda/d^{\alpha+1}$,
  $\alpha > 1$ which applies above the threshold~$\tau$.
  Then there exists a constant~$\zeta$ such that, for large enough~$n$, the probability
  distribution of the sum of~$k$ independent copies
  $D_n^{(i)}$, $1 \leq i \leq k$, is bounded by
  \[
    \P\Big[\, \sum_{i=1}^k D^{(i)}_n \geq d \,\Big]
      \leq \frac{ (e\zeta k)^d }{ d^d e^{\zeta k} }
  \]
  for every~$d \geq \zeta k$.
\end{lemma}
\begin{proof}
  Let~$\Delta$ be the largest possible degree that is realisable in~$D_n$
  and let~$\tau$ be the threshold at which the tail-bound~$h(d)$ holds.
  Then the expected value of~$\mathbb D := \sum_i D^{(i)}_n$ is given by
  \begin{align*}
    \E[\mathbb D] &= k \E[D_n] = k \sum_{d = 1}^{\Delta} d \P[D_n = d]
    \leq k \Big( \sum_{d = 1}^{\tau-1} d + \sum_{d = \tau}^{\Delta} \frac{\lambda}{d^\alpha} \Big) \\
    &\leq k \Big( \tau^2 + \lambda \big(
            \frac{1}{\tau^\alpha}
          + \frac{1}{(\alpha-1)\tau^{\alpha-1}}
          - \frac{1}{(\alpha-1)\Delta^{\alpha-1}}
       \big) \Big)
    =: \zeta k,
  \end{align*}
  where we used the first bound of Lemma~\ref{lemma:harmonic-bound} for the second inequality.
  Applying the Chernoff-bound
  \[
    \P\Big[\, \mathbb D \geq (1+\delta) \zeta k \,\Big]
    \leq \Big( \frac{ e^{\delta} }{ (1+\delta)^{1+\delta}  } \Big)^{\zeta k}
  \]
  with~$\delta = \frac{d}{\zeta k} - 1$ results in the bound
  \[
    \P\Big[\, \mathbb D \geq d \,\Big]
    \leq \frac{ e^{d - \zeta k} }{ (d/\zeta k)^{d} }
    = \frac{ (e\zeta k)^d }{ d^d e^{\zeta k} },
  \]
  as claimed.
\end{proof}

\noindent
Finally we will need the following bound on the \emph{product} of~$k$ randomly
chosen vertices. Again, we can apply the following lemma for drawing~$k$ vertices from
the \emph{same} distribution by applying Lemma~\ref{lemma:no-replacement} as long
as~$k \leq n/2$.

\begin{lemma}\label{lemma:product-bound}
  Let~$(D_n)$ be a sparse degree distribution with (upper) tail-bound~$\lambda/d^{\alpha+1}$,
  $\alpha \geq 1$ which applies above the threshold~$\tau$.
  Then, for large enough~$n$, the probability distribution of the product of~$r$ independent variables
  $D^{(i)}_n$, $1 \leq i \leq r$, is bounded by
  \[
    \P\!\Big[\, \prod_{i=1}^r D^{(i)}_n = d \,\Big]
      \leq \frac{\zeta' (\alpha+1)^r e^{(\alpha+1)\tau}}{(r-1)!} \frac{(\ln d)^{\mathrlap{r-1}}}{d^{(1+\alpha)}},
  \]
  where~$\zeta'$ is some constant.
\end{lemma}
\begin{proof}
  Observe that
  \[
    \P\!\Big[ \prod_{i=1}^r D^{(i)}_n = d \,\Big]
    = \P\!\Big[ \sum_{i=1}^r \ln D^{(i)}_n = \ln d \,\Big]
  \]
  where, for each~$1 \leq i \leq r$,
  \[
    \P\!\big[ \ln D^{(i)}_n = d \,\big] \leq \frac{\lambda'}{e^{(\alpha+1)(d-\tau)}},
  \]
  for some normalising constant~$\lambda'$. That is, the random variables~$\ln D^{(i)}_n$
  follow a (shifted) exponential distribution. Accordingly, the sum of~$r$ independent
  random variables~$\ln D^{(i)}_n$ has an Erlang-distribution:
  \[
    \P\!\Big[\sum_{i=1}^r \ln D^{(i)}_n = \hat d \,\Big]
    \leq \frac{\zeta' (\alpha+1)^r}{(r-1)!} \cdot \frac{(\hat d-\tau)^{r-1}}{e^{(\alpha+1)(\hat d-\tau)}}
    \leq \frac{\zeta' (\alpha+1)^r e^{(\alpha+1)\tau}}{(r-1)!} \cdot \frac{\hat d^{r-1}}{e^{(\alpha+1)\hat d}},
  \]
  for some scaling constant~$\zeta'$. Therefore,
  \[
    \P\!\Big[ \prod_{i=1}^r \hat D^{(i)}_n = d \Big]
    = \P\!\Big[ \sum_{i=1}^r \ln D^{(i)} = \ln d \,\Big]
    \leq \frac{\zeta' (\alpha+1)^r e^{(\alpha+1)\tau}}{(r-1)!} \frac{(\ln d)^{r-1}}{d^{(1+\alpha)}}_\FullStop
  \]
\end{proof}

\noindent
We find that a similar lower-bound holds in cases where we have a lower-bound on the tail-distribution:

\begin{lemma}\label{lemma:product-bound-lower}
  Let~$(D_n)$ be a sparse degree distribution with (lower)
  tail-bound~$\lambda/d^{\alpha+1}$, $\alpha \geq 1$ which applies above the
  threshold~$\tau$. Then, for large enough~$n$, the probability distribution of
  the product of~$r$ independent variables $D^{(i)}_n$, $1 \leq i \leq r$, is
  lower-bounded by
  \[
    \P\!\Big[\, \prod_{i=1}^r D^{(i)}_r = d \,\Big]
      \geq \frac{\zeta' (\alpha+1)^r e^{(\alpha+1)\tau}}{2^{r-1}(r-1)!}
           \cdot \frac{(\ln d)^{\mathrlap{r-1}}}{d^{(1+\alpha)}},
  \]
  where~$\zeta'$ is some constant and~$\ln d \geq 2\tau$.
\end{lemma}
\begin{proof}
  We proceed analogous to the proof of Lemma~\ref{lemma:product-bound}. Again,
  the random variables~$\ln D^{(i)}_n$, $1 \leq i \leq r$, follow a (shifted)
  exponential distribution and the sum of all~$r$ variables follows
  an Erlang-distribution. We now obtain a lower bound using
  the lower bound on the tail of~$D_n$:
  \[
    \P\!\Big[\sum_{i=1}^r \ln D^{(i)} = \hat d \,\Big]
    \geq \frac{\zeta' (\alpha+1)^r}{(r-1)!} \cdot \frac{(\hat d-\tau)^{r-1}}{e^{(\alpha+1)(\hat d-\tau)}}
    \geq \frac{\zeta' (\alpha+1)^r e^{(\alpha+1)\tau}}{2^{r-1}(r-1)!} \cdot \frac{\hat d^{r-1}}{e^{(\alpha+1)\hat d}},
  \]
  for some scaling constant~$\zeta'$, where we used that~$\tau \leq \hat d/2$.
  Therefore,
  \[
    \P\!\Big[ \prod_{i=1}^r D^{(i)}_n = d \, \Big]
    = \P\!\Big[ \sum_{i=1}^r \ln D^{(i)} = \ln d \,\Big]
    \geq \frac{\zeta' (\alpha+1)^r e^{(\alpha+1)\tau}}{2^{r-1}(r-1)!}
          \cdot \frac{(\ln d)^{r-1}}{d^{(1+\alpha)}},
  \]
  where we used that~$\tau \leq \ln d / 2$.
\end{proof}

%
%
%
\subsection{Local density of Chung--Lu graphs}


\noindent
We can now show that Chung--Lu graphs generated with
degree distributions that have a supercubic or cubic tail will not contain
dense subgraphs. At this point supercubic and cubic tails seems to behave
similarly, we will see in the next section why the bounds that hold for
subgraphs in both cases will fail for the cubic case once we lift the result
to shallow minors.

\begin{lemma}\label{lemma:no-dense-subgraph}
  Let~$(D_n)$ be a sparse degree distribution sequence with limit~$D$ and
  tail-bound~$\lambda/d^{\alpha+1}, \alpha \geq 2$. Let~$\zeta$
  be the constant from Lemma~\ref{lemma:weightsum-chernoff} associated with~$(D_n)$.
  Then for every~$\xi \geq e^4 \zeta^2 c$,
  every~$c \geq 2e$ and every~$n \geq 4\xi$ it holds that
  \[
    \P\!\big[\, \exists H \subseteq \chunglu(D_n) ) : |H| \leq n/c ~\text{and}~ \grad_0(H) \geq \xi \,\big]
    \leq \frac{1}{n^\xi}.
  \]
\end{lemma}
\begin{proof}
  Using Lemma~\ref{lemma:why-it-works}, we can bound the probability that a
  dense subgraph on~$k$ vertices exists by considering the probability that $k$
  randomly chosen vertices form a dense subgraph. Taking the union bound of all
  possible~$k$ (note that we need at least~$2\xi + 1$ vertices for a subgraph of
  density~$\xi$, we simplify this lower bound to~$2\xi$), the probability of the
  aforementioned event is at most
  \[
    \sum_{k = 2\xi}^{n/c} {n \choose k} \P\!\big[\,\| \chunglu(D_n)[X_k]) \| \geq k \xi \,\big],
  \]
  where~$X_k$ is a set of~$k$ vertices chosen uniformly at random.
  Let us write~$\mathbb D$ to denote the degree-sum of~$k$ vertices chosen
  from~$D_n$.  By applying Lemma~\ref{lemma:chunglu-chernoff} to the random
  subgraph~$\chunglu(D_n)[X_k]$, we can bound the above probability by
 \begin{align*}
    &\phantom{{}\leq{}}
      \sum_{k = 2\xi}^{n/c} {n \choose k} \sum_{d = k}^{\Delta k}
       \left(    \frac{ec d^2}{2n \xi k e^{d^2/2n}}     \right)^{\xi k} \! \P[\mathbb D = d]  \\
    &\leq
      \sum_{k = 2\xi}^{n/c}
        \Big( \frac{e n}{k} \Big)^k
        \left(  \frac{ec}{2n \xi k}  \right)^{\xi k}
      \sum_{d = k}^{\Delta k} d^{2\xi k}  \P[\mathbb D  = d], \\
    &\leq
      \sum_{k = 2\xi}^{n/c}
        \Big( \frac{e^2 c}{2\xi}
               \cdot \frac{k}{n} \Big)^{\xi k}
        \frac{n^k}{k^{(2\xi+1) k}}
      \sum_{d = k}^{\Delta k} d^{2\xi k}  \P[\mathbb D  = d].
  \intertext{%
    For~$d \geq \zeta k$, we can apply the bound on~$\P[\mathbb D  \geq d]$ from
    Lemma~\ref{lemma:weightsum-chernoff}:
    note that~$\mathbb D \leq \sum^k_{i = 1} D^{(i)}_n$ in the stochastic sense,
    \ie $\P[\mathbb D \geq d] \leq \P[\sum^k_{i = 1} D^{(i)}_n \geq d]$,
    since summing~$k$ independent copies of~$D_n$ will result in a larger degree-sum
    than drawing degrees from~$D_n$ without replacement.
    We split the inner sum at~$d = \zeta k$ and obtain
  }
    &\phantom{{}\leq{}}
      \sum_{k = 2\xi}^{n/c}
        \Big( \frac{e^2 c}{2\xi}
               \cdot \frac{k}{n} \Big)^{\xi k}
        \frac{n^k}{k^{(2\xi+1) k}}
      \bigg(
         \sum_{d = k}^{\zeta k - 1} d^{2\xi k}
       + \sum_{d = \zeta k}^{\Delta k} \frac{d^{2\xi k} (e \zeta k)^d}{d^d e^{\zeta k}}
      \bigg).
  \end{align*}
  Let us first find an appropriate bound for the two innermost sums. The first
  one is easily bounded by
  \[
    \sum_{d = k}^{\zeta k-1} d^{2\xi k} \leq (\zeta k)^{2\xi k + 1} \leq k \zeta^{3\xi k} \cdot k^{2\xi k}
  \]
  and we can find a comparable bound for the second part:
  \begin{claim}
    \[
      \sum_{d = \zeta k}^{\Delta k} \frac{d^{2\xi k} (e \zeta k)^d}{d^d e^{\zeta k}}
      \leq 3 \xi k (e \zeta)^{2\xi k}  \cdot k^{2\xi k}
    \]
  \end{claim}
  \begin{proof}
    The sum in question has a tail that decreases supergeometrically. In order
    to identify the value for~$d$ from which on this property holds, we
    consider the ratio of two consecutive terms:
    \begin{align*}
      &\phantom{{}\leq{}}
      \frac{(d-1)^{2\xi k} (e \zeta k)^{d-1}}{(d-1)^{d-1} e^{\zeta k}} \cdot
      \frac{d^d e^{\zeta k}}{d^{2\xi k} (e \zeta k)^d}
      = \frac{(d-1)^{2\xi k} }{(d-1)^{d-1}} \cdot
      \frac{d^d}{d^{2\xi k} e \zeta k} \\
      &\geq \Big(1-\frac{1}{d}\Big)^{d(2\xi k/d-1+1/d)} \cdot
      \frac{d}{e \zeta k}
      \geq \Big(\frac{1}{2e}\Big)^{2\xi k/d} \cdot
      \frac{d}{e \zeta k},
    \end{align*}
    where we used that~$(1-1/x)^x \geq 1/2e$ for~$x \geq 2$ and that~$d \geq k \geq 2$.
    For~$d = 2\xi k$, the above expression is simply~$\xi / e^2 \zeta$ and for~$\xi \geq e^4 \zeta^2 c > 2e^2 \zeta$
    it is thus at least two. It is easy to verify that this holds also true for all~$d \geq 2\xi k$,
    and we can bound the tail of the sum by twice its first summand:
    \[
        \sum_{d = 2\xi k}^{\Delta k} \frac{d^{2\xi k} (e \zeta k)^d}{d^d e^{\zeta k}}
        \leq 2 \frac{(2\xi k)^{2\xi k} (e \zeta k)^{2\xi k}}{(2\xi k)^{2\xi k} e^{\zeta k}}
        \leq 2 \Big( \frac{(e \zeta)^{2\xi}}{e^\zeta} \Big)^k k^{2\xi k}
        \leq 2 (e \zeta)^{2\xi k} \cdot k^{2\xi k}.
    \]
    To find a bound for the sum in the range~$\zeta k \leq d < 2\xi k$,
    let us express~$d$ as~$d = \tau k$ with $\zeta \leq \tau < 2\xi$. Then a single term
    of the sum has the form
    \[
      \frac{(\tau k)^{2\xi k} (e \zeta k)^{\tau k}}{(\tau k)^{\tau k} e^{\zeta k}}
      = \Big( \frac{\tau^{2\xi} k^{2\xi - \tau} (e \zeta k)^\tau}{\tau^\tau e^\zeta} \Big)^k
      = \Big( \frac{\tau^{2\xi} (e \zeta)^\tau}{\tau^\tau e^\zeta} \Big)^k k^{2\xi k},
    \]
    which again attains its maximum at~$2 \xi k$. Therefore the first part of the sum is
    bounded by
    \[
        \sum_{d = \zeta k}^{2\xi k} \frac{d^{2\xi k} (e \zeta k)^d}{d^d e^{\zeta k}}
        \leq (2\xi - \zeta) k \cdot \Big( \frac{(e \zeta)^{2\xi}}{e^\zeta} \Big)^k  \cdot k^{2\xi k}
        \leq 2\xi k (e \zeta)^{2\xi k}  \cdot k^{2\xi k}
    \]
    and we arrive at the claimed bound by adding the two bounds and bounding
    $(2\xi k + 2)$ by~$3\xi k$ assuming~$\xi k \geq 2$.
  \end{proof}

  \noindent
  We combine the bound on the two inner sums to
  bound the probability that a dense subgraph of~$k \leq n/c$ vertices exists by
  \begin{align*}
      &\phantom{{}\leq{}}
        \sum_{k = 2\xi}^{n/c}
          \Big( \frac{e^2 c}{2\xi}
                 \cdot \frac{k}{n} \Big)^{\xi k}
          \frac{n^k}{k^{(2\xi+1) k}} \cdot
          4\xi k (e \zeta)^{2\xi k} k^{2\xi k}  \\
      &\leq
        \sum_{k = 2\xi}^{n/c}
          4\xi k
          \Big( \frac{e^4 \zeta^2 c}{2\xi} \Big)^{\xi k}
          \cdot
          \Big( \frac{k}{n} \Big)^{\xi k}
          \frac{n^k}{k^k}_{\,\FullStop}
  \end{align*}
  For~$\xi \geq e^4 \zeta^2 c$, the term~$
    4\xi k \big( \frac{e^4 \zeta^2 c}{2\xi} \big)^{\xi k}
  $
  is smaller than one, therefore we are left to bound the sum
  \[
    \sum_{k = 2\xi}^{n/c} \Big( \frac{k}{n} \Big)^{\xi k} \frac{n^k}{k^k}.
  \]
  We prove that this sum is supergeometric by considering the ratio of two
  consecutive terms:
  \begin{align*}
    \Big( \frac{k{-}1}{n} \Big)^{\xi (k-1)} \! \frac{n^{k-1}}{(k{-}1)^{k-1}}
    \cdot \Big( \frac{n}{k} \Big)^{\xi k} \frac{k^k}{n^k}
    &= \Big(1 - \frac{1}{k} \Big)^{(\xi-1)(k-1)} \!
      \Big( \frac{n}{k} \Big)^{\xi - 1}
    \geq \Big( \frac{n}{ek} \Big)^{\xi - 1}_\FullStop
  \end{align*}
  Assuming that~$c \geq 2e$, this expression is at least two for all term and it
  follows that we can bound the whole sum by twice its first term:
  \[
    \sum_{k = 2\xi}^{n/c} \Big( \frac{k}{n} \Big)^{\xi k} \Big(\frac{n}{k}\Big)^k
    \leq 2 \Big( \frac{2\xi}{n} \Big)^{2\xi^2} \Big( \frac{n}{2\xi} \Big)^{2\xi}
    = 2 \Big( \frac{2\xi}{n} \Big)^{2(\xi^2 - \xi)}
  \]
  which is bounded by~$n^{-\xi}$ for~$n \geq 4\xi$ and~$\xi \geq 2.5$, proving the lemma.
\end{proof}

\noindent
We prove a similar bound for the clique-number of Chung--Lu random graphs with
degree distributions that have superquadratic tail-bounds:

\begin{lemma}\label{lemma:chunglu-no-clique}
  Let~$(D_n)$ be a sparse degree distribution sequence with limit~$D$ and
  tail-bound~$\lambda/d^{\alpha+1}, \alpha > 1$.
  Then for large enough~$n$ and~$\xi \geq (9\alpha+9)/(8\alpha-8)$ it holds that
  \[
    \P\!\Big[\, \omega(\chunglu(D_n)) \geq 4\sqrt{\xi (\alpha+1)/(\alpha-1)} \,\Big] \leq \frac{1}{n^\xi}.
  \]
\end{lemma}
\begin{proof}
  Again using Lemma~\ref{lemma:why-it-works}, we can bound the probability that a complete
  subgraph on~$k \geq \xi$ vertices exists by considering the probability that $k$
  randomly chosen vertices form a clique. Recall that we use the
  notation~$\bar k^2 := {k \choose 2}$ for brevity.
  Taking the union bound of
  all possible~$k$, the probability of the aforementioned event is at most
  \[
    \sum_{k = \xi}^{n} {n \choose k} \P\!\big[\,  \chunglu(D_n)[X_k]) \isom K_k \,\big],
  \]
  where~$X_k$ is a set of~$k$ vertices chosen uniformly at random. Given the weights~$d_1,\ldots,d_k$,
  of~$k$ vertices, the probability that they form a complete subgraph is
  \[
    \frac{\prod_{i=1}^k d_i^{k-1}}{(\mu n)^{\bar k^2}}
    = \Big( \frac{\prod_{i=1}^k d_i}{(\mu n)^{k/2}} \Big)^{k-1}.
  \]
  The probability in the statement of the lemma crucially  depends on the value
  of the degree-product~$\mathbb D^k := \prod_{i=1}^k D^{(i)}_n$. We condition the
  probability of a dense subgraph on~$X_k$ by the value of~$\mathbb D^k$ and
  take the union bound over all possible values and then apply
  Lemma~\ref{lemma:product-bound}:
  \begin{align*}
    &\phantom{{}\leq{}}
    \sum_{k = \xi}^{n} {n \choose k}
      \sum_{d=1}^{\Deltak} \Big( \frac{d}{(\mu n)^{k/2}} \Big)^{k-1}
        \P\!\big[\mathbb D^k = d\big] \\
    &\leq
    \sum_{k = \xi}^{n} \Big(\frac{en}{k}\Big)^k
      \sum_{d=1}^{\Deltak} \Big( \frac{d}{(\mu n)^{k/2}} \Big)^{k-1}
        \frac{\zeta' (\alpha+1)^k e^{(\alpha+1)\tau}}{(k-1)!}
            \cdot \frac{\ln^{k-1} d}{d^{(1+\alpha)}} \\
    &\leq
    \sum_{k = \xi}^{n}
        \frac{\zeta''\,(2e\alpha)^k}{k^k (k-1)!}
       \frac{n^k}{(\mu n)^{\bar k^2}}
      \sum_{d=1}^{\Deltak}
          d^{k-1-(\alpha+1)} \ln^{k-1} d,
  \end{align*}
  where~$\zeta'' = \zeta' e^{(\alpha+1)\tau}$. By bounding the inner sum by the number
  of terms times its largest term (which turns out to be the last one), we arrive at
  \begin{align*}
    &\phantom{{}\leq{}}
    \sum_{k = \xi}^{n}
        \frac{\zeta''\,(2e\alpha)^k}{k^k (k-1)!}
       \frac{n^k}{(\mu n)^{\bar k^2}}
       \Delta^{k(k-1-(\alpha+1)) +k} (k \ln \Delta)^{\mathrlap{k-1}},   \\
    &\leq
    \sum_{k = \xi}^{n}
        \frac{\zeta''\,(2e\alpha)^k}{k! \, \mu^{\bar k^2}} \cdot
       \frac{n^k \Delta^  {2\bar k^2 - \alpha k} \ln^{k-1} \Delta}{n^{\bar k^2}}.
  \end{align*}
  Note that the first factor is bounded by a constant, we therefore focus on the
  second one: if it can be bounded by~$n^{-\xi-1}$, the whole sum is bounded
  by~$n^{-\xi}$ (where we invest one factor of~$1/n$ to negate the number of
  terms~$\leq n$), as claimed. Since~$\Delta$ is bounded by~$\lambda
  n^{1/(\alpha + 1)}$ for some constant~$\lambda$, this will be the case when
  \begin{alignat*}{3}
    &&
      \bar k^2 - k - \frac{2}{\alpha+1} \bar k^2 + \frac{\alpha}{\alpha+1} k- (k-1) &\stackrel{!}{\geq} \xi+1 & \\
    &\iff&
      (1 - \frac{2}{\alpha+1}) \bar k^2  - (2 - \frac{\alpha}{\alpha+1}) k  &\stackrel{!}{\geq} \xi,
      \addtocounter{equation}{1}\tag{\theequation} \label{ineq:xi}
  \end{alignat*}
  which certainly holds true for large enough~$k$ if~$\alpha > 1$. To compute
  the threshold for~$k$ at which the above inquality holds true, consider
  quadratic inequalities of the form~$A\bar k^2 - Bk \geq C$. The solutions we
  are interested in satisfy
  \[
    k \geq \frac{A+2B}{2A} \Big(\sqrt{1 + \frac{8AC}{(A+2B)^2}} + 1\Big).
  \]
  Assuming that~$8AC \geq (A+2B)^2$, we can relax this condition to the more
  manageable form
  \[
    k \geq \frac{A+2B}{A} \Big(\sqrt{\frac{16AC}{(A+2B)^2}}\Big)
      = \frac{A+2B}{A} \frac{4\sqrt{AC}}{A+2B} = 4 \sqrt{C/A}.
  \]
  With~$A = (1 - 2/(\alpha+1))$, $B = (2 - \alpha/(\alpha+1))$, and~$C = \xi$ we
  therefore have that, assuming~$\xi \geq (A+2B)^2 / 8A =
  (9\alpha+9)/(8\alpha-8)$, every~$k$ larger than $4\sqrt{C/A} = 4\sqrt{\xi
  (\alpha+1)/(\alpha-1)}$ satisfies inequality~(\ref{ineq:xi}) and the claim
  follows.
\end{proof}

\noindent
Note that the same proof still works if we scale the degree distribution by a
small polynomial factor, a fact we will need later:

\begin{corollary}\label{cor:chunglu-scaled-no-clique}
  Let~$(D_n)$ be a sparse degree distribution with
  tail-bound~$\lambda/d^{\alpha+1}$, $\alpha > 1$ and let~$S = O(n^\beta)$
  for~$\beta < \frac{\alpha-1}{4(\alpha+1)}$. Then for large
  enough~$n$ and~$\xi \geq \frac{(5\alpha+7)^2}{16(\alpha^2-1)}$ it holds that
  \[
    \P\!\Big[\, \omega(\chunglu(SD_n))) \geq 4\sqrt{2\xi (\alpha+1)/(\alpha-1)}  \,\Big] \leq \frac{1}{n^\xi}.
  \]
\end{corollary}
\begin{proof}
  We proceed as in the proof of Lemma~\ref{lemma:chunglu-no-clique}. The
  additional factor of~$S \leq \zeta n^\beta$ in the weights of the vertices
  adds a factor of~$S^{k(k-1)} \leq (\zeta n)^{2\beta \bar k^2}$ to the sum
  and we ultimately need to show that
  \[
      \frac{n^k \Delta^{2\bar k^2 - \alpha k} \ln^{k-1} \Delta}{n^{\bar k^2}} (\zeta n)^{2\beta \bar k^2}
      \stackrel{!}{\leq} \frac{1}{n^{\xi+1}}.
  \]
  Note that we cannot take care of the factor~$\zeta^{2 \beta\bar k^2}$ by the
  lower-order factors (like~$1/k!$) in the proof of
  Lemma~\ref{lemma:chunglu-no-clique}, therefore we include it in the following
  calculation; let~$\epsilon = \log \zeta / \log n$. The above inequality
  is then equivalent to
  \[
      \frac{
        n^k (\lambda n)^{(2\bar k^2 - \alpha k)/(\alpha+1)} \ln^{k-1} \lambda n
      }{
        (\alpha+1)^{k-1}
        n^{\bar k^2}
      }
      n^{(1+\epsilon)2\beta \bar k^2}
      \stackrel{!}{\leq} \frac{1}{n^{\xi+1}}.
  \]
  We can ignore the lower-order factors ($\ln^{k-1} \lambda$, $(\alpha+1)^{k-1}$,
  \etc) in the following. To simplify the proof, we show that already
  \[
      \frac{
        n^k \cdot n^{(2\bar k^2 - \alpha k)/(\alpha+1)} n^{k-1}
      }{
        n^{\bar k^2}
      }
      n^{(1+\epsilon)2\beta \bar k^2}
      \stackrel{!}{\leq} \frac{1}{n^{\xi+1}}
  \]
  holds, where we replaced~$\log n$ by~$n$ to avoid non-elementary functions like
  the product logarithm. The above then is equivalent to showing that
  \begin{alignat*}{3}
    && \bar k^2 - k -\frac{2}{\alpha+1} \bar k^2 + \frac{\alpha}{\alpha+1} k - k + 1
        - (1+\epsilon) 2\beta \bar k^2
        \stackrel{!}{\geq} \xi+1 & \\
    &\iff&
       \Big(1 - \frac{2}{\alpha+1} - (1+\epsilon) 2\beta \Big) \bar k^2
       - (2 - \frac{\alpha}{\alpha+1}) k
        \stackrel{!}{\geq} \xi,
  \end{alignat*}
  which is true if~$(1+\epsilon) \beta < \frac{\alpha-1}{2(\alpha+1)}$ and for
  large enough~$k$. The following calculations become easier if we assume that
  $(1+\epsilon) 2\beta  \leq \frac{\alpha-1}{2(\alpha+1)}$, then the above inequality becomes
  \begin{equation}\label{ineq:xi2}
       \frac{1}{2}\Big(1 - \frac{2}{\alpha+1}\Big) \bar k^2
       - (2 - \frac{\alpha}{\alpha+1}) k
        \stackrel{!}{\geq} \xi.
  \end{equation}
  Using the formula from Lemma~\ref{lemma:chunglu-no-clique} with $A =
  \frac{1}{2} - \frac{1}{\alpha+1} = \frac{\alpha-1}{2(\alpha+1)}$, $B = 2 -
  \frac{\alpha}{\alpha+1} = \frac{\alpha+2}{\alpha+1}$, and~$C = \xi$, we obtain
  that for~$\xi \geq \frac{(5\alpha+7)^2}{16(\alpha^2-1)}$ and~$k \geq
  4\sqrt{2\xi (\alpha+1)/(\alpha-1)}$ Inequality~\ref{ineq:xi2} holds. The
  earlier condition~$(1+\epsilon) 2\beta  \leq \frac{\alpha-1}{2(\alpha+1)}$
  holds for large enough~$n$ when~$\beta < \frac{\alpha-1}{4(\alpha+1)}$, as
  claimed.
\end{proof}

%
%
\subsection{From subgraphs to shallow minors}

\noindent
The next important puzzle piece is expressed in the following lemma: the
probability that a short path exists in~$\chunglu(D_n)$ crucially depends on
the weight of its endpoints and the shape of the degree distribution's tail.
At this point the difference between a cubic and a supercubic tail becomes
visible:

\begin{lemma}\label{lemma:path-to-edge}
  Let~$(D_n)$ be a sparse degree distribution sequence with tail-bound $\lambda/d^{\alpha+1}$, $\alpha > 1$
  and maximum realizable degree~$\Delta$.
  Let~$s,t$ be vertices with weights~$d_s, d_t$ respectively. The probability that
  there exists an~$s$-$t$--path of length~$r$ in~$\chunglu(D_n)$ is, for large enough~$n$, at most
  \[
    f_\alpha(r) \frac{d_s d_t}{\mu n} g_\alpha(\Delta,r)
  \]
  for some function~$f_\alpha$ independent of~$n$ and~$g_\alpha(\Delta, r) = 1$ for~$\alpha > 2$,
  $g_2(\Delta, r) = \ln^r \! \Delta$, and~$g_\alpha(\Delta, r) = \Delta^{r(2-\alpha)} \ln^{r-1} \! \Delta$
  for~$\alpha < 2$.
\end{lemma}
\begin{proof}
  The statement obviously holds for~$r=1$, we will therefore assume~$r \geq 2$
  in the following. The probability that~$r-1$ fixed
  vertices~$v_1,\ldots,v_{r-1}$ with weights~$d_1,\ldots,d_{r-1}$ form a path
  from~$s$ to~$t$ is given by
  \[
    \frac{d_s d_t}{(\mu n)^r} \prod_{i=1}^{r-1} d^2_i = \frac{d_s d_t}{(\mu n)^r} \Big( \prod_{i=1}^{r-1} d_i \Big)^2
  \]
  where~$\mu = E[D_n]$ is a constant for~$\alpha > 1$. We take a similar
  approach to the proof of Lemma~\ref{lemma:chunglu-no-clique} and condition on
  the degree-product of the~$r-1$ vertices. With this approach, the probability
  of a path from~$s$ to~$t$ of length~$r$ is at most
  \[
    {n \choose r-1} (r-1)! \frac{d_s d_t}{(\mu n)^r} \sum_{d = 1}^{\Deltar} d^2 \P[\mathbb D = d]
  \]
  where~$\mathbb D$ is the product of~$r$ independent copies of~$D_n$. We apply
  Lemma~\ref{lemma:product-bound} to the above expression, letting~$\zeta'' :=
  \zeta'(\alpha+1) e^{(\alpha+1)\tau}$, and arrive at the upper bound
  \begin{align*}
    &\phantom{{}\leq{}}
    {n \choose r-1} (r-1)! \frac{d_s d_t}{(\mu n)^r}
      \sum_{d = 1}^{\Deltar} d^2 \frac{\zeta''}{(r-1)!} \cdot \frac{\ln^{r-1} d}{d^{1+\alpha}} \\
    &\leq
      \zeta'' \Big(\frac{en}{r-1}\Big)^{r-1} \frac{d_s d_t}{(\mu n)^r}
      \sum_{d = 1}^{\Deltar} \frac{\ln^{r-1} d}{d^{\alpha-1}} ~
    \leq
      \zeta'' \Big(\frac{e}{(r{-}1)\mu}\Big)^{r-1} \frac{d_s d_t}{\mu n}
      \sum_{d = 1}^{\Deltar} \frac{\ln^{r-1} d}{d^{\alpha-1}}.
  \end{align*}
  We apply the upper bounds derived in Lemma~\ref{lemma:log-harmonic-bound} to the
  sum (were we split at~$r^{2r}$ instead of~$(r-1)^{2(r-1)}$ for simplicity)
  and obtain
  \begin{align*} 
    &\phantom{{}\leq{}}
      \sum_{d = 1}^{\Deltar} \frac{\ln^{r-1} d}{d^{\alpha-1}}
    = \sum_{d = 1}^{r^{2r}{-}1} \frac{\ln^{r-1} d}{d^{\alpha-1}}
        + \sum_{d = r^{2r}}^{\Deltar} \frac{\ln^{r-1} d}{d^{\alpha-1}} \\
    &\leq  r^{2r} (2r)^{r-1} \ln^{r-1} r + \frac{(2r)^{r-1}}{r^{2r(\alpha - 1)}} \ln^{r-1} r + R_\alpha \\
    &\leq  2^r r^{3r}\ln^{r-1} r + R_\alpha,
  \end{align*}
  where~$R_\alpha$ is bounded by
  \begin{align*}
    R_\alpha \leq
    \begin{cases}
        \frac{\zeta 2r}{\alpha-2} r^{2r(2-\alpha)} \ln^{r-1} r                & \!\text{for}~\alpha > 2, \\[1ex]
        \ln^r \Delta                                                          & \!\text{for}~\alpha = 2,\!~\text{and} \\[1ex]
        \frac{\zeta r}{(2-\alpha)^r} \Delta^{r(2-\alpha)} \ln^{r-1} \! \Delta & \!\text{for}~1 < \alpha < 2.
    \end{cases}
  \end{align*}
  Here~$\zeta$ is the constant introduced in Lemma~\ref{lemma:log-harmonic-bound}.
  By simple gathering of terms and small simplifications we can now bound
  the probability that a path of length~$r$ between vertices~$s,t$ exists by
  \begin{align*}
     \zeta'' \Big(\frac{e}{(r{-}1)\mu}\Big)^{r-1} \frac{d_s d_t}{\mu n}
      ( 2^r r^{3r} + \frac{\zeta 2r}{\alpha-2}) \ln^{r-1} r
    &\leq f_\alpha(r) \frac{d_s d_t}{\mu n}
  \intertext{%
  when~$\alpha > 2$, bound it by
  }
      \zeta'' \Big(\frac{e}{(r{-}1)\mu}\Big)^{r-1} \frac{d_s d_t}{\mu n} \cdot 2 \ln^r \Delta
    &\leq f_2(r) \cdot \frac{d_s d_t}{\mu n} \ln^r \Delta
  \intertext{%
  when~$\alpha = 2$ (and $\Delta \geq e^{2r^3} r$), and bound it by
  }
    \zeta'' \Big(\frac{e}{(r{-}1)\mu}\Big)^{r-1} \frac{d_s d_t}{\mu n} \cdot
      \frac{2\zeta r \ln^{r-1} \! \Delta}{(2-\alpha)^r} \Delta^{r(2-\alpha)}
    &\leq f_\alpha(r) \cdot \frac{d_s d_t}{\mu n} \Delta^{r(2-\alpha)} \ln^{r-1} \! \Delta
  \end{align*}
  when~$1 < \alpha < 2$ (and $\Delta \geq 4r^{3/(2-\alpha)}$) for some function~$f_\alpha(r)$ that
  is independent of~$n$.
\end{proof}

\noindent
Note that Lemma~\ref{lemma:path-to-edge} can be applied even if already up
to~$n/2$ weights have been uncovered by applying Lemma~\ref{lemma:no-replacement}.

We now pair Lemma~\ref{lemma:no-dense-subgraph} and
Proposition~\ref{prop:BoundedExpChar} to bound the probability that a dense
shallor minor appears. We first formulate a property (in a sense a weak form
of coupling) of random graph models that implies bounded expansion. For
simplicity, we define~$\mathcal E^X_{\xi,r}$ as the event that a random graph
contains an $r$-shallow topological minor with nails~$X$ of density at
least~$\xi$.

\begin{lemma}\label{lemma:rnd-bnd-exp}
  Let~$\rndmodel(n)$ be a random graph model with the following property:
  for every~$r$ there exists a sparse degree distribution~$(D_n)$ with tail-bound~$h(d) =
  \lambda/d^{\alpha+1}, \alpha > 1$ such that for every~$\xi \geq 2 e^6 \zeta^2$
  (where~$\zeta$ is the constant from Lemma~\ref{lemma:log-harmonic-bound}) it holds that
  \begin{align*}
    \P[ \mathcal E^X_{\xi,r} ] &\leq \P[\,\topgrad_0( \chunglu(D_n)[X] ) \geq \xi\,],
  \end{align*}
  where~$X$ is a random set of at most~$n / 2e(r\xi+1)$ vertices.
  Then~$\rndmodel(n)$ has bounded expansion with high probability.
\end{lemma}
\begin{proof}
  Note that if the event~$\mathcal E^X_{\xi, r}$ occurs, it already occurs
  in a subgraph of size~$|X| + r\xi|X|$. Therefore the maximal size of~$X$
  that needs to be considered in order to apply Proposition~\ref{prop:BoundedExpChar} is
  \[
    |X| + r\xi |X| \leq \frac{n}{2e} \iff |X| \leq \frac{n}{2e (r\xi+1)} \leq \frac{n}{2e}.
  \]
  Exchanging the probability~$\P[\mathcal E^X_{\xi,r}]$ by
  $\P[\topgrad_0( \chunglu(n)[X] ) \geq \xi ]$ in the proof of
  Lemma~\ref{lemma:no-dense-subgraph} immediately shows that
  \[
    \sum_{k = 2\xi}^{n/ 2e (r\xi+1)} {n \choose k} \P[\mathcal E^X_{\xi,r} ]
    \leq \sum_{k = 2\xi}^{n / 2e} {n \choose k} \P[\mathcal E^X_{\xi,r} ]
    \leq \frac{1}{n^\xi}
  \]
  for suitably large~$n$. By Lemma~\ref{lemma:why-it-works}, therefore the
  probability that any set of at most~$n/2e$ vertices form the nails of
  a dense $r$-shallow minor is at most~$n^{-\xi}$. Accordingly,
  setting~$\fnabla = \xi$ and~$\fH = 1/2e$, the second condition of
  Proposition~\ref{prop:BoundedExpChar} holds with probability at least~$1-n^{-\xi}$.
  It is left to show that functions~$\fthresh, \fdeg$ exist.

  \begin{claim}
    Let~$(f_n)$ be the probability mass functions and~$D$ the limit of~$(D_n)$.
    Then every graph matching~$(D_n)$, for~$n$ sufficiently large,
    satisfies Condition~1 of Proposition~\ref{prop:BoundedExpChar}.
  \end{claim}

  \noindent
  Recall that Condition 1 states that there exist functions~$\fthresh, \fdeg$
  such that for all~$\epsilon$ we either  have~$|G| \leq \fthresh(\epsilon)$ or
  it holds that
  \[
    | \{ v \in V(G) \colon \deg(v) \geq \fdeg(\epsilon) \} | \leq \epsilon \cdot |G|.
  \]
  This translates to~$(D_n)$ as follows: for every~$\epsilon > 0$ there exists
  an integer~$0 \leq d \leq n-1$ such that
  \[
    n \sum_{k=d}^{n-1} f_n(k) \leq \epsilon n \iff  \sum_{k=d}^{n-1} f_n(k) \leq \epsilon.
  \]
  We apply Markov's inequality
  to find that
  \[
    \sum_{k=d}^{n-1} f_n(k) = \P[ D_n \geq d] \leq \frac{ \E[D_n] }{d}.
  \]
  Since~$\E[D_n] \converges \E[D]$ and~$\E[D]$ is finite, the right hand side
  can be made small enough by choosing~$\fdeg(\epsilon) = d = \E[D_n] /
  \epsilon$ and~$n$ large enough. This proves the existence of appropriate
  functions~$\fthresh$ and~$\fdeg$ and the claim.

  Hence, we conclude that
  Proposition~\ref{prop:BoundedExpChar} is applicable to~$\rndmodel(n)$
  with probability at least~$(1 - n^{-\xi})$ and the claim follows.
\end{proof}

\noindent
Combining Lemma~\ref{lemma:rnd-bnd-exp} with Lemma~\ref{lemma:no-dense-subgraph}
gives us a proof for the first claim of Theorem~\ref{thm:conf-chunglu-char}
and combining Lemma~\ref{lemma:rnd-bnd-exp} with
Corollary~\ref{cor:chunglu-scaled-no-clique} gives us a proof for the
positive part of the second claim.

Having shown that supercubic tails produce structurally sparse graphs
in the Chung--Lu model, we proceed to the next range of distributions.

\subsection{The quadratic regime}

\noindent
As before, we will work towards an argument which translates
probabilities for events on shallow topological minors to events on
subgraphs. Starting at the bottom, we begin by proving
an analogue of Lemma~\ref{lemma:no-dense-subgraph} for degree distributions
that are scaled by a factor of~$\plog{n}$.

Let~$\mathcal K^X_r$ denote the event that the vertices
of~$X$ form the nails of an $(\leq r)$-subdivision of a complete graph. With
this notation, the following is derived easily:

\begin{corollary}\label{cor:rnd-nowhere-dense}
  Let~$\rndmodel(n)$ be a random graph model with the following property:
  for every~$r$ there exists a sparse degree distribution~$(D_n)$ with tail-bound~$h(d) =
  \Theta(d^{3+o(1)})$ and $\xi \geq 7$ such that
  \begin{align*}
    \P[ \mathcal K^X_r ]
     \leq \P[ \chunglu( \plog{n} D_n)[X] \isom K_{|X|}],
  \end{align*}
  where~$X$ is a random set of at most~$\xi$ vertices. Then~$\rndmodel(n)$ is
  nowhere dense with high probability.
\end{corollary}
\begin{proof}
  By assumption, the probability that a clique of size~$k \geq \xi$ appears as
  an $(\leq r)$-subdivision in~$\rndmodel(n)$ is bounded by the probability
  that a clique of size~$k$ appears in the scaled model~$\chunglu( \plog{n} D_n )$.
  We apply Corollary~\ref{cor:chunglu-scaled-no-clique} with~$\alpha = 2+o(1)$
  to bound the probability of the latter event by
  \[
    \P\big[ \omega\big(\chunglu( \plog{n} D_n)\big) \geq 4\sqrt{6\xi} \big] \leq \frac{1}{n^\xi}.
  \]
  For this application of Corollary~\ref{cor:chunglu-scaled-no-clique} we need
  that~$\xi$ is at least
  \[
    \frac{(5\alpha+7)^2}{16(\alpha^2-1)} = \frac{(17+o(1))^2}{16(3+o(1))}
    = \frac{289+o(1)}{48+o(1)} \leq 7,
  \]
  where the last inequality holds for large enough~$n$. We conclude
  that~$\rndmodel(n)$ is nowhere dense with high probability.
\end{proof}

\noindent
The above corollary will later provide the positive statement, namely, that
large shallow clique-minors are vanishingly improbable. However,
Theorem~\ref{thm:conf-chunglu-char} also states that cubic degree
distributions do \emph{not} result in graphs with bounded expansion. The
following lemma provides us with the necessary negative statement.

\begin{lemma}\label{lemma:chunglu-unbnd-exp}
  Let~$(D_n)$ be a sparse degree distribution sequence with
  lower tail-bound~$\lambda / d^3$. Then
  \[
    \topgrad_1(\chunglu(D_n)) = \Omega( \log^2 n )
  \]
  with high probability.
\end{lemma}
\begin{proof}
  Let us write~$\lambda / d^3$ for the tail-bound and let it hold
  for degree larger than the threshold~$\tau$.
  Hence the maximum realizable degree
  is~$\Delta = (\lambda n)^{1/3}$. Let~$V_h$ contain all vertices of
  weight at least~$\Delta/2$, applying Lemma~\ref{lemma:harmonic-bound} we
  find that
  \[
    |V_h| = \sum_{d = \Delta/2}^{\Delta} \frac{\lambda n}{d^3}
          \geq \frac{\lambda n}{2} \big((\Delta/2)^{-2} - \Delta^{-2} \big)
          \geq \frac{3 \lambda n}{2\Delta^2}
  \]
  and
  \[
    |V_h| = \sum_{d = \Delta/2}^{\Delta} \frac{\lambda n}{d^3}
      \leq \frac{2^3}{\Delta^3} + \frac{\lambda n}{2} \big((\Delta/2)^{-2} - \Delta^{-2} \big)
      \leq \frac{3 \lambda n}{\Delta^2},
  \]
  which holds for large enough~$n$.

  Let us further write~$V_\delta$ for the set of all vertices of weight exactly~$\delta$.
  Then the expected number of $V_\delta$-neighbors of a vertex~$x \in V_h$ is
  \[
    \E[ |N(x) \cap V_\delta| ]_{x \in V_h}
      \geq |V_\delta| \cdot \frac{\delta \Delta}{2\mu n}
      \geq \frac{\lambda n}{\delta^{3}} \cdot \frac{\delta \Delta}{2\mu n}
      = \frac{\lambda}{2\mu} \cdot \frac{\Delta}{\delta^{2}}.
  \]
  The expected number of~$V_h$-neighbors of a vertex~$y \in V_\delta$, on the other
  hand, is
  \[
    \E[ |N(y) \cap V_h| ]_{y \in V_\delta}
      \leq |V_h| \cdot \frac{\delta \Delta}{\mu n}
      \leq \frac{3 \lambda n}{\Delta^{2}}\cdot \frac{\delta \Delta}{\mu n}
      = \frac{3 \lambda}{\mu}\cdot \frac{\delta}{\Delta},
  \]
  which for~$\delta \leq \Delta / \log n$ is at most
  $
    \frac{3 \lambda}{\mu} \cdot \frac{1}{\log n}
  $.
  We apply the Chernoff-bound
  \[
    \P[ |N(y) \cap V_h| \geq (1+\sigma) \eta ] \leq \Big( \frac{e^\sigma}{(1+\sigma)^{1+\sigma}} \Big)^{\eta}
  \]
  with~$\eta = \frac{3 \lambda}{\mu} \cdot \frac{1}{\log n}$ and
  $(1+\sigma)\eta = 1$. The latter implies that~$1+\sigma = 1/\eta$ and~$\sigma
  = 1/\eta - 1$. Thus for~$\eta < 1$ we have that~$\sigma > 0$ and the bound
  applies. In that case, we have that
  \[
    \P[ |N(y) \cap V_h| \geq 1 ] \leq \eta e^{\eta\delta} = \eta e^{1-\eta}
    \leq \eta e = \frac{3 \lambda}{\mu} \cdot \frac{e}{\log n}.
  \]
  By assuming~$\log n$ to be large enough, we can make the right hand side
  arbitrarily small. Thus of all the vertices in~$N(x) \cap V_\delta$ for
  any~$\delta \leq \Delta / \log n$ and~$x \in V_h$, we expect that at most,
  say, half of them have neighbors other than~$x$ in~$V_h$. Since these events
  are independent, this will occur with high probability for every~$x \in
  V_\delta$ for large enough~$n$. For~$x \in V_h$ let therefore~$S_x \subseteq
  N(x)$ contain all neighbors of~$x$ that a) have weight at most~$\Delta/\log n$
  and b) are not connected to any other vertex in~$V_h$. By the above arguments,
  we have that
  \begin{align*}
    \sum_{u \in S_x} \delta_u
      &\geq \sum_{\delta = \tau}^{\Delta/\log n} \delta \cdot |N(x) \cap V_\delta|
       \geq \sum_{\delta = \tau}^{\Delta/\log n} \delta \cdot \frac{\lambda}{4\mu} \cdot \frac{\Delta}{\delta^2} \\
      &\geq \frac{\lambda}{4\mu} \Delta  \sum_{\delta = \tau}^{\Delta/\log n} \frac{1}{\delta}
       \geq \frac{\lambda}{4\mu} \Delta (\ln(\Delta/\log n) - \ln \tau) \\
       &\geq \frac{\lambda \ln 2}{4\mu} \Delta (\log \Delta - \log\log n - \log \tau)
       \geq \frac{\lambda}{16\mu} \Delta \log \Delta, \label{eq:delta-sum}\tag{$\star$}
  \end{align*}
  where the last inequality holds when~$n$ is large enough such
  that $\frac{1}{2}\log \Delta \geq \log \! \log n - \log \tau$.

  Now consider two sets~$S_x, S_z$ for distinct~$x,z \in V_h$. By the above, we
  may assume that they both have a total weight of at
  least~$\frac{\lambda}{16\mu} \Delta \log \Delta$. Therefore the expected
  number of edges between~$S_x$ and~$S_z$ is at least
  \[
    \sum_{u \in S_x} \sum_{v \in S_z} \frac{\delta_u \delta_v}{\mu n}
    \geq \Big( \frac{\lambda}{16\mu} \Delta \log \Delta \Big)^2 \frac{1}{\mu n}.
  \]
  Consider the graph~$H$ on vertices~$V_h$ obtained by contracting each
  set~$S_x$ into the respective vertex~$x \in V_h$. In expectation we then have
  \[
    \|H\| \geq {|V_h| \choose 2} \Big( \frac{\lambda}{16\mu} \Delta \log \Delta \Big)^2 \frac{1}{\mu n}
  \]
  and thus we expect the density of~$H$ to be tightly concentrated around
  \begin{align*}
    \frac{\|H\|}{|H|} &\geq \frac{1}{2} (|V_h|-1)  \Big( \frac{\lambda}{16\mu} \Delta \log \Delta \Big)^2 \frac{1}{\mu n} \\
    &\geq \frac{1}{4} \frac{|V_h|}{\mu n} \Big( \frac{\lambda}{16\mu} \Delta \log \Delta \Big)^2
     \geq \frac{1}{4} \frac{3 \lambda n}{2\mu n \Delta^2} \Big( \frac{\lambda}{48\mu} \Delta \log \lambda n \Big)^2 \\
    &\geq \frac{\lambda^3}{6144\mu^3}\log^2 \lambda n = \Omega(\log^2 n).
  \end{align*}
  The graph~$H$ is obtain by contracting sets of radius~$1$ and thus is a
  $1$-shallow minor of~$G$. This proves the claimed lower bound
  on~$\topgrad_1(G)$.
\end{proof}

\noindent
The construction for the configuration model is the same but relies on
slightly different arguments. Central here is the argument that the
probability that a \emph{stub} $x$ will be paired with another stub~$y$
after~$t$ pairs (none of which contain~$x$ or~$y$) have already been drawn is
given by
$
  \frac{1}{\mu n - 2t - 1},
$
\ie the stub~$y$ is chosen uniformly at random among all still available
stubs. The probability~$p_{A,B}$ that two sets of stubs~$A,B$ (for example
stubs belonging to two different vertices) will be connected by a pair is
therefore bounded by
\[
  1 - \Big(1 - \frac{|B|}{\mu n - 2t - 1} \Big)^{|A|/2}
  \leq
  p_{A,B}
  \leq
  1 - \Big(1 - \frac{|B|}{\mu n - 2t - |A| - 1} \Big)^{|A|} \label{eq:stub-bound} \tag{$\star\star$}
\]
if~$t$ pairs have already been drawn. We prove the following lemma
to obtain a more manageable bound:
\begin{lemma}\label{lemma:stub-pairing}
  Let~$t$ pairs have already been drawn in the random process to construct an
  instance of~$\confmodel(D_n)$ and let~$A,B$ be two sets of stubs that have not
  yet been paired with~$|A| \leq |B|$. Then the probability~$p_{A,B}$ that
  there exist stubs~$a
  \in A, b \in B$ such that~$a$ will be paired with $b$ in the remainder of the
  process is bounded by
  \[
    \frac{1}{4} \frac{|A||B|}{\mu n - 2t} \leq
    p_{A,B}
    \leq
    2 \frac{|A||B|}{\mu n - 2t},
  \]
  assuming that~$|A||B| \leq \mu n - 2t$ and~$\mu n - 2t \geq 8$.
\end{lemma}
\begin{proof}
  We will use in the following that~$(1-x)^t \leq 1 - \frac{xt}{2}$
  for all~$x \leq \frac{1}{t-1}$. Applying this inequality to the second
  term of the lower bound in \eqref{eq:stub-bound}, we find that
  \[
    \Big(1 - \frac{|B|}{\mu n - 2t - 1} \Big)^{|A|/2}
    \leq 1 - \frac{|A|}{4} \frac{|B|}{\mu n - 2t - 1}
    \leq 1 - \frac{1}{4} \frac{|A||B|}{\mu n - 2t}
  \]
  which holds if
  \[
    \frac{|B|}{\mu n - 2t - 1} \leq \frac{1}{|A|-1}
    \iff
    |B|(|A|-1) \leq \mu n - 2t - 1.
  \]
  The latter is implied by the lemma's simpler condition
  that~$|A| |B| \leq \mu n - 2t$.

  For the other direction, we use the well-known Bernoulli inequality
  $(1-x)^t \geq 1 - xt$ which holds for~$x \leq 1$. Applying it to the
  second term of the upper bound in~\eqref{eq:stub-bound} gives us
  \[
    \Big(1 - \frac{|B|}{\mu n - 2t - |A| - 1} \Big)^{|A|} \geq
    1 - \frac{|A||B|}{\mu n - 2t - |A| - 1}
    \geq 1 - 2 \frac{|A||B|}{\mu n - 2t},
  \]
  where the second inequality holds when
  \[
    1 - \frac{|A| + 1}{\mu n - 2t} \geq \frac{1}{2}
    \iff
    |A| \leq \frac{1}{2} (\mu n - 2t) - 1.
  \]
  For~$\mu n - 2t \geq 8$ (actually, $\mu n - 2t \geq 4+2\sqrt{3}$)
  the latter inequality is implied by the lemma's conditions
  that~$|A| |B| \leq \mu n - 2t$ and~$|A| \leq |B|$.
\end{proof}

\noindent
In particular, up to constant scaling the above probability looks very similar
to the probability that two vertices with weights~$|A|$,$|B|$ are connected in
the Chung--Lu model and this remains true even if already a constant fraction
of the stubs have been paired.

\begin{lemma}\label{lemma:config-unbnd-exp}
  Let~$(D_n)$ be a sparse degree distribution sequence with
  lower tail-bound~$\lambda / d^3$. Then
  \[
    \topgrad_1(\confmodel(D_n)) = \Omega( \log^2 n )
  \]
  with high probability.
\end{lemma}
\begin{proof}
  We proceed as in the proof of Lemma~\ref{lemma:chunglu-unbnd-exp}. As
  there, let~$\Delta$ denote the maximum realizable degree and let~$V_h$
  contain all vertices of degree at least~$\Delta/2$;
  the bounds
  \[
    \frac{3 \lambda n}{2 \Delta^2} \leq |V_h| \leq \frac{3 \lambda n}{\Delta^2}
  \]
  hold as before for large enough~$n$. We adapt the construction of the
  graph~$\confmodel(D_n)$ by first pairing all stubs from~$V_h$ (potentially
  with other stubs from~$V_h$). After these first~$\leq 2|V_h|$ pairings,
  the partially constructed graph now contains all edges that have at least
  one endpoint in~$V_h$, thus~$N(V_h)$ is known.

  At this stage, the number~$m$ of remaining stubs is bounded by
  \[
    m \geq \mu n - \frac{3\lambda n}{\Delta^2} \cdot \Delta
    = \mu n - 3 (\lambda n)^{2/3} \geq \frac{\mu n}{2},
  \]
  where the inequality holds for~$n \geq 216 \lambda^2 / \mu^3$,
  and
  \[
    m \leq \mu n - \frac{3 \lambda n}{2\Delta^2} \cdot \frac{\Delta}{2} < \mu n.
  \]
  Accordingly, the expected number of $V_\delta$-neighbors of a vertex~$x \in
  V_h$
  \[
    \E[ |N(x) \cap V_\delta| ]_{x \in V_h}
      \geq |V_\delta| \cdot \frac{\delta \Delta}{m}
      \geq \frac{\lambda n}{\delta^{3}} \cdot \frac{\delta \Delta}{\mu n}
      = \frac{\lambda}{\mu} \cdot \frac{\Delta}{\delta^{2}}.
  \]
  The expected number of~$V_h$-neighbors of a vertex~$y \in V_\delta$, on the
  other hand, is
  \[
    \E[ |N(y) \cap V_h| ]_{y \in V_\delta}
      \leq |V_h| \cdot \frac{\delta \Delta}{m}
      \leq \frac{6 \lambda n}{\Delta^{2}}\cdot \frac{\delta \Delta}{\mu n}
      = \frac{6 \lambda}{\mu}\cdot \frac{\delta}{\Delta},
  \]
  which for~$\delta \leq \Delta / \log n$ is at most
  $
    \frac{6 \lambda}{\mu} \cdot \frac{1}{\log n}
  $.
  We skip the concentration-argument laid out in the proof of
  Lemma~\ref{lemma:chunglu-unbnd-exp} since the calculations are exactly the
  same up to the change of the factor~$3$ to the factor~$6$. As in that proof,
  we can assume in the following that for~$x \in V_h$ at most a constant
  fraction (say, half) of the vertices in~$N(x) \cap V_\delta$ have neighbors
  other than~$x$ in~$V_h$. Let us choose exactly $\frac{\lambda \Delta}{4\mu
  \delta^2}$ vertices from~$N(x) \cap V_\delta$ whose only neighbor in~$V_h$
  is~$x$ and collect all such vertices, for all~$\delta \leq \Delta/\log n$, in
  a set~$S_x$.

  Note that at this stage we have only paired stubs belonging to vertices
  of~$V_h$ to other stubs. Hence for~$u \in S_x$ we know that only one stub
  of~$u$ has been paired so far and~$\delta_u - 1$ stubs are still left.
  Therefore the total number of remaining stubs that belong to vertices
  in~$S_x$ is now given by
  \[
    \sum_{u \in S_x} (\delta_u - 1)
    \geq \sum_{u \in S_x} \delta_u - \sum_{\delta = 1}^{\Delta/\log n} \frac{\lambda \Delta}{4\mu \delta^2}.
  \]
  We can use the bound \eqref{eq:delta-sum} proved in Lemma~\ref{lemma:chunglu-unbnd-exp} for
  the first sum (it differs by a factor of two because of the slight
  difference in the bounds proved above and we inherit the constraint
  that $\frac{1}{2} \log \Delta \geq \log\!\log n - \log \tau$), let us therefore focus on the second sum
  which bounds the size of~$S_x$:
  \begin{align*}
    \sum_{\delta = 1}^{\Delta/\log n} \frac{\lambda \Delta}{4\mu \delta^2}
    &= \frac{\lambda \Delta}{4 \mu} \sum_{\delta = 1}^{\Delta/\log n} \frac{1}{\delta^2}
    \leq \frac{\lambda \Delta}{4 \mu} \Big(\frac{1}{\tau} + \sum_{\delta = \tau}^{\Delta/\log n} \frac{1}{\delta^2}\Big) \\
    &\leq \frac{\lambda \Delta}{4 \mu}
        \Big(\frac{1}{\tau}
            + \frac{1}{\tau^2}
            + \frac{1}{\tau} - \frac{\log n}{\Delta}
        \Big)
      \leq \frac{3 \lambda \Delta}{4 \mu},
  \end{align*}
  which holds when~$\log n \leq \Delta = (\lambda n)^{1/3}$. We therefore will
  have at least
  \[
    \frac{\lambda}{8 \mu} \Delta \log \Delta - \frac{3 \lambda}{4 \mu} \Delta
    \geq \frac{\lambda}{16 \mu} \Delta \log \Delta
  \]
  remaining stubs belonging to vertices in~$S_x$ (where the above inequality
  holds for~$\log \Delta \geq 12$). We now proceed to construct the
  graph~$\confmodel(D_n)$ by pairing all stubs in each~$S_x$ for all~$x \in
  V_h$; notice that the total number of stubs paired this way is sublinear
  in~$\mu n$. Accordingly, the probabilities involved will look very similar to
  those in the Chung--Lu model.

  Let~$\mathcal S_u$ contain all stubs belonging to vertices in~$S_u$ and
  let~$\mathcal S = \bigcup_{u \in V_h} \mathcal S_u$. If we continue the
  pairing process by pairing all stubs in~$\mathcal S$, we can use the following
  crude lower bound on the expected number of stub-pairs with \emph{both}
  endpoints in~$\mathcal S$: note that while we are pairing the first~$|\mathcal
  S|/4$ stubs, the number of still available stubs in~$\mathcal S$ is at
  least~$|\mathcal S|/2$. Hence, the probability that any of the
  first~$|\mathcal S|/4$ stubs is paired with another stub from~$\mathcal S$ is
  at least~$|\mathcal S|/2m$, accordingly we expect at least~$|\mathcal S|^2/4m$
  stub-pairs that have both endpoints in~$\mathcal S$. Concerning the number of
  \emph{edges} created through these pairings, note that the probability of a
  self-loop or parallel edge is on the order of~$1/|V_h|$ and we can therefore
  expect a total of at least~$|\mathcal S|^2/8m$ edges for large enough~$n$
  (using the very crude bound that at most half the created edges are loops or
  parallel).

  Let us finally assemble the minor~$H$ by contracting every set~$S_x$ onto $x$
  for all~$x \in V_h$. As argued in Lemma~\ref{lemma:chunglu-unbnd-exp},
  the density of~$H$ is, with high probability, at least
  \begin{align*}
    \frac{\|H\|}{|H|}
    &\geq \frac{|\mathcal S|^2}{8m} \cdot \frac{1}{|V_h|}
     \geq \frac{(|V_h| \cdot \frac{\lambda}{16 \mu} \Delta \log \Delta)^2}{|V_h| 8\mu n}
     = \frac{|V_h| \cdot \lambda^2 \Delta^2 \log^2 \Delta}{2048 \mu^3 n}
     \geq \frac{\lambda^3 \log^2 \Delta}{1024 \mu^3},
  \end{align*}
  which is the claimed bound of~$\Omega(\log^2 \Delta) = \Omega(\log^2 n)$
  on~$\topgrad_1(\confmodel(D_n))$.
\end{proof}

\noindent
Having characterised the cubic regime, we proceed to the final range;
degree distributions with a subcubic tail.

%
\subsection{The subcubic regime}

\noindent
We now prove that a degree distribution with a tail lower-bounded by~$d^{3-\epsilon}$
for any~$\epsilon > 0$ will with high probability result in the presence of
shallow dense clique minors of arbitrary size, making the model somewhere dense in
this regime. Because we can leverage the powerful Theorem~\ref{thm:clique-subdiv-dense},
the proof is quite straightforward.

\begin{lemma}\label{lemma:chunglu-somewhere-dense}
  Let~$(D_n)$ be a sparse degree distribution sequence with a tail lower-bounded
  by~$\frac{\lambda}{d^{3-\epsilon}}$ for some~$\epsilon > 0$.
  Then~$\chunglu(D_n)$ is somewhere dense with high probability.
\end{lemma}
\begin{proof}
  We proceed analogous to the proof of Lemma~\ref{lemma:chunglu-unbnd-exp}.
  Let us write~$\lambda / d^\gamma$ for the tail-bound with~$\gamma = 3-\epsilon$,
  then the maximum realizable degree is~$\Delta = (\lambda n)^{1/\gamma}$.
  Let again~$V_h$ contain all vertices of weight at least~$\Delta/2$.
  The bounds
  \[
    \frac{3 \lambda n}{2\Delta^{\gamma-1}} \leq |V_h| \leq \frac{3 \lambda n}{\Delta^{\gamma-1}}
  \]
  established in the proof of Lemma~\ref{lemma:chunglu-unbnd-exp} still apply
  for large enough~$n$. Accordingly, we can with high probability find
  sets~$\{S_x\}_{x \in V_h}$ where a) every~$y \in S_x$ satisfies~$N(y) \cap V_h
  = \{x\}$ and b) $|S_x \cap V_\delta| \geq \frac{\lambda}{4\mu}
  \frac{\Delta}{\delta^{\gamma-1}}$. Importantly, the total weight of these
  sets~$S_x$ differs from the previous case: because~$\gamma - 2 < 1$, the
  application of Lemma~\ref{lemma:harmonic-bound} results in a
  different bound. Concretely:
  \begin{align*}
    \sum_{u \in S_x} \delta_u
      &\geq \sum_{\delta = \tau}^{\Delta/\log n} \delta \cdot |N(x) \cap V_\delta|
       \geq \frac{\lambda}{4\mu} \Delta  \sum_{\delta = \tau}^{\Delta/\log n} \frac{1}{\delta^{\gamma-2}}  \\
      &\geq \frac{\lambda}{4\mu} \frac{\Delta}{3-\gamma} ((\Delta / \log n)^{3-\gamma} - \tau^{3-\gamma} ) \label{eq:anon}\tag{$\star\star\star$} \\
      &\geq \frac{\lambda}{8\mu(3 -\gamma)}  \frac{\Delta^{4-\gamma}}{\log^{3-\gamma} n}
  \end{align*}
  holds with high probability, where \eqref{eq:anon} holds for~$(\Delta / \log
  n)^{3-\gamma} \geq 2\tau^{3-\gamma}$ and therefore holds for large enough~$n$.

  Now consider two sets~$S_x, S_z$ for distinct~$x,z \in V_h$. By the above, we may assume that they both
  have a total weight of at least~$\frac{\lambda}{8\mu(3 -\gamma)}  \frac{\Delta^{4-\gamma}}{\log^{3-\gamma} n}$. Therefore
  the expected number of edges between~$S_x$ and~$S_z$ is at least
  \[
    \sum_{u \in S_x} \sum_{v \in S_z} \frac{\delta_u \delta_v}{\mu n}
    \geq \Big( \frac{\lambda}{8\mu(3 -\gamma)}  \frac{\Delta^{4-\gamma}}{\log^{3-\gamma} n} \Big)^2 \frac{1}{\mu n}.
  \]
  The graph~$H$ on vertices~$V_h$ obtained by contracting each
  set~$S_x$ into the respective vertex~$x \in V_h$ has therefore, in expectation,
  \[
    {|V_h| \choose 2} \Big( \frac{\lambda}{8\mu(3 -\gamma)}  \frac{\Delta^{4-\gamma}}{\log^{3-\gamma} n} \Big)^2 \frac{1}{\mu n}
  \]
  edges, thus~$H$ has (with high probability) a density of at least
  \begin{align*}
    \frac{\|H\|}{|H|} &\geq \frac{1}{2} (|V_h|-1) \Big( \frac{\lambda}{8\mu(3 -\gamma)}  \frac{\Delta^{4-\gamma}}{\log^{3-\gamma} n} \Big)^2 \frac{1}{\mu n} \\
    &\geq \frac{1}{4} \frac{|V_h|}{\mu n} \Big(
            \frac{\lambda}{8\mu(3 -\gamma)}  \frac{\Delta^{4-\gamma}}{\log^{3-\gamma} n}
          \Big)^2
     \geq \frac{1}{6} \frac{\lambda n}{\mu n} \frac{1}{\Delta^{2-\epsilon}}  \Big(
            \frac{\lambda}{8\mu \epsilon}  \frac{\Delta^{1+\epsilon}}{\log^\epsilon n}
          \Big)^2 \\
    &\geq \frac{\lambda^3}{384\epsilon^2 \mu^3} \frac{\Delta^{3\epsilon}}{\log^{2\epsilon} n}
     = \Omega\Big( \frac{n^{3\epsilon/(3-\epsilon)}}{\log^{2\epsilon} n} \Big)
     = \Omega\Big( n^{\frac{\epsilon}{2(1-\epsilon/3)}} \Big).
  \end{align*}
  The graph~$H$ is obtain by contracting sets of radius~$1$ and thus is a
  $1$-shallow minor of~$G$. Because it has (strict) superlinear
  density~$\Omega(n^{\epsilon'})$ for~$\epsilon' =
  \frac{\epsilon}{2(1-\epsilon/3)}$, Theorem~\ref{thm:clique-subdiv-dense}
  applies, meaning that for every~$\ell \in \N$ there exists $n$ large enough
  such that the constructed graph~$H$ (and hence~$G$) contains~$K_\ell$ as a
  shallow minor with high probability. In other words, $\chunglu(D_n)$ is
  somewhere-dense with high probability.
\end{proof}

%
\subsection{The proof of Theorem~\ref{thm:conf-chunglu-char}}

\noindent
We begin with the proof for the Chung--Lu model since the application
of Lemma~\ref{lemma:rnd-bnd-exp} is straightforward.
Afterwards, we will show how they can be adapted
to extend the proof to the configuration model.

\begin{lemma}\label{lemma:chunglu-path-upper}
  Let~$(D_n)$ be a sparse degree-distribution whose tail is
  upper-bounded by~$h$ for degrees above~$\tau$. Let~$s,t$
  be vertices in~$\chunglu(D_n)$. Then for every~$r \in \N$ it holds that
  \[
    \P\big[ \exists P_{st} \subseteq \chunglu(D_n), |P_{st}| = r  \bigm\vert d_s, d_t\big]
     = \frac{d_s d_t}{n} O( \E[D^2_n]^{r-1} )
  \]
  and this bound still holds if up to~$n/2$ weights have been uncovered.
\end{lemma}
\begin{proof}
  Consider the probability that a path~$P_{st}$ with endpoints~$s,t$
  is realized in~$G := \chunglu(D_n)$ whose vertices have the weights
  $d_s,d_1,\ldots,d_{r-1}, d_t$:
  \[
    \P\!\big[\, P_{st} \subseteq G \mid d_s, d_t \,\big]
    = {n \choose {r-1}} \frac{d_s d_t \prod_i d^2_i }{ \mu^r n^r } \P[\mathbb D^r = (d_1,\ldots,d_{r-1})],
  \]
  where~$\mathbb D^r$ is a random $(r-1)$-tuple drawn from~$D_n$ without
  replacement (since we condition on~$d_s,d_t$ we only draw the inner~$r-1$
  weights of~$P_{st}$). Since, by assumption, only a constant fraction of the
  vertex weights been uncovered, instead of  the weights without replacement,
  we can use~$\hat D_n$ from
  Lemma~\ref{lemma:no-replacement} to sample the weights~$d_i$ independently.
  Recall that $(\hat D_n)$ has the same tail-bound~$h$ as~$(D_n)$ and note that
  \[
    \E[D_n^2]
    = \sum_{d = 1}^{\tau-1} \P[D_n = d] \cdot d^2 + \sum_{d = \tau}^{\Delta} \frac{d^2}{h(d)}
    = \Theta( \E[\hat D_n^2] ).
  \]
  Let now $\hat D_{1,n}, \ldots, \hat D_{r-1,n}$ be independent
  copies of~$\hat D_n$ used to sample the weights~$d_1,\ldots,d_{r-1}$.
  Taking the union-bound over all possible weights, we have that
  \begin{align*}
    \prod_i d^2_i \cdot \P\!\big[\,\bigwedge_i \hat D_{i,n} = d_i\big]
    &\leq \sum_{d_1,\ldots,d_{r-1}} \prod_i d^2_i \cdot \P[\hat D_{i,n} = d_i] \\
    &= \prod_{d_1,\ldots,d_{r-1}} \sum_i d^2_i \cdot \P[\hat D_{i,n} = d_i]
     = \prod_{d_1,\ldots,d_{r-1}} \E[\hat D^2_{i,n}] \\
    &= \Theta( \E[ D^2_n ]^{r-1} ).
  \end{align*}
  We arrive at the upper bound
  \[
    \P[P_{st} \subseteq G \mid d_s, d_t,F ]
    \leq {n\choose r-1} \frac{d_s d_t }{ \mu^r n^r }  \Theta( \E[D^2_n]^{r-1 })
    = \frac{d_s d_t }{n}  O( \E[D^2_n]^{r-1 }),
  \]
  where we used that~$\mu$ is a constant.
\end{proof}

\noindent
We finally have all the ingredients for the main proof.

\begin{proof}[Proof of Theorem~\ref{thm:conf-chunglu-char} for~$\chunglu(D_n)$]
  First consider a sparse degree distribution sequence~$(D_n)$ with
  tail-bound~$h(d) = \Omega(d^{3+\epsilon})$ for some~$\epsilon > 0$.
  By Lemma~\ref{lemma:chunglu-path-upper}, the probability of an~$s$-$t$-path
  of length~$\leq r$ existing in~$G := \chunglu(D_n)$ is
  \begin{align*}
    \P[\exists P_{st} \subseteq G ]
      &\leq \sum_{r'=1}^{r} \frac{d_s d_t}{n} O(\E[D^2_n]^{r'-1})
       = \frac{d_s d_t}{n} O(\E[D^2_n]^{r'-1})  \\
      &= \frac{1}{n} \big(d_s O(\sqrt{\E[D^2_n]^{r-1}}) \cdot  d_t O(\sqrt{\E[D^2_n]^{r-1}}) \big)
  \end{align*}
  and we can interpret the right-hand side as the probability that an edge
  exists between $s$, $t$ in a Chung--Lu graph with scaled
  distribution~$O(\sqrt{\E[D^2_n]^{r-1}}) D_n$. Because $h(d)$ is supercubic,
  $\E[D^2_n]$ is a constant and so is the scaling factor~$c_r := O(\sqrt{\E[D^2_n]^{r-1}})$.
  Accordingly,
  \begin{align*}
    \P[\exists P_{st} \subseteq G ] &\leq \P[st \in \chunglu(c_rD_n) ]
  \end{align*}
  and this relation is still true if conditioned by the knowledge of up to
  $n/2$ vertex-weights.

  Let~$\mathcal E^X_{\xi,r}$ denote the event that an $(\leq r)$-subdivision
  of density at least~$\xi$ with nails~$X$ exists.
  Then for any random set~$X$ of at most~$n/2$ vertices we have that
  \[
    \P[ \mathcal E^X_{\xi,r} ]_{\chunglu(D_n)}
    \leq \P\!\big[ \grad_0( \chunglu(c_r D_n)[X] )) \geq \xi \,\big]
    = \P[ \mathcal E^X_{\xi,r} ]_{\chunglu(c_r D_n)} .
  \]
  Thus by Lemma~\ref{lemma:rnd-bnd-exp} and the fact that~$c_r D_n$ is sparse and
  has the same tail-bound as~$D_n$, the model~$\chunglu(D_n)$ has
  bounded expansion with high probability.

  Next, assume~$(D_n)$ has a tail~$h(d) = \Theta(d^3)$ and hence
  $\E[D^2_n] = \Theta(\log n)$.
  Applying Lemma~\ref{lemma:chunglu-path-upper}, the probability of
  an~$s$-$t$-path of length~$r$ existing in~$G := \chunglu(D_n)$ is therefore
  \begin{align*}
    \P[\exists P_{st} \subseteq G ]
      &\leq \P[st \in \chunglu( \Theta\big( \sqrt{ \E[D^2_n]^{r-1} } \big) D_n ) ] \\
      &\leq \P[st \in \chunglu( \Theta( \plog{n} D_n ) ) ]
  \end{align*}
  and this relation is still true if conditioned by the knowledge of up to~$n/2$
  vertex-weights. Let again~$\mathcal K^X_r$ denote the event that an
  $(\leq r)$-subdivision
  of a complete subgraph with nails~$X$ exists in~$G$. Since the graph~$G$ is
  sparse with high probability, we focus on the case~$r \geq 1$. Now
  for any random set~$X$ of at most~$\sqrt{n/2r}$ vertices we have that
  \begin{align*}
    \P[ \mathcal K^X_r ]_{\chunglu(D_n)}
    &\leq \P\!\big[ \chunglu( \plog{n} D_n)[X] ) \isom K_{|X|} \big] \\
    &= \P[ \mathcal K^X_r ]_{\chunglu(\plog{n} D_n)}.
  \end{align*}
  and thus by Corollary~\ref{cor:rnd-nowhere-dense} it follows
  that~$\chunglu(D_n)$ is nowhere dense with high probability.
  By Lemma~\ref{lemma:chunglu-unbnd-exp}, we further have that already the
  measure~$\topgrad_1(G)$ grows at a rate of at least~$\Omega(\log^2 n)$,
  hence~$\chunglu(D_n)$ has unbounded expansion.

  Finally, assume~$(D_n)$ has a tail-bound~$h(d) = O(d^{3-\epsilon})$ for some
  $\epsilon > 0$. By Lemma~\ref{lemma:chunglu-somewhere-dense} we already have
  that~$\chunglu(D_n)$ is somewhere dense with high probability.
\end{proof}

\noindent
This proof can be extended to the configuration model, the main difficulty here
is that edges are not sampled independently of each other. We first prove
a variant of Lemma~\ref{lemma:chunglu-path-upper}. The bound proved here
crucially depends on the number of \emph{unmatched} stubs: recall that,
instead of matching up stubs by choosing a random matching, we can match them up
pair-by-pair (\cf beginning of Section~\ref{sec:chung-lu-conf}).
From this perspective we can stop the process at any point and
express the probabilities at this stage in terms of the remaining number of stubs.

\begin{lemma}\label{lemma:conf-path-upper}
  Let~$(D_n)$ be a sparse degree-distribution whose tail is bounded
  by~$h$ for degrees above~$\tau$. Let~$s,t$
  be vertices in~$\confmodel(D_n)$. Then for every~$r \in \N$ it holds that
  \[
    \P[ \exists P_{st} \subseteq \confmodel(D_n), |P_{st}| = r  \mid d_s, d_t]
     = \frac{d_s d_t}{m} O( \E[D^2_n]^{r-1} ).
  \]
  where~$m$ is the number of unmatched stubs.
\end{lemma}
\begin{proof}
  Let~$G := \confmodel(D_n)$.
  By~$M(n) := (n-1)!!$ we denote the number of matchings on~$n$ vertices, where~$!!$
  denotes the double factorial:
  \[
    n!! := \begin{cases}
      n \cdot (n-2) \cdot \ldots \cdot 5 \cdot 3 \cdot 1 & \text{for $n > 0$ odd,} \\
      n \cdot (n-2) \cdot \ldots \cdot 6 \cdot 4 \cdot 2 & \text{for $n > 0$ even, and} \\
      1 & n \in \{0,-1\}.
    \end{cases}
  \]
  We will need the following bound for~$k < n$:
  \begin{gather*}
    \frac{M(n-k)}{M(n)}
    \leq \Big( \frac{ (2e)^k (n-k)^{n-k}   }{ n^n } \Big)^{1/2}
    \leq \Big( \frac{ 2e }{ n } \Big)^{k/2}.
  \end{gather*}

  \noindent
  The number of
  available stubs decreases with each edge added to the graph and hence the
  probability of an edge crucially depends on the number~$m$ of \emph{remaining}
  stubs.

  Fix a path~$P_{st}$ of length~$r$ and let~$d_1,\ldots,d_{r-1}$ denote the weights
  of its internal vertices.  The probability of this path existing
  in~$G$, conditioned on the weights of its endpoints, is bounded by
  \begin{align*}
    \P[P_{st} \subseteq G \mid d_s, d_t]
      &\leq d_s d_t \frac{M(m-2r)}{M(m)}
                  \sum_{d_1,\ldots,d_{r-1}} \prod_{i=1}^{r-1} d^2_i \P[\hat D_n = d_i] \\
      &\leq \frac{ d_s d_t }{m^{r}} O( \E[D^2_n]^{r-1} ).
    \intertext{%
      Therefore the probability that \emph{some}~$s$-$t$-path of length~$r$ exists is
    }
    \P[\exists P_{st} \subseteq G \mid d_s, d_t]
        &\leq \frac{ d_s d_t }{m} O( \E[D^2_n]^{r-1} ),
  \end{align*}
  as claimed.
\end{proof}

\begin{proof}[Proof of Theorem~\ref{thm:conf-chunglu-char} for~$\confmodel(D_n)$]
  By Lemma~\ref{lemma:conf-path-upper}, the probability of an~$s$-$t$-path
  of length~$r$ existing in~$G := \confmodel(D_n)$ is
  \[
    \P[ \exists P_{st} \subseteq \confmodel(D_n), |P_{st}| = r  \mid d_s, d_t]
     = \frac{d_s d_t}{m} O( \E[D^2_n]^{r-1} ).
  \]
  Note that this probability looks almost identical to the one given by
  Lemma~\ref{lemma:chunglu-path-upper}, provided that the number of
  remaining stubs~$m$ is~$\Theta(n)$. Since we want to estimate the probability of
  the event~$\mathcal E^X_{r,\xi}$, only up to~$r\xi|X|$ edges need to
  be considered at once; meaning that at least
  \[
    m - 2r\xi|X| \geq 2\mu n - 2r\xi n / 4e(r\xi + 1) \geq (\mu - 1) 2n
  \]
  stubs remain (where~$\mu = \E[D_n]$). As in the proof for the Chung--Lu
  model, we have that
  \[
    \P[\mathcal E^X_{r,\xi}]_{\confmodel(D_n)} \leq \P[\mathcal E^X_{r,\xi}]_{\chunglu(\Theta(D_n))}
  \]
  and we conclude that~$\confmodel(D_n)$ has bounded
  expansion for distributions with tail-bound~$\Omega(d^{3+\epsilon})$.
  Similarly, the event~$\mathcal K^X_r$ for any set of vertices~$|X| \leq \sqrt{n / 2r}$ concerns
  at most~$n / 2$ edges and hence the number of stubs left is~$m = \Theta(n)$. Thus
  \[
     \P[\mathcal K^X_r]_{\confmodel(D_n)} \leq \P[ K^X_r ]_{\chunglu(\Theta(D_n))}
  \]
  and therefore~$\confmodel(D_n)$ is nowhere dense for distributions
  with a tail that is in~$\Theta(d^{3})$.
  The lower bounds provided by
  Lemma~\ref{lemma:chunglu-unbnd-exp} and Lemma~\ref{lemma:chunglu-somewhere-dense}
  can be easily adapted in a similar way to apply to the configuration model.
\end{proof}

%
%
\subsection{Perturbed bounded-degree graphs} \label{subsec:genErdosRenyi}

\noindent
An interesting application of the well-understood \Erdos-\Renyi random graphs is
to generate a \emph{perturbation} of some $n$-vertex base graph~$G^\star$.
We will use the notation~$G = G^\star + G(n,\mu/n)$ to denote the
graph obtained from~$G^\star$ by adding every possible edge not
already contained in~$G^\star$ independently with probability~$\mu/n$.
We also allow the graph~$G^\star$ to be random; in that case,
$G^\star + G(n,\mu/n)$ denotes the random graph process of drawing
a graph of size~$n$ according to~$G^\star$ and then adding the perturbation
edges as above.

This procedure is more flexible than many existing generalizations of the
\ErdosRenyi model, like (sparse)
\emph{inhomogeneous random graphs}\cite{bollobas2007phase} or (sparse)
\emph{generalized random graphs}\cite{New03b,grg2}. Similar
(\eg~\cite{similarmodel}) and more general (\eg~\cite{alon1995note}) models
have been defined before, yet there seems to be no consensus on a name or
notation.

In particular, uniform perturbation can be seen as the baseline for more
complicated models, like the small-world model by Kleinberg (described below),
models used in percolation theory~(\eg~\cite{LongRangePercolation}) and the
hybrid model by Chung and Lu~\cite{ChungLuHybrid} (described above). The
central question is: what graph classes are still structurally sparse after
the addition of few random edges?

We will call the graph~$G^\star$ drawn in the first step the \emph{base graph}.
In the \ErdosRenyi\ model, the base graph
$G^\star$ would be the edgeless graph and the edge probabilities $p_n$ constant
functions for all~$n$.

\begin{theorem}\label{thm:perturb-bnd-exp}
  Let~$\cal G$ be a class of bounded-degree graphs and~$\mu$ a constant.
  Let~$G^{\G}$ be a random graph model which draws graphs from
  $\mathcal G$ with an arbitrary probability distribution.
  Then the composite model~$G^{\G} + G(n,\mu/n)$ has
  bounded expansion with high probability.
\end{theorem}

\noindent
Note that this theorem in particular applies to~$G(n,\mu/n)$ itself.
The result carries over to the
\emph{stochastic block model}, if the parameters involved are small enough.
This model was first studied in mathematical sociology by Holland, Laskey, and
Leinhardt in 1983~\cite{holland1983stochastic} and extended by Wang and Wong
to directed graphs~\cite{wang1987stochastic}. We will supplement
the above result by demonstrating that there exist very sparse graph classes of
unbounded degree for which such a perturbation results in dense clique minors.

The following technical lemma subsumes Theorem~\ref{thm:perturb-bnd-exp}.
For a fixed graph~$G$, let~$D_{r,G}$ be a random variable which describes
the size of the $r$-th neighbourhood $|N^r(x)|$ for a uniformly chosen random
vertex~$x \in G$. The distribution of~$D_{r,G}$, given by
\[
  \Pr[ D_{r,G} = d ] = \frac{  |\{ x \in G : |N^r(x)| = d \}|  }{|G|},
\]
is an important factor in whether graph classes maintain bounded expansion under perturbation:

\begin{lemma}
  Let~$\cal G$ be a class of graphs with the following properties:
  \begin{itemize}
      \item $\cal G$ has bounded expansion, and
      \item for~$G \in \cal G$ and every~$r \in \N$ the distribution of~$\mathcal N^r$
            has a tail-bound~$h$ with~$h(d)  = \Omega(d^{3+\epsilon})$ for some~$\epsilon > 0$.
  \end{itemize}
  Let~$G^{\G}$ be a random graph model which draws graphs from
  $\mathcal G$ with an arbitrary probability distribution.
  Then~$G^{\G}(n) + G(n,\mu/n)$ has bounded expansion with high probability.
\end{lemma}
\begin{proof}
  Let~$G_\nabla \in \cal G$, $\widetilde G = G(n,\mu/n)$ and let~$G = G_\nabla + \widetilde G$.
  Assume~$H$ is an $r$-shallow topological minor of~$G$ and consider
  an embedding~$\phi_V$, $\phi_E$ of~$H$ witnessing this fact.
  Since~$\topgrad_r(G_\nabla)$ is a constant,
  most of $H$'s density must depend on random edges, \ie
  \[
    | \{ e \in H \mid \phi_E(e) \cap E(\widetilde G) = \emptyset \} | \leq \topgrad_r(G_\nabla) |H|.
  \]
  Therefore it suffices to bound the density of topological minors
  whose embedding use at least one edge of~$\widetilde G$ for each edge
  of the minor. Consider a path~$P$ of length~$r$ in~$G$ that uses at least one
  edge of~$\widetilde G$: each component of~$P \setminus E(\widetilde G)$
  is contained in a subgraph~$G_\nabla[N^r(v)]$ for some vertex~$v$. Let
  $N_1, N_2, \ldots, N_p$ be these subgraphs of the path~$P$: then we can bound
  the probability that~$P$ exists by considering the probability that
  there exist at least one edge between~$N_i$ and $N_{i+1}$ in~$\widetilde G$,
  for~$1 \leq i \leq p-1$.

  Since the probability that two $r$-neighbourhoods~$N^r(u)$, $N^r(v)$ in~$G_\nabla$
  are connected by an edge in~$\widetilde G$ is at most
  \[
    \frac{\mu |N^r(u)| |N^r(v)|}{n}
  \]
  we can stochastically bound the occurrence of $r$-paths in~$G$ by the occurrence of
  edges in~$\chunglu( D_{r,G})$. Hence we have that
  \[
    \topgrad_r( G ) - \topgrad_r(G_\nabla) \leq \topgrad_r( \chunglu( D_{r,G} ) )
  \]
  in the stochastic sense. Since the latter has bounded expansion with
  high probability by Theorem~\ref{thm:conf-chunglu-char}, we conclude
  that~$\cal G + G(n,\mu/n)$ does as well.
\end{proof}

\noindent
The above result poses the question: are there structurally sparse classes
which do \emph{not} stay sparse under perturbation? The answer is
yes: consider the class of graphs consisting of~$\Theta(\sqrt n)$
copies of~$S_{\sqrt n}$. The probability that two such stars
will be connected by a randomly added edge is lower-bounded by
some constant, hence the minor obtained by contracting the former
stars has density~$\Theta(n)$ (while only having $\sqrt n$ vertices).
Hence, the perturbed class is actually somewhere dense with high probability.

This example can be easily generalised: the presence of~$n^\alpha$
vertices~$X$ to which we can assign at least~$n^\beta$ of their
respective $r$-neighbours (not assigning any neighbour to more than
on vertex of~$X$), for some constant~$r$, will yield an
$r$-shallow minor whose density is concentrated around~$n^{3\alpha + \beta - 1}$.
Hence for all~$\alpha, \beta$ that satisfy~$3\alpha + \beta > 2$, the
perturbed class is somewhere dense with high probability.

%
%
\section{Graph Models without Bounded Expansion}\label{sec:negative-theory}

\noindent
In this section we consider the
Kleinberg~\cite{kleinberg2000navigation,Kle00} and
\Barabasi-Albert~\cite{BA99,barabasi-albert} Models, which,
respectively, were designed to replicate ``small-world'' properties
and heavy-tailed (power-law) degree distributions observed in complex
networks. We show that both these models (with typical parameters) do
not have bounded expansion, and in fact are somewhere dense
w.h.p./non-vanishing probability, respectively. This is done by
showing the existence of two/one-subdivisions of cliques respectively
in the generated graphs with a certain probability.

\subsection{The Kleinberg Model}\label{subsec:Kleinberg}
\def\KL{\mathcal G_{KL}}

\noindent
Many social networks exhibit a property that is commonly referred to as the
``small-world phenomenon.'' This property asserts that any two people in a
network are likely to be connected by a short chain of acquaintances. This was
first observed by Stanley Milgram in a study published in 1967~\cite{Mil67}.
Milgram's study suggested that individuals in a social network who only knew
the locations of their immediate acquaintances are collectively able to
construct short chains between two points in the network. More recently,
Kleinberg proposed a family of network models to explain the success of
decentralized algorithms in finding short paths in social
networks~\cite{Kle00}.

Kleinberg's model starts with a $n \times n$ grid as the base graph and allows
edges to be directed. For a universal constant $p \geq 1$, a node~$u$ has a
directed edge to every other node within lattice distance~$p$. These are the
\emph{local neighbors} of~$u$. For universal constants $q \geq 0$ and
$r \geq 0$, node~$u$ has~$q$ \emph{long range neighbors} chosen independently
at random. The $i^{\text{th}}$ directed outarc from $u$ has endpoint~$v$ with
probability $d(u,v)^{-r}/\sum_{x} d(u,x)^{-r}$.

When $r = 0$, the long-range contacts are uniformly distributed throughout the
grid, and one can show that there exist paths between every pair of nodes of
length bounded by a polynomial in $\log n$, exponentially smaller than the
number of nodes. Kleinberg shows that in this case, the expected delivery time
of \emph{every} decentralized algorithm (one that uses only local information)
is $\Omega(n^{\twothird})$. When $p = q = 1$ and $r = 2$, then short chains continue
to exist between the nodes of the network, but here is a
decentralized algorithm to transmit a message that takes $O(\log^2 n)$
time in expectation between any two randomly chosen points.

What Kleinberg's model shows is that if the long-range contacts are formed
independently of the geometry of the grid, then short chains exist between
every pair of nodes, but nodes working with local knowledge are unable to find
them. If the long-range contacts are formed by taking into account the grid
structure in a specific way, then short chains exist and nodes working with
local knowledge are able to discover them.
We show that for those parameters where greedy routing is efficient,
not only does the model not have bounded expansion, it is in fact,
\emph{somewhere dense} w.h.p.

\begin{theorem} \label{thm:Kleinberg_notBE}
  The Kleinberg model with parameters $p = q = 1$ and $r = 2$ is somewhere dense
  w.h.p.
\end{theorem}
\begin{proof}
Let $\Gamma_n$ be an $n \times n$ grid. For $p = q = 1$ and $r = 2$, the probability that
a node $u$ has $v$ as its long-range contact is proportional to $d_{\Gamma_{n}}(u, v)^{-1}$
and the normalizing factor in this case is $O(1/ \log n)$. This can be easily seen by
summing up $1/ d_{\Gamma_n}(u, x)^2$ for all $x$ and noticing that in the grid, there
are $4d$ neighbors that are at a distance of~$d$ from $u$.
\begin{equation*}
  \sum_{x} \frac{1}{d_{\Gamma_{n}}(u, x)^2} = \sum_{d = 1}^{n} \frac{4d}{d^2}
                                            \sim 4 \log n.
\end{equation*}

\noindent
To show that the model is somewhere dense, we show that 2-subdivisions
of cliques of a certain size $g(n)$ occur with high probability. Later we will see
that $g(n) = \Omega(\log \log n)$. To this end, let $\Gamma_{c \cdot g(n)}'$
denote some fixed $c \cdot g(n) \times c \cdot g(n)$ subgrid of $\Gamma_n$, where $c$ is some
constant that we will fix later. Choose $V'$ and $E'$ to be, respectively,
a set of $g(n)$ nodes and a set of $g(n)^2$ edges from the subgrid $\Gamma_{c \cdot g(n)}'$
with the following properties:
\begin{inparaenum}[(i)]
  \item the endpoints of the edges in $E'$ are different from the nodes in $V'$;
  \item no two edges in $E'$ share an endpoint.
\end{inparaenum}

Given any pair of vertices $u, v \in V'$ and an edge $e \in E'$ with endpoints $a, b$, the
probability that $a$ has $u$ as its long-range neighbor is $\Omega((d_{\Gamma_n}(a, u) \cdot \log n)^{-1})$.
Similarly, the probability that $b$ has $v$ as its long-range neighbor is
$\Omega((d_{\Gamma_n}(b, v) \cdot \log n)^{-1})$ The probability of both these events
happening is
\begin{equation}\label{eqn:prob_directed_edges}
  \frac{1}{d_{\Gamma_n}(a, u)^2 d_{\Gamma_n}(b, v)^2} \cdot \frac{1}{\log^2 n} \geq \frac{1}{c^4 g(n)^4 \log^2 n},
\end{equation}
where we upper-bounded distances $d_{\Gamma_n}(x, y)$ by $c \cdot g(n)$. Thus the probability
that there exists a $2$-subdivided $g(n)$-clique in $\Gamma_{c \cdot g(n)}'$ is at least:
\begin{equation} \label{exp:clique_prob}
  \left ( \frac{1}{c^4 g(n)^4 \log^2 n}\right )^{g(n)^2} =: f(n, c).
\end{equation}
The probability that there does \emph{not} exist a $2$-subdivided $g(n)$-clique in $\Gamma_{c \cdot g(n)}'$
is at most $1 - f(n, c)$. Hence the probability that there does not exist
a $2$-subdivided $g(n)$-clique in \emph{any} $c \cdot g(n) \times c \cdot g(n)$ subgrid is at most
\[
  (1 - f(n, c))^{\frac{n}{c^2 g(n)^2}} \leq
      \exp \left ( -\frac{n}{c^2 \cdot g(n)^2 \cdot \left ( c^4 \cdot g(n)^4 \cdot \log^2 n \right )^{g(n)^2}} \right )
      := e^{\frac{n}{h(n)}}.
\]
This follows from the inequality $( 1 - x/p )^p \leq e^{-x}$.

Choose $g(n) = \log \log n$ and $c = 3$ (actually any $c \geq 3$ works). Then it
is easy to show that $h(n) < \sqrt{n}$. Thus the probability of
a $2$-subdivided $(\log \log n)$-clique \emph{not} existing is at most $e^{- \sqrt{n}}$
(which goes to zero as $n \to \infty$), and we
conclude that the graph model is somewhere dense.
\end{proof}

\subsection{The \Barabasi-Albert model}

\noindent
The \Barabasi-Albert model uses a preferential attachment paradigm to produce
graphs with a degree distribution that mimics the heavy-tailed
distribution observed in many real-world networks~\cite{BA99}. This model uses
a random graph process that works as follows: Start with a small number $n_0$
of nodes and at every time step, add a new node and link it to $q \leq n_0$
nodes already present in the ``system.'' To model preferential attachment, we
assume that the probability with which a new node $u$ is connected to node $v$
already present in the system is proportional to the degree of $v$, so that
$\P[u \to v] = \deg (v)/ \sum_{x} \deg (x)$, where the sum in the denominator
is over all vertices $x$ that are already in the system. After $t$ time steps,
the model leads to a random network with $t + n_0$ vertices and $qt$ edges.

\Barabasi and Albert suggested that such a network evolves into one
in which the fraction~$P(d)$ of nodes of degree~$d$ is proportional
to $d^{- \gamma}$. They observed experimentally that
$\gamma = 2.9 \pm 0.1$ and suggested that $\gamma$ is actually~$3$.
This model was rigorously analyzed by \Bollobas, Riordan, Spencer and
Tusn\'{a}dy in~\cite{BRST01} who showed that it is indeed the case that
the fraction of vertices of degree~$d$ fall off as $d^{-3}$ as $d \to \infty$.
In~\cite{BR04}, \Bollobas and Riordan showed that the diameter of
the graphs generated by this model is asymptotically
$\log n / \log \log n$.

We first provide a formal restatement of the \Barabasi-Albert Model.
Note that this is slightly different from the formalization of \Bollobas
et al.~in~\cite{BRST01}. We start with a ``seed'' graph~$G_0$ with \
$n_0$ nodes $u_1, \ldots, u_{n_0}$ with degrees $d_1, \ldots, d_{n_0}$.
The number of edges in the seed graph is denoted by $m_0$.
At each time step $t = 1, 2, \ldots$,
we create a graph~$G_t$ by adding a new node $v_t$ and linking it to
$q$ nodes in $G_{t - 1}$. These $q$ nodes are picked independently
and with a probability that is proportional to their degrees in
$G_{t - 1}$. That is, we choose $u \in V(G_{t - 1})$ to link to
with probability
\[
  \P\!\big[ \{v_t, u \} \in E(G_t) \,\big] = \frac{\deg_{G_{t - 1}}(u)}{2 \cdot |E(G_{t - 1})|}
                              = \frac{\deg_{G_{t - 1}}(u)}{2(m_0 + q(t - 1))}.
\]
Note that in this model, all edges between $v_t$ and nodes of $G_{t - 1}$ are
assumed to be added simultaneously (so that the increasing degrees of nodes
which receive edges from $v_t$ do not influence the probabilities for this
time step.) In this restated version, $n_0$, $d_1, \ldots, d_{n_0}$, and $q$
are the parameters of the model.

\begin{lemma}
  Given any fixed~$r$, a graph $G_n$ generated by the preferential
  attachment model with parameters $n_0$, $d_1, \ldots, d_{n_0}$ and
  $q \geq 2$ has a $1$-subdivided $K_r$ as a subgraph with probability
  at least $(4 (\frac{m_0}{q} + r+ r^2) )^{- r^2}$, provided $n \geq r
  + r^2$.
\end{lemma}
\begin{proof}
  Choose any $r \in \N$. We will show that there exists a $1$-subdivided $K_r$
  in the graph with probability that depends only on $m_0$, $q$, and $r$.

  Consider the graph after the first $r + r^2$ time steps.
  Let $v_1, \ldots, v_{r}$ be the nodes that were added in the
  first $r$ time steps and fix two nodes $v_i, v_j$ from among these.
  The probability that a new node $v_{k}$
  (for any $r + 1 \leq k \leq r^2$) is connected to these two fixed nodes
  is at least
  \[
    \left ( \frac{q}{2(m_0 + q (r + r^2))} \right )^2 =: f(m_0, r, q),
  \]
  where the denominator is the sum of the vertex degrees after $r + r^2$ time steps.
  Now if $v_{r + 1}$ is linked to $v_1, v_2$ and $v_{r + 2}$
  is linked to $v_1, v_3$ and so on such that the nodes added after time step $r$
  connect the first $r$ nodes in a pairwise fashion, we would have a $1$-subdivided
  $K_r$ in the graph $G_{r + r^2}$. The probability of this happening
  is at least $f(m_0, r, q)^{r^2}$. Thus, for every~$r$, the probability
  of $1$-subdivided $K_r$ existing is non-zero if the graph is large enough.
\end{proof}

\noindent
It immediately follows that the \Barabasi-Albert model (and similar
preferential attachment models) is not a.a.s. nowhere dense and in particular
does not have bounded expansion a.a.s. We note that this result is more of
theoretical interest, since the probabilities involved might be small enough
to be irrelevant in practice. As such it would be worthwhile to investigate
whether the \Barabasi-Albert model is somewhere dense a.a.s.

%

\section{Experimental Evaluation}\label{sec:experiments}
\noindent
Although it has been established that real-world networks are sparse, and tend
to have low degeneracy (relative to the size of the network), it is natural to
ask whether there is empirical evidence that they satisfy the stronger
conditions of bounded expansion. Unfortunately, since bounded expansion itself
is a property of graph classes and not single graphs, it is impossible to
determine whether individual instances have bounded expansion or not. One
natural proxy would be to evaluate the \emph{grad} of these graphs,
calculating the maximum density of an $r$-shallow minor for each $r \in
\mathbf N$ (obviously stopping when $r$ is the diameter of the network), but
it is not known how to find such minors in reasonable time.

In order to get around these difficulties, we calculate upper bounds on
$\chi_{p-1}$ (the $p$-centered coloring number), a good proxy, since it and
the grad are both related to each other by factors independent of the graph
size. This is further justified by the fact that $p$-centered colorings are
directly applicable to algorithm design, where the complexity of such
algorithms depends heavily on the number of required colors (as will be shown in
Section~\ref{sec:AlgoImplications}). Since it is very time-consuming to obtain
a good $p$-centered coloring for large~$p$ (this is analogous to determining a
reasonable bound for the maximum density of an $r$-shallow minor for large
$r$), we evaluate this property for small values only. This is also roughly
the range of $p$ which is relevant to the algorithms presented later in this paper when
applied to practical settings.

To obtain upper bounds on $\chi_{p-1}$, we implemented the
transitive-fraternal augmentation procedure of \Nesetril and Ossona de
Mendez~\cite{NOdM08a}. Our theoretical results predict that graphs
generated with the configuration model for typical degree
distributions of complex networks will likely have bounded expansion.
We thus compared the results for $\chi_3$ of this procedure between
real networks and networks generated using the configuration model for
the same degree distributions. We chose $\chi_3$ since it is
relatively easy to compute for large networks but still heavily
influenced by one-subdivisions of cliques. The results of this
experiment can be found in Figure~\ref{fig:violin}.

\begin{figure}[t!]
  \centering
  \hspace*{-4.15em}\includegraphics[width=1.142\textwidth]{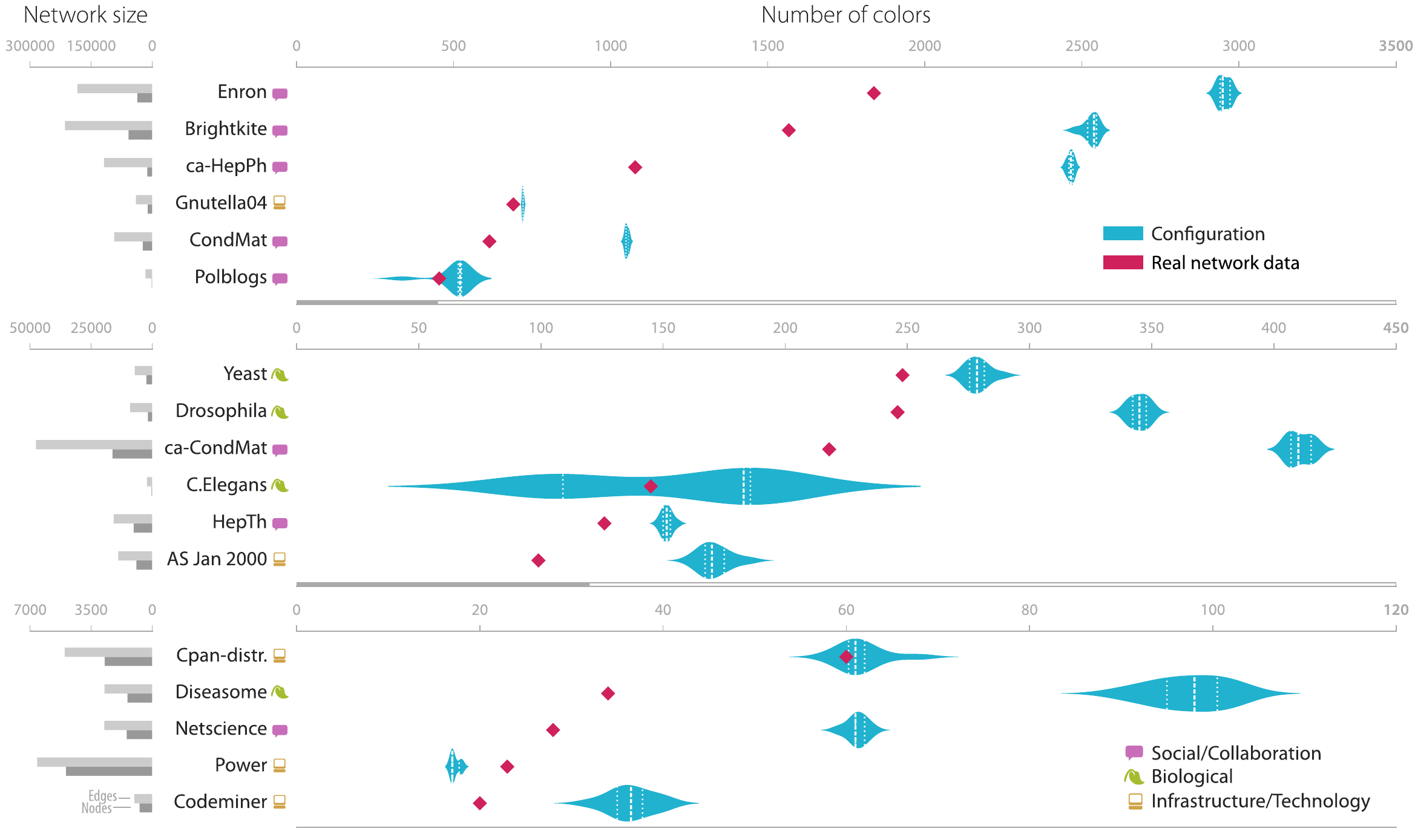}
  \caption{\footnotesize Comparison of $4$-centered coloring numbers
    on real-world networks (red diamonds) compared to synthetic graphs
    (blue violins) with the same degree distribution. Each violin
    represents 10 random instances generated with the configuration
    model, with median and quartiles marked with dashed and dotted
    lines. Networks are partitioned into three groups by size
    (indicated on the left) to enable rescaling axes. See
    Table~\ref{tab:colors-real-networks} for data
    sources.\label{fig:violin}}
\end{figure}

For almost all networks the bound for the real-world network is either
smaller or comparable to the values for the synthetic graphs. The one
exception is ``power'', which is not surprising since this network has
a relatively complex (grid-like) structure over low degree nodes.

We furthermore extended the algorithm based on transitive-fraternal
augmentations with simple heuristic improvements (e.g. giving high
degree nodes their own private color, merging color classes where
possible), ran it to find $p$-centered colorings on a small corpus of
well-known complex networks and verified the results. The results of
the best colorings we were able to achieve with this relatively simple
method can be found in Table~\ref{tab:colors-real-networks}.

\begin{table}[t!]
  \def\lim{${}^*$}
  \centering
  \newcolumntype{L}{l}
  \newcolumntype{C}{>{\sffamily}c}
  \newcolumntype{R}{>{\sffamily}r}
  \newcolumntype{X}{>{\global\let\currentrowstyle\relax}}
  \newcolumntype{^}{>{\currentrowstyle}}
  \newcommand{\rowstyle}[1]{\gdef\currentrowstyle{#1}%
    #1\ignorespaces
  }
  \renewcommand{\arraystretch}{1.15}
  {\footnotesize 

  \begin{tabular}{@{}XL^R^R^C^C^C^C^C^C@{}}
    \toprule
    \rowstyle{\rmfamily}
                                                                  &          &        & \multicolumn{6}{^C}{p}                   \\ \cmidrule(lr){4-9}
    \rowstyle{\rmfamily}
    Network                                                       & Vertices & Edges  & 2   & 3   & 4    & 5    & 6   & $\infty$ \\
    \midrule
    Karate~\cite{karate}                                           & 35       & 78     & 6  & 7   & 9    & 9    & 10  & 8    \\
    Dolphins~\cite{dolphins}                                       & 62       & 159    & 7  & 11  & 17   & 18   & 19  & 24   \\
    Lesmiserables~\cite{lesmiserable}                              & 77       & 254    & 10 & 15  & 16   & 16   & 16  & 16   \\
    Polbooks~\cite{polbooks1,polbooks2}                            & 105      & 441    & 8  & 16  & 22   & 29   & 31  & 30   \\
    Word-adjacencies~\cite{newman2006finding}                      & 112      & 425    & 8  & 18  & 27   & 35   & 41  & 48   \\
    Football~\cite{football}                                       & 115      & 613    & 9  & 22  & 33   & 49   & 62  & 69   \\
    Airlines~\cite{gephinets}                                      & 235      & 1297   & 11 & 28  & 39   & 47   & 55  & 64   \\
    Sp-data-school~\cite{spschool}                                 & 238      & 5539   & 23 & 100 & 138  & 157  & 168 & 171  \\
    C.Elegans~\cite{WSS98}                                         & 306      & 2148   &  8 & 36  & 74   & 83  & 118 & 153  \\
    Hex-grid                                                      & 331      & 930    & 3  & 9   & 20   & 21   & 25  & 69   \\
    Codeminer~\cite{gephinets}                                     & 724      & 1017   & 5  & 10  & 15   & 17   & 23  & 51   \\
    Cpan-authors~\cite{cpan}                                       & 839      & 2212   & 9  & 24  & 34   & 43   & 47  & 224  \\
    Diseasome~\cite{diseasome}                                     & 1419     & 2738   & 12 & 17  & 22   & 25   & 30  & 30   \\
    Polblogs~\cite{polblogs}                                       & 1491     & 16715  & 30 & 118 & 286 & 354  & 392 & 603  \\
    Netscience~\cite{newman2006finding}                            & 1589     & 2742   & 20 & 20  & 28   & 28   & 28  & 20   \\
    Drosophila~\cite{drosophila}                                   & 1781     & 8911   & 12 & 65  & 137  & 188  & 263 & 395  \\
    Yeast~\cite{yeast}                                             & 2284     & 6646   & 12 & 38  & 178  & 254  & 431 & 408  \\
    Cpan-distr.~\cite{cpan}                                        & 2719     & 5016   & 5  & 14  & 32   & 42   & 56  & 224  \\
    Twittercrawl~\cite{gephinets}                                  & 3656     & 154824 & 89 & 561 & 1206 & 1285 & 1341  & --   \\
    Power~\cite{WSS98}                                             & 4941     & 6594   & 6  & 12  & 20   & 21   & 34  & 95   \\
    AS Jan 2000~\cite{p2p-gnutella-4-1-as20000102}                 & 6474     & 13895  & 12 & 29  & 70  & 102  & 151 & 357  \\
    Hep-th~\cite{hep-th-cond-mat}                                  & 7610     & 15751  & 24 & 25  & 104  & 328  & 360  & 558  \\
    Gnutella04~\cite{p2p-gnutella-4-1-as20000102,p2p-gnutella-4-2} & 10876    & 39994  & 8  & 43  & 626  & --   & --  & --   \\
    ca-HepPh~\cite{p2p-gnutella-4-1-as20000102}                    & 12008    & 118489 & 239 & 296 & 1002 & --   & --  & --       \\
    CondMat~\cite{hep-th-cond-mat}                                 & 16264    & 47594  & 18 & 47  & 255  & 1839 & --  & 1310 \\
    ca-CondMat~\cite{p2p-gnutella-4-1-as20000102}                  & 23133    & 93497  & 26  & 89  & 665  & --   & --  & --       \\
    Enron~\cite{enron1,enron2}                                     & 36692    & 183831 & 27  & 214 & 1428   & --   & --  & --       \\
    Brightkite~\cite{brightkite}                                   & 58228    & 214078 & 39  & 193  & 1421   & --   & --  & --       \\
    \bottomrule
  \end{tabular}
  }
  \caption{\footnotesize Number of colors in $p$-centered colorings computed on real world networks and upper bounds
    of their respective treedepth. These networks were mostly taken from the datasets found in
    \cite{snapnets,polbooks2,gephinets}.}
  \label{tab:colors-real-networks}
\end{table}

The results show that some networks clearly have a moderately growing grad; in
particular the larger networks \emph{Netscience}, both \emph{Cpan-}networks
and \emph{Diseasome}. Other networks, like \emph{Twittercrawl}, have such
quickly growing $p$-centered coloring numbers that we did not invest the time
to determine the value for larger $p$. Since graphs of bounded crossing number
(as infrastructure networks tend to be) and bounded degree have bounded
expansion, we are not surprised at the small number of colors needed by
\emph{Power} and \emph{Hex}. Finally some networks, like \emph{CondMat} and
\emph{Hep-th}, start out reasonably well for small $p$ but show a sudden jump
at $p=3$. At present we do not know whether this is an artifact of the
procedure we use to obtain the coloring or whether the networks have indeed
dense minors from a certain depth on. As shown in
Section~\ref{sec:negative-theory}, some complex network models predict
such an occurrence already for depth at most two. The growth behavior for
small $p$ as depicted in Table~\ref{tab:colors-real-networks} might therefore
serve as a property to distinguish types of networks, meriting future
research.

In the last column, we provide upper bounds for the treedepth of
these graphs. Notice that the number of colors needed for a $p$-centered
coloring will be fewer than the treedepth of the the graph for any $p \leq n$:
Given a treedepth decomposition of a graph of depth $t$ we can color every
node by its depth in the treedepth decomposition. Since the treedepth is a
hereditary property, the graph induced by any subset of nodes will have
treedepth at most $t$. Notice that the simple coloring algorithm we used
sometimes colors the graph with more than $t$ colors. This is a good
indication that a better coloring algorithm or heuristic exists. \looseness-1

We also note that the subgraph isomorphism algorithm presented in
Section~\ref{sec:AlgoImplications} should be directly applicable to some of
these graphs, given their comparatively low treedepth.

Finally, we argue that for
practical purposes the definition of $p$-treedepth colorings is too
strict. Remember that per definition, in a $p$-centered
coloring every graphs induced on $i < p$
colors has treedepth $\leq i$. Relaxing this latter condition,
we arrive at the following variation of
Proposition~\ref{theorem:low-td-coloring}:

\begin{observation}
  Let $\mc G$ be a graph class of bounded expansion. There exists
  functions $f$ and $g$ such that for every $G \in \mc G$, $p \in \N$,
  the graph $G$ can be colored with $f(p)$ colors so that any $i < p$
  color classes induce a graph of treedepth $\leq g(i)$ in $G$.
\end{observation}

\noindent
Obviously this follows from Proposition~\ref{theorem:low-td-coloring}, taking
$g$ as the identity. However, algorithms based on $p$-treedepth colorings
would run faster if $g$ allows for a larger margin, provided we can decrease
the number of colors~$f$. This is owed to the fact that for large number of
colors, iterating through all~$(< p)$-sized subsets will be the deciding factor
in the running time.

%
\section{Algorithms}\label{sec:AlgoImplications}
\def\augm{\vec}

\noindent
In this section, we present efficient algorithms for several important problems arising
in the study complex networks which exploit bounded expansion. Note that in both cases, only the algorithm's
running time relies on the grad being small, not its correctness---in that sense, these
algorithms are oblivious to whether the input graph is sparse or not (unless, for example,
algorithms that rely on planarity)
The problems that we discuss revolve around the
themes of \emph{subgraph counting} and \emph{centrality estimation}.

Computing the frequency of small fixed pattern graphs inside a network is the
key algorithmic challenge in using \emph{network motifs} and \emph{graphlet
degree distributions} to analyze network data (both of which are described in
more detail in Section~\ref{sec:subgraphs}). We present a parameterized
algorithm for counting the number of subgraphs with at most~$h$ vertices with
a running time of $6^{h} \cdot t^{h} \cdot n$, where~$t$ is the treedepth of
the input graph. In a graph class of bounded expansion, we use this algorithm
in conjunction with $p$-centered colorings.

Another topic of interest in complex networks is estimating the relative importance
of a vertex in the network (for example, how influential
a person is inside a social network, which roads are busiest in a
road-network, or which location is most attractive for business).
The typical approach is to define/select an appropriate {\em centrality measure}
(see~\cite{KLP05} for a survey of common measures). We focus on the \emph{closeness centrality},
which was introduced by Sabidussi~\cite{Sab66}, and related extensions.
These measures are related in the sense that computing them requires
knowledge of all the pairwise distances between the vertices of the network,
which even in sparse networks takes time $O(n^2)$ to compute~\cite{Bra01}.
We introduce localized variants of the closeness-based centrality
measures and design a linear-time algorithm to compute them in
bounded expansion classes and provide experimental data that suggests
that these measures are able to recover the topmost central vertices quite well.

\subsection{Counting graphlets and subgraphs}\label{sec:subgraphs}

\noindent
In the following we highlight three domain-specific applications of
computing the frequency of small fixed pattern graphs inside a
network. In particular, the concept of \emph{network motifs}
and \emph{graphlets} has proven very useful in the area of
computational biology.

A \emph{network motif} is a subgraph (not necessarily induced and possibly
labeled) that appears with a significantly higher frequency in a real-world
network than one would expect by pure chance. Introduced in~\cite{MSI02} under
the hypothesis that such frequently occurring structures have a functional
significance, motifs have been identified in a plethora of
different domains---including protein-protein-interaction networks~\cite{AA04}, brain
networks~\cite{SK04} and electronic circuits~\cite{ILKM05}. We point the
interested reader to the surveys of Kaiser, Ribeiro and Silva~\cite{RSK09} and
Masoudi-Nejad, Schreiber, and Kashani~\cite{MSK12} for a more extensive
overview.

\emph{Graphlets} are a related concept, though their application is
in an entirely different scope. While motifs are used to identify and
explain local structure in networks, graphlets are used to `fingerprint'
them. Pr{\v z}ulj~\cite{Prz07} introduced the \emph{graphlet degree distribution}
as a way of measuring network similarity. To compute it, one enumerates
all connected graphs up to a fixed size (five in the original paper) and
computes for each vertex of the target graph how often it appears in a subgraph
isomorphic to one of those patterns. Since some graphlets exhibit higher symmetry
than others, the computation takes into account all possible automorphisms.
The degree distribution then describes for each graphlet $G_i$, how many
vertices of the target graph are contained in $0,1,2,\dots$ subgraphs
isomorphic to $G_i$---more precisely, in how many orbits of the respective
automorphism groups it appears in. Note that if the set of graphlets only contains the
single-edge graph this computation yields exactly the classical degree
distribution.

The application of this distribution is two-fold: On the
one hand, it can be used to measure similarity of multiple networks, in
particular, networks related to biological data~\cite{HSP13}.
On the other hand, the local structure around a vertex can reveal domain-specific
functions. This is the case for protein-protein interaction networks, where
local structure correlates with biological activity~\cite{MP08}. This fact
has been applied to identify cancer genes~\cite{MMG10} and construct
phylogenetic trees~\cite{KMM10}. Graphlets have further been used
in analysis of workplace dynamics~\cite{TSL13}, photo cropping~\cite{BCL13}
and DoS attack detection~\cite{PPV07}.

A third application of subgraph counting was given by Ugander \etal~\cite{UBK13}:
their empirical analysis and subsequent modeling of social networks revealed that
there is an inherent bias towards the occurrence of certain subgraphs. Thus frequencies
of small subgraphs are an important indicator for the social domain,
similar to the role of graphlet frequencies in biological networks.

In Theorem~18.9 from \cite{NOdM12} it was shown that for a graph class
of bounded expansion counting the number of satisfying assignments of
a fixed boolean query is possible in linear time on a labeled
graph. This implies that (labeled) graphlet and motif counting are
linear time on a graph classes of bounded expansion. The result is
achieved by using the algorithm presented in Lemma~17.3 in~\cite{NOdM12}
to count the number of satisfying assignments of a fixed
boolean query parameterized by treedepth. For a graphlet with $h$ nodes,
this algorithm runs in time $O(2^{h t} \cdot ht \cdot n)$, where
$t$ is the treedepth of the host graph. We provide an algorithm with a running time of $O(6^{h} \cdot t^h \cdot h^2 \cdot n)$, which implies a faster algorithm on graph classes of bounded expansion, as explained below.

The tool of choice for applying a counting algorithm designed for bounded-treedepth
graphs to a class of bounded expansion are $p$-treedepth colorings: to compute the frequency
of a given graphlet~$H$ of size $h$, we first compute a $(h+1)$-treedepth coloring of
the input graph in linear time as per Proposition~\ref{theorem:low-td-coloring}\footnote{
Since $p$-treedepth coloring is also a $p'$-treedepth coloring for~$p' \leq p$, we only
need to compute one such coloring.}.

We then enumerate all choices of $i \leq h$ colors and count the frequency of~$H$
in the graph induced by those colors: since this induced subgraph has
treedepth at most~$i$, we are able to apply the counting algorithm for
bounded-treedepth graphs (thus imagine replacing~$t$ by~$h$ in the above running
times, establishing the asymptotic improvement given by our approach in this setting). Afterwards it is a matter of simple inclusion-exclusion over the frequencies
found for each collection of~$\leq h$ colors to recover the frequency of~$H$ in the whole graph.

Central to the dynamic programming we will use to count isomorphisms
is the following notion of a \emph{$k$-pattern} which is very similar to
the well-known notion of \emph{boundaried graphs}; the main difference
being that $k$-patterns describe specific decompositions of our input
graph $H$.

\begin{definition}
  A \emph{$k$-pattern} of the graph $H$ is a triple $M = (W,X,\pi)$
  where $X \subseteq W \subseteq V(H)$, $|X| \leq k$, such that
  $W \setminus X$ has no edge into~$V(H) \setminus W$,
  and $\pi
  \colon X \to [k]$ is an injective function. We will call the set $X$
  the \emph{boundary} of $M$. For a given $k$-pattern $M$ we denote
  the underlying graph by $H[M] = H[W]$, the vertex set by $V(M) = W$,
  the boundary by $\bound(M) = X$ and the mapping by $\pi^{M}$.
\end{definition}

\noindent
We denote by $\mc P_k(H)$ the set of all $k$-patterns of $H$. Note
that every $k$-pattern $(W,X,\pi)$ is also a $(k+1)$-pattern. In the
following we denote by $|H| = |V(H)|$.

\begin{lemma}\label{lemma:number-patterns}
  Let $H$ be a graph. Then $|\mc P_k(H)| \leq 3^{|H|}\cdot k^{|H|}$.
\end{lemma}
\begin{proof}
  The vertices of $H$ can be partitioned in $3^{|H|}$ possible ways
  into boundary vertices, pattern vertices and remainder. The number
  of ways an injective mapping for a boundary of size $b \leq |H|$
  into $[k]$ can be chosen is bounded by $k^{|H|} = 2^{|H| \log k}$.
  In total the size of $\mc P_k(H)$ is always less than $3^{|H|}\cdot
  k^{|H|}$.
\end{proof}

\noindent
We will use $k$-patterns during dynamic programming via the
basic join and forget operations defined below. Respectively, these correspond to gluing
two patterns together and to demoting a boundary-vertex to a simple vertex.

\begin{definition}[$k$-pattern join]
  Let $H$ be a graph and $M_1 = (W_1,X_1,\pi_1)$, $M_2 = (W_2,X_2,\pi_2)$
  $k$-patterns of $H$.  Then the two patterns are
  \emph{compatible} if $W_1 \cap W_2 = X_1 = X_2$ and for all $v \in
  X_1$ it holds that $\pi_1(v) = \pi_2(v)$. Their \emph{join} is
  defined as the $k$-pattern $M_1 \oplus M_2 = (W_1 \cup
  W_2,X_1,\pi_1)$.
\end{definition}

\begin{definition}[$k$-pattern forget]
  Let $H$ be a graph, let $M = (W,X,\pi)$ be a $k$-pattern of $H$ and
  $i \in [k]$. Then the \emph{forget operation} is the $k$-pattern
  \begin{align*}
    M \ominus i = \begin{cases}
          (W,X \setminus \pi^{-1}(i), \pi|_{X\setminus \pi^{-1}(i) })
             & \text{if}~\pi^{-1}(i) \neq \emptyset
               \text{~and~} N_H(\pi^{-1}(i)) \subseteq W \\
          \bot
             & \text{if}~\pi^{-1}(i) \neq \emptyset
               \text{~and~} N_H(\pi^{-1}(i)) \not\subseteq W \\
          (W,X,\pi) & \text{otherwise}
    \end{cases}
  \end{align*}
\end{definition}

\noindent
Structurally, the $k$-pattern's boundaries will represent vertices from
the path of the root vertex to the currently considered vertex in the treedepth
decomposition, while the remaining vertices of the pattern represent vertices
somewhere below it. The following two notations help expressing these properties.

\begin{definition}[Subtree and root path]
  Let $T$ be a treedepth decomposition of $G$ rooted at $r \in G$ and let $v \in V(G)$ be a vertex. Then the
  \emph{subtree of $v$} is the subtree~$T_v$ of $T$ rooted at $v$.
  The \emph{root path} of $v$ is the unique path~$P_v$ from the root~$r$ to~$v$ in $T$. We let
  $P_v[i]$ denote the $i$\th\ vertex of the path (starting at the root), so that $P_v[1] = r$ and
  $P_v[\,|P_v|\,] = v$.
\end{definition}

\noindent
We can now state the main lemma. The proof contains the description of the
dynamic programming which works bottom-up on the vertices of the given
treedepth decomposition (\ie starting at the leaves and working towards
the root of the decomposition).

\begin{lemma}\label{lemma:inducedsubgraphs}
  Let $H$ be a fixed graph on~$h$ vertices. Given a graph~$G$ on~$n$
  vertices and a treedepth decomposition $T$~of height~$t$, one can
  compute the number of isomorphisms from~$H$ to induced subgraphs
  of~$G$ in time $O(6^{h} \cdot t^h \cdot h^2 \cdot n)$
  and space $O(3^{h} \cdot t^h \cdot ht \cdot \log n)$.
\end{lemma}

\begin{proof}
  We provide the following induction that easily lends itself to dynamic programming over $T$.
  Denote by $M_H = (V(H),\emptyset,\epsilon)$ the trivial $t$-pattern of $H$, here
  $\epsilon \colon \emptyset \to \emptyset$ denotes the null function.
  Consider a set of vertices $v_1,v_2,\dots,v_\ell \in G$ with a common parent $v$ in $T$
  with respective subtrees $T_{v_i}$ and root paths $P_{v_i}$ for $1 \leq i \leq \ell$. Note that
  the root paths $P_{v_1}, \dots, P_{v_\ell}$ all have the same length $k$ and share the path
  $P_v$ as a common prefix.

  Let $M_1$ be a fixed $t$-pattern of $H$. We define the mapping $\psi_v^{M_1}\colon \bound(M_1) \to P_v$
  via $\psi_v^{M_1}(u) = P_v[\pi^{M_1}(u)]$ which takes the pattern's boundary and maps it to the
  vertices of the root-path.

  For patterns $M_1$ that satisfy that for all $u \in \bound(M_1)$, $\pi^{M_1}(u) \leq l$,
  we denote by $f[v_1,\dots,v_\ell][M_1]$ the number of isomorphisms $\phi_1\colon V(M_1) \to V(G)$ such that
  \begin{enumerate}[(i)]
  \item\label{prop:path} $\phi_1|_{\bound(M_1)} = \psi_v^{M_1}$
  \item\label{prop:rest} $\phi_1(V(M_1)\setminus \bound(M_1)) \subseteq G[V(T_{v_1} \cup \cdots \cup T_{v_\ell})]$.
  \end{enumerate}

  \noindent
  In other words we charge subgraphs to patterns whose boundaries lie on the
  shared root-path $P_v$, such that the labeling of the boundary
  coincides with the numbering induced by $P_v$ while the rest of the
  pattern is contained entirely in the subtree below $v$.

  Let $r$ be the root of the treedepth decomposition. By the above
  definition, $f[r][M_H]$ counts exactly the number of isomorphisms of $H$
  into subgraphs of $G$.\looseness-1

  We will show now how we can compute $f[r][M_H]$ recursively. For a
  leaf $v \in T$ and a $t$-pattern $M_1 = (W_1,X_1,\pi_1) \in \mc P_k(H)$ we compute
  $f[v][M_1]$ as follows: Defined the value $p_v^{M_1}$ to be $1$ if
  the function $\psi\colon W_1 \to P_v$ defined as $\psi(w) = P_v[\pi_1[w]]$
  is an isomorphism from $H[W_1]$ to $G[\psi(W_1)]$ and $0$ otherwise. In particular,
  $p_v^{M_1}$ will be zero if $W_1 \neq X_1$ or $|W_1| > |P_v|$.
  Then for the leaf $v$, we compute
  $$
    f[v][M_1] = \sum_{\, M_2\ominus |P_v| = M_1} p_v^{M_2}
  $$
  where $M_2 \in \cal P_t(H)$.

  The following recursive definitions show how $f[\cdot][M_1]$ can
  be computed for all inner vertices of $T$.
  \begin{align*}
    f[v][M_1] &= \sum_{\, M_2\ominus |P_v| = M_1} f[v_1,\dots,v_\ell][M_2] \tag{forget} \\
    f[v_1,\dots,v_{j-1},v_j][M_1] &= \sum_{M_2 \oplus M_3 = M_1}
                                       f[v_1,\dots,v_{j-1}][M_2] \cdot f[v_j][M_3] \tag{join}
  \end{align*}
  where $M_2,M_3 \in \cal P_t(H)$.

  We need to prove that the table $f$ correctly reflects the number of isomorphisms to subgraphs
  satisfying properties~\ref{prop:path} and~\ref{prop:rest}.

  Consider the \emph{join}-case first: fix a pattern $M_1 \in \mc P_t(H)$.
  By induction, the entries $f[v_1,\dots,v_{j-1}][\cdot]$ and
  $f[v_j][\cdot]$ correspond to the number of isomorphisms to subgraphs that
  satisfy properties~\ref{prop:path} and~\ref{prop:rest} with the node tuples $v_1,\dots,v_{j-1}$
  and $v_j$, respectively. We need to show that $f[v_1,\dots,v_j][M_1]$
  as computed above counts the number of isomorphisms from $H[M_1]$ to subgraphs
  of $G$ such that $\phi_1|_{\bound(M_1)} = \psi_v^{M_1}$ and
  $\phi_1(V(M_1)\setminus \bound(M_1)) \subseteq G[V(T_{v_1} \cup \cdots \cup T_{v_j})]$.

  Consider the set $\Phi_1$ of all isomorphisms from $H[M_1]$ to subgraphs of $G$ satisfying
  properties~\ref{prop:path} and~\ref{prop:rest} for the vertex tuple $v_1,\dots,v_j$. For any vertex subset
  $R \subseteq V(M_1)\setminus \bound(M_1)$, define the slice $\Phi_1(R) \subseteq \Phi_1$
  as those isomorphisms $\phi$ that satisfy $\phi^{-1}(\phi(V(H)) \cap T_{v_j}) = R$.
  Let $L = (V(M_1) \setminus \bound(M_1)) \setminus R$ and define the patterns
  $M_L = (L \cup \bound(M_1),\bound(M_1),\pi^{M_1})$ and $M_R = (R \cup \bound(M_1),\bound(M_1),\pi^{M_1})$.
  Then by induction $|\Phi_1(R)| = f[v_1,\dots,v_{j-1}][M_L] \cdot f[v_j][M_R]$.
  Since $M_1 = M_L \oplus M_R$ and clearly $M_L, M_R \in \cal P_t(H)$, the sum
  computes exactly $\sum_{R \subseteq V(M_1)\setminus \bound(M_1)}|\phi_1(R)| = |\phi_1|$.

  Next, consider the \emph{forget}-case. Again, fix $M_1 \in \mc P_t(H)$ and let $u$
  be the parent of $v$ in $T$. Let $\Phi_1$ be the set of those isomorphisms from $H[M_1]$ to subgraphs of $G$ for which
  $\phi_1|_{\bound(M_1)} = \psi_u^{M_1}$ and $\phi_1(V(M_1)\setminus \bound(M_1)) \subseteq G[V(T_v)]$.
  We partition $\Phi_1$ into $\Phi_1 = \Phi_{1,v} \cup \Phi_{1,\bar v}$ where
  $\Phi_{1,v}$ contains those isomorphisms $\phi$ for which $\phi^{-1}(v) \neq \emptyset$
  and $\Phi_{1,\bar v}$ the rest. Since $|\Phi_{1,\bar v}| = f[v_1,\dots,v_\ell][M_1]$
  we focus on $\Phi_{1,v}$ in the following. For $w \in V(M_1) \setminus \bound(M_1)$, define
  $\Phi_{1,v}(w)$ as the set of those isomorphisms $\phi$ for which $\phi(w) = v$. Clearly,
  $\{\Phi_{1,v}(w)\mid w \in V(M_1) \setminus \bound(M_1)\}$ is a partition of $\Phi_{1,v}$.
  Define the pattern $M_w = (V(M_1),\bound(M_1)\cup \{w\},\pi_w^{M_1})$ where
  $\pi_w^{M_1}$ is $\pi^{M_1}$ augmented with the value $\pi_w^{M_1}(v) = |P_v|$. Note that
  by construction $M_1 = M_w \ominus |P_v|$. By induction, $|\Phi_{1,v}(w)| = f[v_1,\dots,v_\ell][M_w]$
  and therefore
  \[
    |\Phi_1| = |\Phi_{1,\bar v}| + \sum_{w \in V(M_1)\setminus\bound(M_1)} |\Phi_{1,v}(w)| =
     \sum_{M_2 \ominus |P_v|} f[v_1,\dots,v_\ell][M_2]
  \]

  \noindent
  It remains to prove the claimed
  running time.  Initialization of $f$ for a leaf takes time $O(|\mc P_t(H)| h^2)$
  since we need to test whether the
  function $\psi$ defined above is an isomorphism for each pattern $M_1 \in \mc P_t(H)$.

  For the other vertices, a forget operation can be achieved in time
  $O(|\mc P_t(H)|)$ per vertex by enumerating all $t$-patterns, performing
  the forget operation and looking up the count of the resulting pattern
  in the previous table.

  A join operation needs
  time $O(|\mc P_t(H)| \cdot h \cdot 2^h)$ per vertex, since for
  a given pattern $M_1$ those patterns $M_2,M_3$ with $M_1 = M_2 \oplus M_3$
  are uniquely determined by partitions of the set $V(M_1)\setminus\bound(M_1)$.

  In total the running time of the whole algorithm is $O(|\mc P_t(H)|
  \cdot 2^h \cdot h^2 \cdot n)$. Note that we only have to keep at
  most $O(t)$ tables in memory, each of which contains the occurrence
  of up to $|P_t(H)|$ patterns stored in numbers up to $n^h$. Thus in
  total the space complexity is $O(|\mc P_t(H)| \cdot t \cdot \log
  (n^h)) = O(|\mc P_t(H)| \cdot ht \cdot \log n)$.
\end{proof}

\noindent
To count the occurrences of $H$ as an induced subgraph instead of the
number of subgraph isomorphisms, one can simply determine the number
of automorphisms of $H$ in time $2^{O(\sqrt{h \log h})}$~\cite{BKL83,Mat79}
and divide the total count by this value (since this preprocessing time
is dominated by our running time we will not mention it in the following).
Counting isomorphisms to non-induced subgraphs can be done
in the same time and space by changing the initialization on the
leaves, such that it checks for an subgraph instead of an
induced subgraph. Dividing again by the number of automorphisms gives
the number of subgraphs. By allowing the mapping of the
patterns to map several nodes to the same value, we can use them to
represent homomorphisms. Testing the leaves accordingly the same
algorithm can be used to count the number of homomorphisms from $H$ to
subgraphs of $G$. By keeping all tables in memory, thus sacrificing
the logarithmic space complexity, and using backtracking we can also
label every node with the number of times it appears as a certain vertex
of $H$.

From these observations and Lemma~\ref{lemma:number-patterns} we
arrive at the following theorem:

\begin{theorem}
  Given a graph $H$ on~$h$ vertices, a graph~$G$ on~$n$ vertices and a
  treedepth decomposition of~$G$ of height~$t$, one can compute the
  number of isomorphisms from $H$ to subgraphs of $G$,
  homomorphisms from $H$ to subgraphs of $G$,
  or (induced) subgraphs of~$G$ isomorphic to~$H$
  in time $O(6^{h} \cdot t^h \cdot h^2 \cdot n)$
  and space $O(3^{h} \cdot t^h \cdot ht \cdot \log n)$.
\end{theorem}

\noindent
Note that for graphs of unbounded treedepth the running time of the
algorithm degenerates to $O(6^{h} \cdot h^2 \cdot n^{h+1})$, which is
comparable to the running time of $2^{O(\sqrt{h\log h})} \cdot n^h$ of
the trivial counting algorithm.

\begin{theorem}
  Given a graph $H$ and a a graph $G$ belonging to a class of bounded
  expansion, there exists an algorithm to count the appearances of $H$
  as a subgraph of $G$ in time
  $$
  O\left({{f(h)}\choose{h}} \cdot 6^{h} \cdot h^{h+2} \cdot n \right)
  $$
  where $f$ is a function depending only on the graph class.
\end{theorem}

\noindent
This immediately extends to nowhere dense classes, which have $p$-treedepth-colorings
with at most~$n^\epsilon$ colors (for sufficiently large graphs) for any~$\epsilon > 0$.
Choosing the graphs large enough and setting~$\epsilon' = \epsilon/h$, we
can bound the term~${ f(h) \choose h}$ by $n^{\epsilon' \cdot h} = n^\epsilon$.

\begin{theorem}
  Let~$\cal G$ be a nowhere-dense class and let $H$ be a graph. For
  every~$\epsilon > 0$ there exists~$N_\epsilon \in \N$, such that for any
  graph $G \in \cal G, |G| > N_\epsilon$ there exists an algorithm to count
  the appearances of $H$ as a subgraph of $G$ in time
  $$
  O\left( 6^{h} \cdot h^{h+2} n^{1+\epsilon} \right).
  $$
\end{theorem}

\noindent
Finally, we would like to point out that this counting algorithm is
trivially parallelizable.

\subsection{Localized Centrality}

\noindent
Centrality is a notion used to ascribe the relative importance of a vertex in
the network. A centrality measure is a real-valued function that assigns each
vertex of the network some value with the understanding that higher values
correspond to more central vertices. Depending on the application,
``central'' vertices need not be those with high degree (for example, a
cut-vertex may have high centrality as it is the only way for information to flow
between two large subgraphs). There have been a wide variety of centrality
scores introduced in the literature, including degree centrality, closeness
centrality, eigenvector centrality, betweenness centrality and others. For a
comprehensive introduction to centrality measures in social networks see, for
instance, \cite{Fre77,Fre79}. There are several recent articles devoted to the
topic of centrality measures in general~\cite{PF08,KLP05}. In this section, we
consider localized variants of measures similar to the \emph{closeness
centrality} introduced by Sabidussi~\cite{Sab66} (see
Table~\ref{table:closeness}). The global versions of these measures
require one to compute the distance between all vertex pairs in the network, a
sub-routine where the fastest known algorithm (due to Brandes~\cite{Bra01}) is
$O(n(n+m))$ which in the context of sparse networks reduces to quadratic time.

In these localized variants, we compute the
measure of a vertex with respect to its $r^{\text{th}}$ neighborhood
rather than with respect to the whole graph.  We give linear time algorithms for computing
these measures on graphs of bounded expansion for every constant~$r$.
As the value of~$r$ increases, the value of
the measure computed for a vertex approaches its unlocalized variant
at the expense of an increase in running time.\footnote{For general
  graphs, we can compute these localized variants in time $O(n(n+m))$,
  by performing a breadth-first search from every vertex, for
  instance.  We do not know whether a better running time is
  possible.} The measures in question and their localized variants
are listed in Table~\ref{table:closeness}.

\begin{table}[t!]
  \small\centering
\begin{tabular}{@{}lll@{}}
  \toprule
  Measure & Definition & Localized \\ \midrule\noalign{\vskip .3em}
  Closeness~\cite{Sab66}  & $ c_C(v) = ( \sum_{u \in V(G)} \limits d(v,u) )^{-1} $
                        & $ c^r_C(v) = ( \sum_{u \in N^r(v)} \limits d(v,u) )^{-1} $
                        \\
  Harmonic~\cite{OAS10}
                        & $c_H(v) = \sum_{u \in V(G)} \limits  d(v,u)^{-1} $
                        & $c^r_H(v) = \sum_{u \in N^r(v)} \limits  d(v,u)^{-1} $
                        \\
  Lin's index~\cite{Lin76}  & $\displaystyle c_L(v) = \frac{|\{ u \mid d(v,u) < \infty \}|^2}{\sum_{u \in V(G): d(v,u) < \infty}d(v,u)}$
                        & $\displaystyle c^r_L(v) = \frac{|N^r[v]|^2}{\sum_{u \in N^r[v]}d(v,u)} $
                          \vspace*{.3em}\\
  \bottomrule
  \vspace{0pt} 
\end{tabular}
\caption{\label{table:closeness}\footnotesize Distance-based centrality measures with localized variants that can be computed in linear time on graphs of bounded expansion.}
\end{table}

One natural question is the utility of localized variants (and their accuracy
in reflecting the global measure). We remark that Marsden demonstrated that
for some networks, calculating the measure for a vertex~$v$ inside its closed
neighborhood $G[N[v]]$ can be used as a viable substitute for the full
measure~\cite{Mar02}. In the context of computer networks, Pantazopoulos et
al.~\cite{PKS13} consider local variants which lend themselves to distributed
computing and found a close correlation to the full measures on a sample of
networks. We can show experimentally that our variants reliably capture the
top ten percent in (arbitrarily) selected networks of our real-world corpus.
To that end, we compare the top 10 percent as identified by our localized
variants to those 10 percent identified by the respective full centrality
measure. Specifically, we use the Jaccard index~\cite{Jaccard1901}---defined
as $|A \cap B| / |A \cup B|$ for two sets~$A,B$---to measure
similarity\footnote{Since we compare sets of equal size, the measures precision
and recall---and accordingly the $F_1$-score---are all the same. The Jaccard
index is better suited for this situation.}.

The results in Figure~\ref{fig:centrality-plots} suggest that already
a value of~$r$ equal to about half the diameter yields very good results
across all three measures. Note that we do not compare rankings; but rather
only the difference between the sets of the identified top vertices---our
experiments showed that rank ordering is not preserved reliably. Furthermore,
there seems to be a slight positive tendency towards the localized version
being better in larger networks, though it is hard to draw any conclusions on
a small experiment like that.


While the localized Lin's index and the localized harmonic closeness work as
depicted in Table~\ref{table:closeness}, the localized closeness needs a small
normalization tweak in order to yield good results: this is achieved by treating
the ($r+1$)-st neighborhood of very vertex as if it would contain all remaining
vertices and adding this value accordingly (this obviously does not change the
running time of the algorithm).

\begin{figure}[t!]
  \centering
  \begin{minipage}{.5\textwidth}
    \centering
    \small \hspace{1.2cm} Closeness centrality\\
    \includegraphics[width=\textwidth]{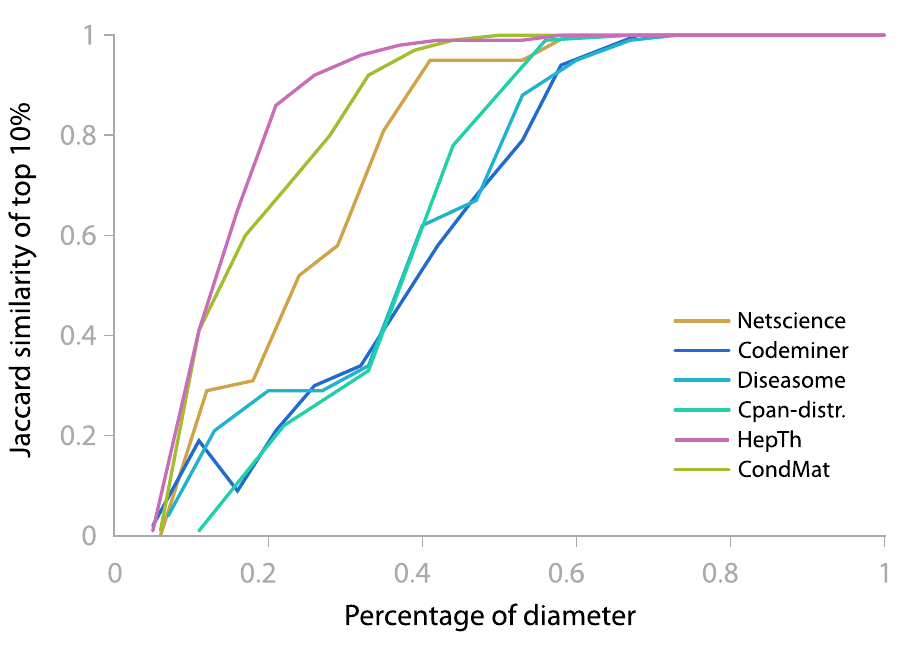}
  \end{minipage}%
  \begin{minipage}{.5\textwidth}
    \centering
    \small \hspace{1.2cm} Harmonic centrality\\
    \includegraphics[width=\textwidth]{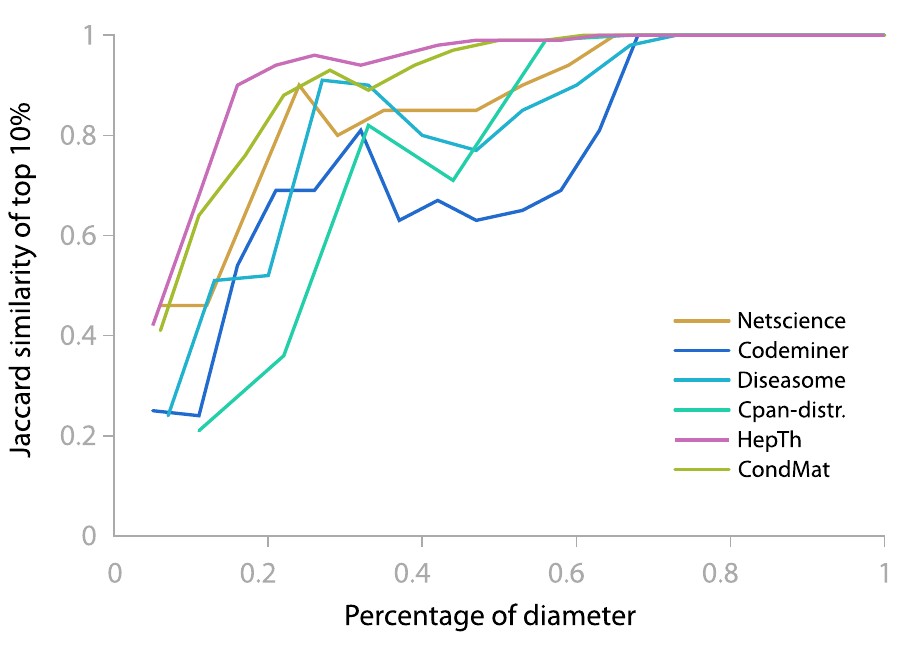}
  \end{minipage}

  \vspace*{5pt}
  \begin{minipage}{.5\textwidth}
    \centering
    \small \hspace{1.2cm} Lin's index\\
    \hspace*{-6pt}
    \includegraphics[width=\textwidth]{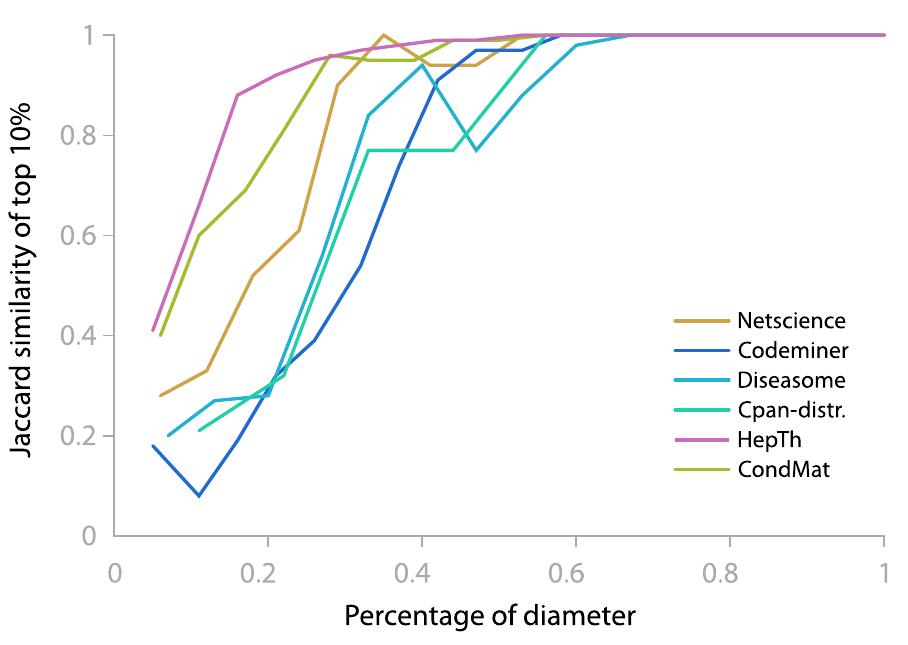}
  \end{minipage}
  \begin{minipage}{.49\textwidth}
    \vspace{-0cm}
    \hspace{1cm} \small
    \begin{tabular}{l>{\sffamily}r>{\sffamily}r}
      Network             & \multicolumn{1}{c}{size} & \multicolumn{1}{c}{diam.} \\
      \midrule
      Netscience          & 379       & 17            \\
      Codeminer           & 667       & 19            \\
      Diseasome           & 1419      & 15            \\
      Cpan-distr.         & 2719      & 9             \\
      HepTh              & 5835      & 19           \\
      CondMat            & 13861     & 18            \\ \cmidrule(lr){2-3}
      & \multicolumn{2}{c}{\footnotesize \emph{g.c. only}} \\
    \end{tabular}
  \end{minipage}
  \caption{\footnotesize Quality of localized centrality measures in terms of
           similarity of the top ten percent vertices against the full
           centrality measure. Measurements were taken only in the giant
           component of the networks displayed in the table.}
\label{fig:centrality-plots}
\end{figure}

We will rely on the following proposition to compute the localized centrality measures.

\begin{proposition}[Truncated distances~\cite{NOdM08a}]\label{theorem:trunc-dist}
  Let $G$ be a graph of bounded expansion. For every $r$ one can
  compute in linear time a directed graph $\augm G_r$ with in-degree
  bounded by $f(r)$ -- for some function $f$ -- on the same vertex set
  as $G$ and an arc labeling $\omega \colon \augm E(\augm G_r) \to \mathbf
  N$ such that for every pair $u,v \in G$ with $d_G(u,v) \leq r$
  one of the following holds:
    \begin{enumerate}[(i)]
      \item $uv \in \augm G_r$ and $\omega(uv) = d_G(u,v)$;
      \item $vu \in \augm G_r$ and $\omega(vu) = d_G(u,v)$;
      \item there exists $w \in N_{\augm G_r}^-(u) \cap N_{\augm G_r}^-(v)$
            such that $\omega(wu) + \omega(wv) = d_G(u,v)$.
    \end{enumerate}
\end{proposition}

Let $G$ be a graph from a class of bounded expansion and let $\augm G_r$ be the directed
graph with in-degree bounded by $f(r)$, for some function~$f$, that is obtained from~$G$
as by Proposition~\ref{theorem:trunc-dist}. In the following, we let
$N^-_r(v)$ denote the in-neighborhood of the vertex~$v$ in the directed graph~$\augm G_r$.
We assume that the vertices of $G$ are ordered so that every vertex set has a unique
representation as a tuple and, by slight abuse of notation, we use both representations interchangeably.
For $v \in V(G)$ and $A = (a_1,a_2,\dots,a_p) \subseteq N^-_r(v)$, define
the \emph{distance vector} from~$v$ to~$A$ as
$\dist(v,A) := (\omega(a_1 v), \omega(a_2 v), \dots, \omega(a_p v))$,
where $\omega$ denotes the arc-labeling from Proposition~\ref{theorem:trunc-dist}. Since
$a_i \in N^{-}_r(v)$, $\omega(a_i, v) = d_G(a_i, v)$.

\begin{definition}\label{def:in-intersection}
  Let $v \in \augm G_r$, $\emptyset \neq X \subseteq N^-_r(v)$, $\alpha: V(G) \to \mathbf R$ a vertex
  weighting and let $\bar d \in [r]^{|X|}$ be a distance vector.
  We define
  $$
    N(v,X,\bar d) := \{ v \neq u \in V(G) \mid N^-_r(v) \cap N^-_r(u) = X ~\text{and}~ \dist(u,X) = \bar d\}
  $$
  as those vertices whose in-neighborhood in $\augm G_r$ overlap with the in-neighborhood of $v$ in
  exactly $X$ and whose distance-vector to $X$ is exactly $\bar d$. Then the
  \emph{query-function} $c_\alpha$ is defined as
  $$
    c_\alpha(v,X,\bar d) := \sum_{u \in N(v,X,\bar d)} \alpha(u).
  $$
 \end{definition}

\begin{lemma}\label{lemma:query-constant-time}
  Given $\augm G_r$, one can compute a data structure in time $O(n)$ such that
  queries $c_\alpha(v,X,\bar d)$ as in Definition~\ref{def:in-intersection} can be answered in constant time.
\end{lemma}

\begin{proof}
  We define an auxiliary dictionary $R$ indexed by vertex sets $X \subseteq N^-r(v)$, for some vertex $v$.
  At each entry $v \in \augm G_r$, we will store another dictionary indexed by distance vectors
  which in turn stores a simple counter.
  We initialize $R$ as follows: for every $v \in \augm G_r, X \subseteq N^-_r(v)$ and
  every distance vector $\bar d \in [r]^{|X|}$, set $R[X][\bar d] = 0$.
  Note that in total, $R$ contains only $O(n)$ entries since all in-neighborhoods in $\augm G_r$ have constant size.
  We can implement $R$ as a hash-map to achieve the desired (expected) constant-time for insertion
  and look-up, though this would yield a randomized algorithm. A possible way to implement $R$ on a RAM deterministically
  is the following: We store the key $X = \{x_1,x_2,\dots,x_p\}$ at address $x_1 + n\cdot x_2 + \dots + n^p \cdot x_p$.
  This uses addresses up to size $n^c$, for some constant $c$, but since we only insert $O(n)$ keys the
  setup takes only linear time. Our later queries to $R$ will be restricted to keys that are guaranteed to be contained
  in the dictionary, thus we will never visit a register that has not been initialized.

  For every $v \in \augm G_r$ and $X \subseteq N^-_r(v)$, increment the counter $R[X][\dist(v,X)]$ by $\alpha(v)$.
  We now claim that queries of the form $c_\alpha(v,X,\bar d)$ can be computed using inclusion-exclusion as follows:
  $$
    c_\alpha(v,X,\bar d) = \sum_{X \subseteq Y \subseteq N^-_r(v)} (-1)^{|Y\setminus X|} \sum_{\bar d': \bar d'|_X = \bar d} R[Y][\bar d'].
  $$
  The computation of the sum clearly takes constant time\footnote{We tacitly assume that the weights $\alpha$ only assign numbers
  polynomially bounded by the size of the graph.}. We now prove that it indeed computes the quantity $c_\alpha(v,X,\bar d)$.

  Consider a vertex $u \in \augm G_r$, such that $N^-_r(u) \cap N^-_r(v) = X$ and $\dist(u,X) = \bar d$. We argue that
  $\alpha(u)$ is counted once by the above sum: $\alpha(u)$ is not counted by any $R[Y][\cdot]$ with
  $Y \supsetneq X$, therefore only the term where $X = Y$ counts $\alpha(u)$ and does so exactly once.
  It remains to be shown that the weight of vertices that do not conform with Definition~\ref{def:in-intersection} are
  either not counted by the sum or cancel out.

  Consider a vertex $w \in \augm G_r$ with such that $\dist(w,X) \neq \bar d$. The weight of such a vertex is not counted by
  the above sum, since $\alpha(w)$ is only counted in entries of $R$ that do not occur as summands.

  Finally, consider a vertex $w' \in \augm G_r$ with $N^-_r(w') \cap N^-_r(v) = Z$ where $X \subsetneq Z \subseteq N^-_r(v)$
  and such that $\dist(w',X) = \bar d$. The weight of this vertex is counted in each term of
  $$
    \sum_{X \subseteq Y \subseteq Z} (-1)^{|Y\setminus X|} R[Y][\dist(w',Z)|_Y]
  $$
  since
  $$
    \sum_{X \subseteq Y \subseteq Z} (-1)^{|Y\setminus X|} = \sum_{0 \leq k \leq |Z\setminus X|} (-1)^{k} {|Z\setminus X| \choose k} = 0
  $$
  we know that the signs cancel out and thus $\alpha(w')$ does not contribute to $c_\alpha(v,X,\bar d)$. Hence
  the above sum computes exactly the query $c_\alpha(v,X,\bar d)$.
\end{proof}

\begin{theorem}
  Let $\mc G$ be a graph class of bounded expansion, $G \in \mc G$
  a graph and $r \in \mathbf N$ an integer. Then one can compute the
  quantities $\alpha_d(v) = \sum_{w \in N^d(v)} \alpha(w)$ for all
  $v \in G, d \leq r$ in linear time.
\end{theorem}

\begin{proof}
  By Theorem~\ref{theorem:trunc-dist} we can compute $\augm G_r$ in
  linear time, thus we can employ Lemma~\ref{lemma:query-constant-time}
  to answer queries as defined in Definition~\ref{def:in-intersection}
  in constant time. To compute the quantity $\alpha_d(v)$ for all
  $0 < d \leq r$ and $v \in V(G)$, we proceed as follows.
  Initialize an array $C$ by setting $C[v][d] = 0$
  for every $v \in G, 0 < d \leq r$.

  Now for every $v \in V(G)$, every $X \subseteq N^-_r(v)$ and every
  distance vector $\bar d \in [r]^{|X|}$, update $C$ via
  \begin{align*}
    C[v][\min(\bar d + \dist(v,X))] & \leftarrow C[v][\min(\bar d + \dist(v,X))] + c_\alpha(v,X,\bar d) \\
  \intertext{and then apply the correction}
    C[v][\min(\dist(v,X)+\dist(v,X))] & \leftarrow C[v][\min(\dist(v,X)+\dist(v,X))] - 1
  \end{align*}
  in both cases with the convention that we dismiss entries where $\min(\bar d + \dist(v,X)) > r$.
  The second case corrects the query $c_\alpha(v,N^-(v),\dist(v,N^-(v)))$
  counting the vertex~$v$ itself.

  At this point, $C[v][d]$ contains the sum of weights of vertices $u$ for which
  $\min(\dist(v,X)+\dist(u,X)) = d$ where $X = N^-_r(v) \cap N^-_r(u) \neq \emptyset$.
  This follows directly from the definition of $c_\alpha$.

  By Theorem~\ref{theorem:trunc-dist}, every pair of vertices of distance $< r$ in $G$
  either is connected by an arc or they share a common in-neighbor in $\augm G_r$. Accordingly,
  we update the values of $C$ as follows: for every $uv \in \augm E(\augm G_r)$
  \begin{itemize}
     \item if $N^-_r(u) \cap N^-_r(v) = \emptyset$, the weights of the vertices $u$ and $v$ were not
           counted in $C[v][\cdot],C[u][\cdot]$ respectively, thus we update $C$ via
        \begin{align*}
            C[v][\omega(uv)] & \leftarrow C[v][\omega(uv)] + \alpha(u) \\
            C[u][\omega(uv)] & \leftarrow C[u][\omega(uv)] + \alpha(v)
        \end{align*}
     \item if $X = N^-_r(u) \cap N^-_r(v) \neq \emptyset$, the weights of the vertices $u$ and $v$ were
           counted in $C[v][d']$ and $C[u][d']$ for $d' = \min(\dist(u,X)+\dist(v,X))$,
           respectively. Since $d'$ might be larger than $\omega(uv)$ (but cannot be
           smaller), we update
           $C$ via
        \begin{align*}
            C[v][d'] & \leftarrow C[v][d'] - \alpha(u) \\
            C[u][d'] & \leftarrow C[u][d'] - \alpha(v) \\
            C[v][\omega(uv)] & \leftarrow C[v][\omega(uv)] + \alpha(u) \\
            C[u][\omega(uv)] & \leftarrow C[u][\omega(uv)] + \alpha(v)
        \end{align*}
          where we again ignore the update of $C[\cdot][d']$ if $d' > r$.
  \end{itemize}
  Note that this procedure
  is problematic if both $uv$ and $vu$ are present in the graph, since then the
  this correction would (wrongly) be applied twice. The simple solution is that in the
  case of both arcs being present we only apply the above update for that arc where
  the start vertex is smaller than the end vertex, \ie to $uv$ if $u < v$ and $vu$ otherwise.

  \noindent At this point, $C[v][d]$ contains the sum of weights of vertices $u$ for which either
  \begin{itemize}
    \item the value $d = \min(\dist(v,X)+\dist(u,X))$ where $X = N^-_r(v) \cap N^-_r(u) \neq \emptyset$
          and $uv \not \in \augm E(\augm G_r)$,
    \item or $d = \omega(uv)$ and $uv \in \augm E(\augm G_r)$.
  \end{itemize}
  Thus by Theorem~\ref{theorem:trunc-dist} we have that $C[v][d] = \alpha_d(v)$ for $d < r$ and $v \in G$.
  Since all of the above operations take time linear in $|V(G)|$, the claim follows.
\end{proof}

\noindent
If we take $\alpha(\cdot) = 1$, the above algorithm counts exactly the
sizes of the $d^{\text{th}}$ neighborhoods of each vertex, for $d <
r$. Thus it can be used to compute the $r$-centric centrality measures
presented in Table~\ref{table:closeness}.

\begin{corollary}
  Let $\mc G$ be a graph class of bounded expansion, $G \in \mc G$ a graph and $r \in \mathbf N$ an integer.
  Then the $r$-centric closeness, harmonic centrality and  Lin's index
  can be computed for all vertices of $G$ in total time $O(|V(G)|)$.
\end{corollary}

%
%
%
\section{Conclusion and Open Problems}\label{sec:Conclusion}
\noindent
We propose unifying structural graph algorithms with complex network
analysis by searching for observable structural properties that
satisfy the litmus test of enabling efficient algorithms for network
analysis.  We presented theoretical and empirical results that support
our hypothesis that complex networks are structurally sparse in a
well-defined and robust sense. Efficient algorithms are known for
networks of bounded expansion \cite{NOdM12,DvorakKT13,GroheKS14}, and
we show that for key network analysis problems these algorithms can be
even further improved.  On the theoretical side, we show that several
random graph models of complex networks exhibit bounded expansion with
high probability, although not all do---suggesting an interesting
dichotomy of networks.  On the experimental side, we confirm these
mathematical results, and show that many real-world complex networks
additionally appear to exhibit bounded expansion as measured using
specialized colorings.  This new approach enables fast algorithms to
analyze features including communities, centrality, and motifs while
more broadly providing a rigorous framework for a deeper understanding
of real-world networks and related models.

There are a plethora of random graph models specifically designed to mimic
properties of complex networks. Which of these models exhibit structural
sparsity (and which ones do not)? There is also room for debate about
how to establish that a model will generate graphs with certain properties
\emph{in practice}. Asymptotic behavior is only a proxy, although
we took care to provide details on the speed of convergence in our
proofs where possible. As exemplified by the relatively weak result about
the \Barabasi-Albert model, the practical implications are sometimes difficult
to judge.

On the algorithmic side, there are several key challenges
remaining. Does there exist a better algorithm/heuristic to obtain
$p$-treedepth colorings, in particular taking into consideration the
special structure of complex networks? Does a good coloring algorithm
exist that provides a trade-off between the number of colors and the
treedepth of subgraphs induced by few color classes? Can we compute or
approximate lower bounds for either $\chi_p$ or $\topgrad_r$ with
reasonable margins of error? Both would likely improve our current
empirical understanding of the grad of networks. We would also like to
investigate whether the grad for small depths is a reliable measure to
differentiate networks; both our empirical and theoretical results
seem to indicate so.

Finally, algorithms exploiting low grad should be tested extensively via
computational experiments, to ascertain the feasibility of applying these
techniques to real-world networks.

\bigskip
\footnotesize
\noindent\textbf{Acknowledgements.}

\noindent E. Demaine supported in part by NSF grant CCF-1161626 and
DARPA GRAPHS/AFOSR grant FA9550-12-1-0423. \\[1ex]
Blair D. Sullivan
supported in part by DARPA GRAPHS/SPAWAR Grant
N66001-14-1-4063, the Gordon \& Betty Moore Foundation's Data-Driven Discovery Initiative through Grant GBMF4560, and the National Consortium for Data Science. \\[1ex]
Peter Rossmanith, Fernando S{\'a}nchez Villaamil
supported by DFG Project RO 927/13-1 ``Pragmatic Parameterized Algorithms.''. \\[1.5ex]
Any opinions, findings, and conclusions or recommendations expressed in this
publication are those of the author(s) and do not necessarily reflect the
views of DARPA, SSC Pacific, AFOSR, the Moore Foundation, or the NCDS.

%
%
%

\def\redefineme{
    \def\LNCS{LNCS}%
    \def\ICALP##1{Proc. of ##1 ICALP}%
    \def\FOCS##1{Proc. of ##1 FOCS}%
    \def\COCOON##1{Proc. of ##1 COCOON}%
    \def\SODA##1{Proc. of ##1 SODA}%
    \def\SWAT##1{Proc. of ##1 SWAT}%
    \def\IWPEC##1{Proc. of ##1 IWPEC}%
    \def\IWOCA##1{Proc. of ##1 IWOCA}%
    \def\ISAAC##1{Proc. of ##1 ISAAC}%
    \def\STACS##1{Proc. of ##1 STACS}%
    \def\IWOCA##1{Proc. of ##1 IWOCA}%
    \def\ESA##1{Proc. of ##1 ESA}%
    \def\WG##1{Proc. of ##1 WG}%
    \def\LIPIcs##1{LIPIcs}%
    \def\LIPIcs{LIPIcs}%
    \def\LICS##1{Proc. of ##1 LICS}%
}

\bibliographystyle{abbrv}
\footnotesize
\bibliography{./biblio}

\end{document}